\documentclass[12pt]{article}
\usepackage[authoryear]{natbib}
\usepackage{amssymb}
\usepackage{amsfonts}
\usepackage{amsmath}
\usepackage{dsfont}
\usepackage{soul}
\usepackage[nohead]{geometry}
\usepackage{graphicx}
\usepackage{amsthm}
\usepackage{color}
\usepackage{comment}
\usepackage{setspace}
\usepackage{framed}
\usepackage{enumitem}
\usepackage{todonotes}
\usepackage{tikz}
\usepackage{rotating}
\usepackage{subfigure}
\usepackage[flushleft]{threeparttable}
\usepackage{multirow}
\usepackage[colorlinks=true,citecolor=blue,linkcolor=blue]{hyperref}
\usepackage{xr}
\usepackage{bm}
\usepackage{xcolor}
\usepackage{booktabs}
\usepackage{lscape}
\usepackage{algorithm}
\usepackage{algorithmic}
\usepackage[toc,title,page]{appendix}
\usepackage{lscape}
\usepackage{pdflscape}

\usepackage{natbib}

\renewcommand{\today}{\ifcase \month \or January\or February\or March\or %
	April\or May \or June\or July\or August\or September\or October\or November\or %
	December\fi, \number \year} 

\usepackage{caption} 
\allowdisplaybreaks

7

\newcommand{\bA}{\mbox{\bf A}}

\newcommand{\bB}{\mbox{\bf B}}

\newcommand{\bD}{\mbox{\bf D}}
\newcommand{\bE}{\mbox{\bf E}}
\newcommand{\bF}{\mbox{\bf F}}
\newcommand{\bG}{\mbox{\bf G}}
\newcommand{\bH}{\mbox{\bf H}}

\newcommand{\bL}{\mbox{\bf L}}

\newcommand{\bW}{\mbox{\bf W}}
\newcommand{\bI}{\mbox{\bf I}}
\newcommand{\bJ}{\mbox{\bf J}}
\newcommand{\bV}{\mbox{\bf V}}

\newcommand{\bU}{\mbox{\bf U}}

\newcommand{\bY}{\mbox{\bf Y}}

\newcommand{\bGamma}{\mbox{\boldmath $\Gamma$}}
\newcommand{\bLambda}{\mbox{\boldmath $\Lambda$}}
\newcommand{\bTheta}{\mbox{\boldmath $\Theta$}}

\newcommand{\bSigma}{\mbox{\boldmath $\Sigma$}}
\newcommand{\bOmega}{\mbox{\boldmath $\Omega$}}

\newcommand{\bPhi}{\mbox{\boldmath $\Phi$}}
\newcommand{\bPsi}{\mbox{\boldmath $\Psi$}}

\newcommand{\cov}{\mathrm{cov}}

\newcommand{\tr}{\mathrm{tr}}
\newcommand{\diag}{\mathrm{diag}}

\newcommand{\beq}{\begin{eqnarray*}}
\newcommand{\eeq}{\end{eqnarray*}}
\newcommand{\er}{\mathrm{ER}}

\usepackage{graphicx,rotating,booktabs,threeparttable}

\usepackage{adjustbox}
\usepackage{array}

\newcolumntype{R}[2]{%
	>{\adjustbox{angle=#1,lap=\width-(#2)}\bgroup}%
	l%
	<{\egroup}%
}

\newtheorem{thm}{Theorem}[section]

\newtheorem{lem}{Lemma}[section]

\newtheorem{assum}{Assumption}[section]

\numberwithin{equation}{section}
\theoremstyle{definition}

\newtheorem{remark}{Remark}[section]
\makeatletter
\def\@biblabel#1{\hspace*{-\labelsep}}
\makeatother
\begin{document}
	\bibliographystyle{ecta}
	
	\title{Large Volatility Matrix Analysis Using Global and National Factor Models}

	\date{\today}

	\author{
		Sung Hoon Choi\thanks{%
			Department of Economics, University of Connecticut, Storrs, CT 06269, USA. 
			E-mail: \texttt{sung\_hoon.choi@uconn.edu}.} \\ 
		\and Donggyu Kim\thanks{Corresponding author. College of Business,
			Korea Advanced Institute of Science and Technology (KAIST), Seoul, Republic of Korea.
			Email: \texttt{donggyukim@kaist.ac.kr}.}\\ 
		}
	\maketitle
	\pagenumbering{arabic}	
	\begin{abstract}
		\onehalfspacing
		Several large volatility matrix inference procedures have been developed, based on the latent factor model. 
		They often assumed that there are a few of common factors, which can account for volatility dynamics.
		However, several studies have demonstrated the presence of local factors.
		In particular, when analyzing the global stock market, we often observe that nation-specific factors explain their own country's volatility dynamics. 
		To account for this, we propose the Double Principal Orthogonal complEment Thresholding (Double-POET) method, based on multi-level factor models, and also establish its asymptotic properties.
		Furthermore, we demonstrate the drawback of using the regular principal orthogonal component thresholding (POET)  when the local factor structure exists. 
		We also describe the blessing of dimensionality using Double-POET for local covariance matrix estimation.
		Finally, we investigate the performance of the Double-POET estimator in an out-of-sample portfolio allocation study using international stocks from 20 financial markets.\\

					
	\end{abstract}

	\newpage
	
	\doublespacing
	\section{Introduction}
High dimensional factor analysis and principal component analysis (PCA), which are powerful tools for dimension reduction, have been extensively studied \citep{bai2003inferential, bernanke2005measuring, stock2002forecasting}. 
They have various applications in economics and finance, such as forecasting macroeconomic variables and optimal portfolio allocation. 
	Recently, a multi-level factor structure with global and local factors has received increasing attention \citep{ando2016panel, bai2015identification, choi2018multilevel, han2021shrinkage}. 
	The global factors   impact on all individuals, while the local factors only impact on those in the specific group. 
	The local group can be defined by regional, country, or industry level. 
	In many economic and financial applications,  the local or country factors naturally exist. 
	For example, \cite{kose2003international} characterized the comovement of international business cycles on the global, regional, and country levels by imposing a multi-level factor structure on a dynamic factor model. 
	\cite{moench2013dynamic} showed that the local factors play an important role in explaining the U.S. real activities. See also \cite{ando2015asset, ando2017clustering, aruoba2011globalization, gregory1999common, hallin2011dynamic} for related articles. 
	Hence, when the local factor structure is ignored, the conventional factor analysis might yield misleading results.

	Based on the factor models, several large volatility matrix estimation procedures have been developed to account for the strong cross-sectional correlation in the stock market. 
	For example, \cite{fan2008high} studied the impacts of covariance matrix estimation on optimal portfolio allocation and portfolio risk assessment when the factors are observable.  
	In contrast,  \cite{fan2013large} considered latent factor models and developed the covariance matrix estimation procedure by PCA and thresholding, which is called the principal orthogonal component thresholding (POET). 
	In this procedure, the unobservable factors can be consistently estimated with a large number of assets. 
	Recent studies, such as \cite{ait2017using, fan2016incorporating, fan2018large, fan2019structured,   jung2022next, wang2017asymptotics}, have also been conducted along this direction.  
	However, when considering the global market, we often observe not only the global common factors but also the nation-specific risk factors \citep{kose2003international}.
	That is, the single-level factor model cannot sufficiently explain volatility dynamics.

This paper proposes a novel large volatility matrix estimation procedure that incorporates a global and national  factor structure as well as a sparse idiosyncratic volatility matrix.
Specifically, we consider the global asset market, and, to account for the nation-specific risk factors, we assume that there are latent multi-level factors, such as the global common factors and national risk factors. 
Since the national risk is a regional risk, we further assume that the local factor membership is known. 
Under this latent multi-level factor models, we first use the PCA procedure to capture the latent global factors.
Then, after removing the latent global factors, we apply the PCA procedure in each national group to accommodate the latent national factors. 
Finally, to account for the sparse idiosyncratic volatility matrix, we adopt an adaptive thresholding scheme, which we call this the  Double Principal Orthogonal complEment Thresholding (Double-POET).
We then derive convergence rates for Double-POET and its inverse under different matrix norms. 
When the local factor membership is unknown, we suggest the regularized spectral clustering method \citep{amini2013pseudo} to detect the latent local structure.
We further discuss the benefit of the proposed Double-POET compared to the regular POET procedure. 
For example, if we ignore the local factors and treat them as idiosyncratic and employ the POET procedure, the POET estimator can be inconsistent.
When estimating the local volatility matrix for each country, Double-POET can estimate global factors better than the regular POET estimator. 
That is, we can find the blessing of dimensionality.
The empirical study supports the theoretical findings.

The rest of the paper is organized as follows. 
Section \ref{section2} sets up the model and proposes the Double-POET estimation procedure. 
Section \ref{asymp} presents an asymptotic analysis of the Double-POET estimators. 
The merits of the proposed method are illustrated by a simulation study in Section \ref{simulation} and by real data application on portfolio allocation in Section \ref{empiric}.
 In Section \ref{conclusion}, we conclude  the study. 
 All proofs are presented in Appendix \ref{proofs}.

\section{Model Setup and Estimation Procedure} \label{section2}

		 Throughout this paper, let $\lambda_{\min}(\bA)$ and $\lambda_{\max}(\bA)$ denote the minimum and maximum eigenvalues of a matrix $\bA$, respectively. 
		 In addition, we denote by $\|\bA\|_{F}$, $\|\bA\|_{2}$ (or $\|\bA\|$ for short), $\|\bA\|_{\infty}$, and $\|\bA\|_{\max}$ the Frobenius norm,  operator norm, $l_{\infty}$-norm and elementwise norm, which are defined respectively as $\|\bA\|_{F} = \tr^{1/2}(\bA'\bA)$, $\|\bA\|_{2} = \lambda_{\max}^{1/2}(\bA'\bA)$, $\|\bA\|_{\infty} = \max_{i}\sum_{j}|a_{ij}|$, and $\|\bA\|_{\max} = \max_{i,j}|a_{ij}|$. 
		 When $\bA$ is a vector, the maximum norm is denoted as $\|\bA\|_{\infty}=\max_{i}|a_{i}|$, and both $\|\bA\|$ and $\|\bA\|_{F}$ are equal to the Euclidean norm. We denote $\diag(\bA_{1},\ldots, \bA_{n})$ with the diagonal block entries as $\bA_{1},\ldots, \bA_{n}$.

\subsection{Multi-Level Factor Model}
We consider the following multi-level factor model:
\begin{equation}\label{model}
	y_{it} = b_{i}'G_{t} + {\lambda_{i}^{g_{i}}}'f_{t}^{g_{i}}+ u_{it}, 
\end{equation}
where $y_{it}$ is the observed data for the $i$th individual  at time $t$, for $i = 1,\ldots, p$, and $t = 1,\ldots,T$; $G_{t}$ is a $k \times 1$ vector of unobservable ``global" common factors, $b_{i}$ indicates the corresponding factor loadings; $f_{t}^{g_{i}}$ is an $r_{g_i} \times 1$ vector of unobservable ``local" factors that only affect  the group $g_{i}$, $\lambda_{i}^{g_{i}}$ indicates the corresponding factor loadings in each group; and  $u_{it}$ is an idiosyncratic error component.  
	We denote the cluster or group membership as $g_{i} \in \{1,\dots, J\}$. 
	Throughout the paper, we assume that the group membership $\{g_{i}\}_{i=1}^{p}$ is known and the global and local factors are latent.
	In this paper, we consider the nation-specific local factors; thus, the group membership is the country. 
	 In addition, the numbers of global factors and the number of local factors in each group are fixed.

	Given the group membership, for $g_i = j$, we can stack the observations and denote them as $y_{t} \equiv ({y_{t}^{1}}',\dots, {y_{t}^{J}}')'$, where $y_{t}^j = (y_{1t}, \dots, y_{p_jt})'$ and the number of individuals $p_j$ within group $j$ such that $p = \sum_{j=1}^{J}p_j$.
	 Let $F_{t} = {(f_{t}^{1}}', \dots, {f_{t}^{J}}')'$, where $f_{t}^{j}$ is an $r_j \times 1$ vector of local factors and the number of factors $r_j$ for each group $j$ such that $r = \sum_{j=1}^{J} r_j$. We define the $p \times r$ block diagonal matrix as 
	$$\bLambda = \diag(\Lambda^1, \Lambda^2, \dots, \Lambda^J),$$ 
	where $\Lambda^{j} = (\lambda_{1}^{j}, \dots, \lambda_{p_{j}}^{j})'$ is a $p_j \times r_j$ matrix of local factor loadings for each $j$. 
	Then, the model (\ref{model}) can be written in vector form as follows:
	\begin{equation*}
	y_{t} = \bB G_{t} + \bLambda F_{t} +u_{t},
	\end{equation*}
	where $\bB = (b_1, \dots, b_p)'$ and $u_{t} = (u_{1t}, \dots, u_{pt})'$.
	 In matrix notation, we have
	\begin{equation}\label{modelUU}
		\bY = \bG\bB' + \bF\bLambda' + \bU,
	\end{equation}
	where $\bY = (y_{1},\dots,y_{T})'$ is a $T \times p$ matrix of observed data, $\bG = (G_{1},\dots, G_{T})'$ is a $T \times k$ matrix of global factors, $\bF = (F_{1}, \dots, F_{T})'$ is a $T \times r$ matrix of local factors, and $\bU = (u_{1}, \dots, u_{T})'$ is a $T\times p$ matrix of idiosyncratic errors that are uncorrelated with global and local factors. Throughout the paper, we further assume that global and local factors are orthogonal each other.

	In this paper, the parameter of  interest is the $p \times p$ covariance matrix, $\bSigma$,  of $y_{t}$ as well as its inverse. 
	Given model (\ref{model}), the covariance matrix can be written as
	\begin{equation}\label{model2}
	\bSigma = \bB \cov(G_{t})\bB' + \bLambda \cov(F_{t})\bLambda' + \bSigma_{u},
	\end{equation}
	where $\bSigma_{u}$ is a sparse idiosyncratic covariance matrix.
We can demonstrate the multi-level factor analysis in the presence of spiked eigenvalues at different levels. 
Specifically,  decomposition (\ref{model2}) can be written as
	\begin{equation}\label{decomposition}
		\bSigma = \bB \cov(G_{t})\bB' + \bSigma_{E}, 
	\end{equation}
 where  $\bSigma_{E} = \bLambda \cov(F_{t})\bLambda' + \bSigma_{u}.$
Decomposition \eqref{decomposition} is a usual single-level factor-based covariance matrix. 
 Thus, based on \eqref{decomposition}, we can apply the regular POET procedure. 
 However, the eigenvalue of the covariance matrix $\bSigma_{E}$ diverges, which makes the POET procedure inefficient. 
 In Section \ref{asymp}, we discuss this inefficiency of the regular POET procedure. 
To further accommodate the local factor structure, we  decompose the covariance matrix $\bSigma_{E}$ as follows:
	\begin{equation}\label{local-decomposition}
	\bSigma_{E}^{j} = \Lambda^{j}\cov(f_{t}^{j}){\Lambda^{j}}' + \bSigma_{u}^{j}, \;\;\;\; \text{ for } j = 1, \dots, J,
	\end{equation}
	where the $j$th group covariance matrix  $\bSigma_{E}^{j}$ is a $p_j \times p_j$ diagonal block of $\bSigma_{E}$.
	We assume that the idiosyncratic covariance matrix $\bSigma_{u}=(\sigma_{u,ij})_{p\times p}$ is sparse as follows:
	\begin{align} \label{sparsity}
	m_{p} = \max_{i\leq p}\sum_{j\leq p} |\sigma_{u,ij}|^{q}, 
	\end{align}
for some $q \in [0,1)$, where $m_p$ diverges slowly, such as $\log p$. 
Intuitively, after removing the global and local factor components, most of  pairs are weakly correlated \citep{bickel2008covariance, cai2011adaptive}. 
Theoretically, since $\|\bSigma_{u}\| \leq \|\bSigma_{u}\|_{1} = O(m_{p})$, when $m_{p} = o(p_{j})$ for all $j \leq J$, there are distinguished eigenvalues between the local factor components and the idiosyncratic error components.
In light of this, in many applications, the sparsity condition on the factor model residuals has been considered \citep{boivin2006more, fan2016incorporating}.
Thus,  $\bSigma_{E}^{j}$ is the low-rank plus sparse structure, which has been widely used \citep{ait2017using, cai2013sparse, candes2009exact, fan2019structured, fan2008high, fan2013large, johnstone2009consistency, ma2013sparse, negahban2011estimation}.
When considering the U.S. market only, we have the usual single-level factor-based form. 
Then, based on \eqref{local-decomposition}, we can apply the regular POET procedure to estimate the local covariance matrix. 
However, this approach does not use assets outside of the local group, which causes inefficiency. 
We also discuss this inefficiency in Section  \ref{asymp}.
To accommodate the multi-level factor structure, we propose a novel large covariance matrix estimation procedure in the following section.

\subsection{Double-POET Procedure}  \label{estimation procedure}

 To make  model (\ref{model}) identifiable, without loss of generality, we impose normalization conditions: $\cov(G_{t}) = \bI_{k}$ and  $\cov(f_{t}^{j}) = \bI_{r_j}$, where $G_{t}$ and $f_{t}^{j}$ are uncorrelated with each other; both $\bB'\bB$ and ${\Lambda^{j}}'\Lambda^{j}$  for $j \in \{1, \dots, J\}$ are diagonal matrices. 
 We also assume the pervasiveness conditions, such that (i) the eigenvalues of $p^{-a_{1}}\bB'\bB$ are strictly greater than zero, and (ii) the eigenvalues of $p_{j}^{-a_{2}}{\Lambda^{j}}'\Lambda^{j}$ are strictly greater than zero for each $j$, where $a_{1}, a_{2} \in (0,1]$ are the strengths of global and local factors, respectively.
Then, the first $k$ eigenvalues of $\bB \cov(G_{t})\bB'$ diverge at rate $O(p^{a_{1}})$, while the first $r$ eigenvalues of $ \bLambda \cov(F_{t})\bLambda'$ diverge at rate $O(p^{ca_{2}})$, which does not grow too fast as $p^{a_1} \rightarrow \infty$. 
We note that $p^c$ is related to the number of stocks in each country.
In addition, all eigenvalues of $\bSigma_u$ are bounded.	
Let $\bGamma = \diag(\delta_1,\dots, \delta_k)$ be the leading eigenvalues of $\bSigma$ and $\bV = (v_1,\dots, v_k)$ be their corresponding leading eigenvectors.

To accommodate  the multi-level factor model \eqref{model}, we propose a non-parametric estimator of $\bSigma$ as follows:
\begin{enumerate}

\item Let $\widehat{\delta}_{1} \geq \widehat{\delta}_{2} \geq \dots \geq \widehat{\delta}_{p}$ be the ordered eigenvalues of the sample covariance matrix $\widehat{\bSigma} = T^{-1}\sum_{t=1}^{T}(y_{t}-\bar{y})(y_{t}-\bar{y})'$ and $\{\widehat{v}_{i}\}_{i=1}^{p}$ be their corresponding eigenvectors. Then, we compute 
		$$
		\widehat{\bSigma}_{E} = \widehat{\bSigma} - \widehat{\bV}\widehat{\bGamma}\widehat{\bV}',
		$$
		where $\widehat{\bGamma} = \diag(\widehat{\delta}_{1}, \dots, \widehat{\delta}_{k})$ and $\widehat{\bV}=(\widehat{v}_{1},\dots, \widehat{v}_k)$.

\item Define $\widehat{\bSigma}_{E}^{j}$ as each $p_j \times p_j$ diagonal block of $\widehat{\bSigma}_{E}$. 
For the $j$th block, let $\{\widehat{\kappa}_{i}^{j},\widehat{\eta}_{i}^{j}\}_{i=1}^{p_j}$ be the eigenvalues and eigenvectors of $\widehat{\bSigma}_{E}^{j}$ in decreasing order. 
Then, we compute the principal orthogonal complement $\widehat{\bSigma}_{u}$ as follows:		
		$$
		\widehat{\bSigma}_{u} =  \widehat{\bSigma} - \widehat{\bV}\widehat{\bGamma}\widehat{\bV}' - \widehat{\bPhi}\widehat{\bPsi}\widehat{\bPhi}',
		$$
		where $\widehat{\bPsi} = \diag(\widehat{\Psi}^{1}, \dots, \widehat{\Psi}^{J})$ for $\widehat{\Psi}^{j}=\diag(\widehat{\kappa}_{1}^{j}, \dots, \widehat{\kappa}_{r_j}^{j})$, and the block diagonal matrix $\widehat{\bPhi} = \diag(\widehat{\Phi}^{1},\dots, \widehat{\Phi}^{J})$ for $\widehat{\Phi}^{j} = (\widehat{\eta}_1^{j},\dots,\widehat{\eta}_{r_{j}}^{j})$ for $j = 1,2, \dots, J$.

\item Apply the adaptive thresholding method on $\widehat{\bSigma}_{u}=(\widehat{\sigma}_{u,ij})_{p\times p}$ following \cite{bickel2008covariance} and \cite{fan2013large}.
 Specifically, define $\widehat{\bSigma}_{u}^{\mathcal{D}}$ as  the thresholded error covariance matrix estimator:
		\begin{equation*}
			\widehat{\bSigma}_{u}^{\mathcal{D}} = (\widehat{\sigma}_{u,ij}^{\mathcal{D}})_{p\times p}, \;\;\;\;\; \widehat{\sigma}_{u,ij}^{\mathcal{D}} = \left\{ \begin{array}{ll}
				\widehat{\sigma}_{u,ij}, & i=j \\ 
				s_{ij}(\widehat{\sigma}_{u,ij})I(|\widehat{\sigma}_{u,ij}| \geq \tau_{ij}), & i \neq j
			\end{array}\right., 
		\end{equation*}
		where  an entry-dependent threshold $\tau_{ij} = \tau(\widehat{\sigma}_{u,ii}\widehat{\sigma}_{u,jj})^{1/2}$ and $s_{ij}(\cdot)$ is a generalized shrinkage function (e.g., hard or soft thresholding; see \citet{cai2011adaptive, rothman2009generalized}). The thresholding constant $\tau$ will be determined in Theorem \ref{thm1}.
		\item The final estimator of $\bSigma$ is then defined as
		\begin{equation}\label{estimator}
		\widehat{\bSigma}^{\mathcal{D}} =  \widehat{\bV}\widehat{\bGamma}\widehat{\bV}' + \widehat{\bPhi}\widehat{\bPsi}\widehat{\bPhi}'+\widehat{\bSigma}_{u}^{\mathcal{D}}.
		\end{equation}
		\end{enumerate}

The above procedure can be equivalently represented by a least squares method. 
In particular, the global and local factor matrices and corresponding loading matrices are estimated as follows.
To obtain the global factor part, we first solve the following optimization problem:
			\begin{align}\label{least squares}
			(\widehat{\bB}, \widehat{\bG} ) = \arg\min_{\scalebox{.7}{$\scriptscriptstyle \bB, \bG$}}\|\bY - \bG\bB' \|_{F}^{2},
			\end{align}
subject to the normalization constraints such that
			\begin{align*}
			\frac{1}{T}\bG'\bG = \bI_{k}  \quad \text{and} \quad   \bB'\bB \text{ is diagonal matrix}.
			\end{align*}
The columns of $\widehat{\bG}/\sqrt{T}$ are the eigenvectors corresponding to the $k$ largest eigenvalues of the $T\times T$ matrix $T^{-1}\bY\bY'$ and $\widehat{\bB} = (\hat{b}_{1}, \dots, \hat{b}_{p})' = T^{-1}\bY'\widehat{\bG}$. 
Given $\widehat{\bG}$ and $\widehat{\bB}$, let a $T \times p$ matrix $\widehat{\bE}  = \bY-\widehat{\bG}\widehat{\bB}' \equiv (\widehat{E}^{1},\dots, \widehat{E}^{J}) $. 
Then, with $\widehat{\bE}$, we estimate the local factor part as follows:
			\begin{align}\label{least squares-2}
			(\widehat{\Lambda}^j, \widehat{F}^j) = \arg\min_{\scalebox{.7}{$\scriptscriptstyle  \Lambda ^j, F^j$}}\| \widehat{E}^j   - F ^j \Lambda^{j \prime} \|_{F}^{2},
			\end{align}
subject to the normalization constraints such that
			\begin{align*}
			  \frac{1}{T} F^{j \prime}  F^j  = \bI_{r_j} \quad \text{and} \quad \Lambda^{j \prime}\Lambda^j   \text{ is diagonal matrix}.
			\end{align*}
Then, we obtain $\widehat{\bF} = (\widehat{F}^{1}, \dots, \widehat{F}^{J})$ and $\widehat{\bLambda} = \diag(\widehat{\Lambda}^{1}, \dots, \widehat{\Lambda}^{J})$, where the columns of $\widehat{F}^{j}/\sqrt{T}$ are the eigenvectors corresponding to the $r_{j}$ largest eigenvalues of the $T\times T$ matrix $T^{-1}\widehat{E}^{j}\widehat{E}^{j\prime}$ and $\widehat{\Lambda}^{j} = (\hat{\lambda}_{1}^{j}, \dots, \hat{\lambda}_{p_{j}}^{j}) = T^{-1} \widehat{E}^{j\prime}\widehat{F}^{j}$ for $j=1,2,\dots,J$. 
Finally, we apply the above adaptive thresholding method to $\widehat{\bSigma}_{u} = T^{-1}\widehat{\bU}'\widehat{\bU}$, where $\widehat{\bU} = \bY -\widehat{\bG}\widehat{\bB}' - \widehat{\bF}\widehat{\bLambda}'$. 
Similar to the decomposition (\ref{model2}), we have the following substitution estimators:
			\begin{align}\label{substitution estimator}
			\widetilde{\bSigma}^{D} = \widehat{\bB}\widehat{\bB}' + \widehat{\bLambda}\widehat{\bLambda}' + \widehat{\bSigma}_{u}^{\mathcal{D}}.
			\end{align}		
Similar to the proofs of Theorem 1 of \cite{fan2013large}, we can show that the estimators (\ref{estimator}) based on PCA and (\ref{substitution estimator}) based on least squares approaches are equivalent.
In  particular, by the Eckart-Young theorem, the estimators $\widehat{\bB}$ and $\widehat{\Lambda}^j$, after normalization, are the first $k$ and $r_{j}$ empirical eigenvectors of the sample covariance matrix $\widehat{\bSigma}$ and $\widehat{\bSigma}_{E}^j$, respectively.
Then, there exist orthogonal matrices $\bH$ and $\bJ$ such that 
\begin{align*}
	\|\widehat{\bB}-\bB\bH'\|_{\max}= O_{P}(\widetilde{\omega}_{T}), \qquad \|\widehat{\bLambda} - \bLambda\bJ'\|_{\max} = O_{P}(\omega_T),
\end{align*}
where $\widetilde{\omega}_{T}$ and $\omega_T$ are defined in Section \ref{regularPOET} and Theorem \ref{thm1}, respectively.
The above rates can be easily obtained using the preliminary results in Appendix \ref{proofs}.

In summary, given the knowledge of group membership, we employ PCA on each diagonal block of the remaining components of the sample covariance matrix, after removing the first $k$ principal components. 
Then, we apply thresholding on the remaining residual components. 
We call this procedure the Double Principal Orthogonal complEment Thresholding (Double-POET).
Double-POET is the two-step estimation procedure, which makes it possible to estimate global and local factors separately by considering the block structure on the local factors.
In contrast, if we employ the single step estimation procedure, such as POET, the local factor structure cannot be explained well. 
We discuss the theoretical inefficiency of POET in Sections \ref{regularPOET} and \ref{blessinng of Dim},  and the numerical study illustrated in Section \ref{simulation} shows that Double-POET outperforms POET. 
 

\section{Asymptotic Properties} \label{asymp}

In this section, we establish the asymptotic properties of the Double-POET estimator. 
To do this, we impose the following technical conditions.

\begin{assum} \label{assum1} ~
\begin{itemize}
\item[(i)] For some constants $c \in (0,1]$, $a_{1} \in (\frac{3+2c}{5},1]$ and $a_{2} \in (\frac{3}{5},1]$, 
all  eigenvalues of $\bB'\bB/p^{a_{1}}$ and $\Lambda^{j\prime}\Lambda^{j}/p_{j}^{a_{2}}$ are strictly bigger than zero as $p, p_{j} \rightarrow \infty$, for $j \in \{1,\dots,J\}$ and $a_{1} \geq ca_{2}$. 
If $a_{1} = ca_{2}$, we further assume that $\lambda_{\min}(\bB'\bB) -\lambda_{\max}(\bLambda'\bLambda) \geq d_{\lambda}$ for some positive constant $d_{\lambda}$.
In addition, there is a constant $d>0$ such that $\|\bB\|_{\max}\leq d$ and $\|\bLambda\|_{\max} \leq d$.

\item[(ii)] There exists constants $d_{1}, d_{2} > 0$ such that $\lambda_{\min}(\bSigma_{u})>d_{1}$ and $\|\bSigma_{u}\|_{1} \leq d_{2}m_{p}$.

\item[(iii)] The sample covariance matrix $\widehat{\bSigma}$ satisfies
	\begin{equation*}
		\|\widehat{\bSigma}-\bSigma\|_{\max} = O_{P}(\sqrt{\log p /T}).  \label{max norm}
	\end{equation*}
\end{itemize}
\end{assum}

\begin{remark}
Assumption \ref{assum1}(i) ensures that the factors are pervasive. 
This pervasive assumption is essential to analyze low-rank matrices \citep{fan2013large, fan2018eigenvector} and is reasonable in many financial applications \citep{bai2003inferential, Chamberlain1983, kim2019factor, lam2012factor, stock2002forecasting}. 
 For example, under the multi-level factor model, the global factors impact  on most of the assets, while the national factors affect only those belonging to each country. 
This structure of the latent factors implies the pervasive condition, which is related to the incoherence structure \citep{fan2018eigenvector}. 
To analyze large matrix inferences, we impose the element-wise convergence condition (Assumption \ref{assum1}(iii)).
Under the sub-Gaussian condition and the mixing time dependency, as considered in \citet{fan2013large}, this condition can be easily satisfied (\citealp{fan2018large, fan2018eigenvector, vershynin2010introduction, wang2017asymptotics}). 
We can obtain this condition under the heavy-tailed observations with the condition of bounded fourth moments \citep{fan2017estimation, fan2018large, fan2018eigenvector, fan2021shrinkage}. 
Furthermore, when observations are martingales with bounded fourth moments, we can obtain the element-wise convergence condition \citep{fan2018robust, shin2021adaptive}.
\end{remark}

To measure large matrix estimation errors,  we consider the  relative Frobenius norm, introduced by \cite{stein1961estimation}: 
$$
\|\widehat{\bSigma}-\bSigma\|_{\Sigma} = p^{-1/2}\|\bSigma^{-1/2}\widehat{\bSigma}\bSigma^{-1/2}-\bI_{p}\|_{F}.
$$
Note that the factor $p^{-1/2}$ performs normalization, such that $\|\bSigma\|_{\Sigma} = 1$. 
   \cite{fan2013large} showed that, under this relative Frobenius norm, the POET estimator is consistent as long as $p=o(T^{2})$, while the sample covariance is consistent only if $p=o(T)$ in the approximate single-level factor model.
In Section \ref{regularPOET}, we will compare the convergence rates of POET and Double-POET in a multi-level factor model.

The following theorem provides the convergence rates  under various norms.
\begin{thm}\label{thm1}
	Under Assumption \ref{assum1}, suppose that $p_{j} \asymp p^{c}$ for each $j$ and $m_{p} = o(p^{c(5a_{2}-3)/2})$. 
	Let $\omega_{T} =p^{\frac{5}{2}(1-a_{1})+\frac{5}{2}c(1-a_{2})}\sqrt{\log p/T}+1/p^{\frac{5}{2}a_{1}-\frac{3}{2}+c(\frac{5}{2}a_{2}-\frac{7}{2})} +m_{p}/\sqrt{p^{c(5a_{2}-3)}}$ and  $\tau \asymp \omega_{T}$. 
	If $m_{p}\omega_{T}^{1-q} = o(1)$, we have
	\begin{align}
		&\|\widehat{\bSigma}_{u}^{\mathcal{D}} - \bSigma_{u}\|_{\max} = O_{P}(\omega_{T}),\label{u_max}\\
		&\|\widehat{\bSigma}_{u}^{\mathcal{D}} - \bSigma_{u}\|_{2} = O_{P}(m_{p}\omega_{T}^{1-q}), \quad 	\|(\widehat{\bSigma}_{u}^{\mathcal{D}})^{-1} - \bSigma_{u}^{-1}\|_{2} = O_{P}(m_{p}\omega_{T}^{1-q}),  \label{error rate}\\
		&\|\widehat{\bSigma}^{\mathcal{D}} - \bSigma\|_{\max} = O_{P}(\omega_{T}), \label{maxnorm}\\
		& \|(\widehat{\bSigma}^{\mathcal{D}})^{-1} - \bSigma^{-1}\|_{2} = O_{P}(m_{p}\omega_{T}^{1-q}). 	\label{inverse rate}
	\end{align}
	In addition, if $a_{1}>\frac{3}{4}$ and $a_{2}>\frac{3}{4}$, we have
	\begin{align}
		&\|\widehat{\bSigma}^{\mathcal{D}} - \bSigma\|_{\Sigma} =O_{P}\left(m_{p}\omega_{T}^{1-q}+ 	p^{\frac{11}{2}-5a_{1}+5c(1-a_{2})}\frac{\log p}{T}+ 
		\frac{1}{p^{5a_{1}-\frac{7}{2}-c(7-5a_{2})}} +\frac{m_{p}^2}{p^{5ca_{2}-3c-\frac{1}{2}}}\right). 	\label{entropynorm}
	\end{align}
\end{thm}

\begin{remark}
	The additional terms $1/p^{\frac{5}{2}a_{1}-\frac{3}{2}+c(\frac{5}{2}a_{2}-\frac{7}{2})}$ and $m_{p}/\sqrt{p^{c(5a_{2}-3)}}$ in $\omega_{T}$ result from the estimation of unknown global factors and national factors. 
	They are negligible when the dimensional $p$ is high as long as $0<c<\frac{5a_{1}-3}{7-5a_{2}}$.
	In contrast, for the regular POET in the single-level factor model, $1/\sqrt{p}$ appears in the threshold \citep{fan2013large}. 
	Under the relative Frobenius norm, the upper and lower bounds of $c$ are required to maintain the consistency of the Double-POET estimator. 
	Specifically, when $m_{p} = O(1)$, the national factor estimation requires $\frac{1}{2(5a_{2}-3)} < c < \frac{10a_{1}-7}{14-10a_{2}}$, which is satisfied when both strength levels of global and local factors are sufficiently large.
\end{remark}

\begin{remark}
	The asymptotic theory relies on the relative rates of the number of assets in each group, $p_j \asymp p^{c}$, as well as the number of groups, $G \asymp p^{1-c}$. 
	We note that the total number of local factors grows at a rate $O(p^{1-c})$.
	For simplicity, consider the case of strong global and local factors (i.e., $a_{1}=1$ and $a_{2}=1$).
	Hierarchically, to effectively estimate the total local factors, there should be sufficiently large number of assets in each group, but the number of group should not grow much faster than the number of assets in each group to enjoy the blessing of dimensionality. 
	Otherwise, the sum of local factor estimation errors will  diverge as $G$ grows, which causes the curse of dimensionality. 
	Practically, the number of assets in each country is sufficiently larger than the number of countries. 
	Hence, the Double-POET method can account for the global stock market data.
\end{remark}

\subsection{POET for the Multi-Level Factor Models}\label{regularPOET}
In this subsection, we analyze the regular POET of \cite{fan2013large} in the multi-level factor models  and compare it with the proposed  Double-POET.

The regular POET method only captures the global factors and regards both the local factor structure and idiosyncratic terms as the idiosyncratic part. 
That is, the sparsity level of $\bSigma_{E} = (\varepsilon_{ij})_{p\times p}$ is 
$$
\mu_{p} = \max_{i\leq p}\sum_{j\leq p} |\varepsilon_{ij}|^{q},
$$
for some $q \in [0,1)$. 
We note that when $q=0$, $\mu_{p} \asymp p^{c}$, which corresponds to the maximum number of nonzero elements in each row of $\bSigma_{E}$.
Then,  the  thresholding approach is applied to $\widehat{\bSigma}_{E}$ instead of $\widehat{\bSigma}_{u}$ to obtain $\widehat{\bSigma}_{E}^{\mathcal{T}}$.
Therefore, the regular POET estimator is defined as 
$$
\widehat{\bSigma}^{\mathcal{T}} = \widehat{\bV}\widehat{\bGamma}\widehat{\bV}' +  \widehat{\bSigma}_{E}^{\mathcal{T}}.
$$
Similar to the proofs of \cite{fan2013large},  we can show that the regular POET yields
\begin{align}
	&\|\widehat{\bSigma}_{E}^{\mathcal{T}}-\bSigma_{E}\|_{2} = O_{P}(\mu_{p}\widetilde{\omega}_{T}^{1-q}),\label{errorrate_POET}\\
	&\|\widehat{\bSigma}^{\mathcal{T}}-\bSigma\|_{\Sigma} = O_{P}\left(\mu_{p}\widetilde{\omega}_{T}^{1-q} + p^{\frac{11}{2}-5a_{1}}\frac{\log p}{T} + \frac{1}{p^{5a_{1}-\frac{7}{2}-2c}}\right),\label{entropy_POET}
\end{align}
where $\widetilde{\omega}_{T} = p^{\frac{5}{2}(1-a_{1})}\sqrt{\log p/T} + 1/p^{\frac{5}{2}a_{1}-\frac{3}{2}-c}$.
We compare the convergence  rates of the regular POET and proposed Double-POET.
For example, under the relative Frobenius norm, when $q = 0$, $m_{p} = O(1)$, $a_{1} =1$, and $a_{2}=1$, we have 
\begin{align*}
	&\|\widehat{\bSigma}^{\mathcal{T}}-\bSigma\|_{\Sigma} = O_{P}\left(p^{c}\sqrt{\frac{\log p}{T}} + \frac{1}{p^{1-2c}} + \frac{\sqrt{p}\log p}{T}\right),\\
	&\|\widehat{\bSigma}^{\mathcal{D}} - \bSigma\|_{\Sigma}=O_{P}\left(\sqrt{\frac{\log p}{T}} + \frac{1}{p^{1-c} + p^{c}}+\frac{\sqrt{p}\log p}{T} + \frac{1}{p^{\frac{3}{2}-2c}} + \frac{1}{p^{2c-\frac{1}{2}}}\right).
\end{align*}
The number of assets within a group, $p_j \asymp p^{c}$, for some constant $c>0$ plays a crucial role in a convergence rate. 
Theorem \ref{thm1} shows that $\|\widehat{\bSigma}^{\mathcal{D}} - \bSigma\|_{\Sigma}$ can be convergent as long as $p=o(T^{2})$ and $\frac{1}{4}<c<\frac{3}{4}$. 
In contrast, the rate of the regular POET estimator, $\|\widehat{\bSigma}^{\mathcal{T}}-\bSigma\|_{\Sigma}$, does not converge if $c>\frac{1}{2}$ or $p^{\alpha}>T$ with $\alpha = \min\{\frac{1}{2}, 2c\}$. 
In other words, the regular POET requires  the number of assets in each country to be small enough to avoid the curse of dimensionality. 
However, the number of assets in each country is larger than the number of countries; thus, it is more reasonable to assume $c>\frac{1}{2}$. 
Therefore, under the global and national factor models, the regular POET does not provide a consistent estimator in terms of the relative Frobenius norm.
Furthermore, when the global factors are weak (i.e., $a_{1}<1$) and the local factors are strong (i.e., $a_{2}=1$), Double-POET is equal to or better than POET in terms of the convergence rate under the relative Frobenius norm.

\subsection{Orthogonality between the Global and Local Factor Loadings }\label{orthogonality}

When the signal of local factors is strong, we often consider the local factors as the global weak factors.
Then, we can apply the regular POET method with the global and local factors.
Theoretically, to identify the latent factors, we additionally need to impose an orthogonality condition between the global and local factor loadings, $\bB$ and $\bLambda$.
In this section, under this orthogonality condition, we investigate asymptotic behaviors of the Double-POET procedure and compare POET with Double-POET.

Under the orthogonality condition, we first obtain the following modified theoretical results for Double-POET.

\begin{thm}\label{DPOET with orthogonality}
	Suppose that $\bB$ and $\bLambda$ are orthogonal each other.
	Under Assumption \ref{assum1}, suppose that $p_{j} \asymp p^{c}$ for each $j$ and $m_{p} = o(\min\{p^{(5a_{1}-3)/2}, p^{c(5a_{2}-3)/2}\})$. 
	Let $\omega_{T,2} = p^{\frac{5}{2}(1-a_{1})+\frac{5}{2}c(1-a_{2})}\sqrt{\log p/T} +m_{p}/\sqrt{p^{c(5a_{2}-3)}} +m_{p}/\sqrt{p^{5a_{1}-3-5c(1-a_{2})}}$ and  $\tau \asymp \omega_{T,2}$. 
	If $m_{p}\omega_{T,2}^{1-q} = o(1)$, we have
				\begin{align}
							&\|\widehat{\bSigma}_{u}^{\mathcal{D}} - \bSigma_{u}\|_{\max} = O_{P}(\omega_{T,2}),\label{u_max2}\\
							&\|\widehat{\bSigma}_{u}^{\mathcal{D}} - \bSigma_{u}\|_{2} = O_{P}(m_{p}\omega_{T,2}^{1-q}), \quad 	\|(\widehat{\bSigma}_{u}^{\mathcal{D}})^{-1} - \bSigma_{u}^{-1}\|_{2} = O_{P}(m_{p}\omega_{T,2}^{1-q}),  \label{error rate2}\\
							&\|\widehat{\bSigma}^{\mathcal{D}} - \bSigma\|_{\max} = O_{P}(\omega_{T,2}), \label{maxnorm2}\\
							& \|(\widehat{\bSigma}^{\mathcal{D}})^{-1} - \bSigma^{-1}\|_{2} = O_{P}(m_{p}\omega_{T,2}^{1-q}). 	\label{inverse rate2}
				\end{align}
	In addition, if $a_{1}>\frac{3}{4}$ and $a_{2}>\frac{3}{4}$, we have
				\begin{align}
							&\|\widehat{\bSigma}^{\mathcal{D}} - \bSigma\|_{\Sigma} =O_{P}\left(m_{p}\omega_{T,2}^{1-q}+ 	p^{\frac{11}{2}-5a_{1}+5c(1-a_{2})}\frac{\log p}{T}+ 
								\frac{m_{p}^{2}}{p^{5a_{1}-\frac{7}{2}-5c(1-a_{2})}} +\frac{m_{p}^2}{p^{c(5a_{2}-3)-\frac{1}{2}}}\right). 	\label{entropynorm2}
				\end{align}
\end{thm}

\begin{remark}
The orthogonality condition helps identify the latent global and local factors by reducing perturbation terms. 
Thus, compared to the results of Theorem \ref{thm1},  Theorem \ref{DPOET with orthogonality} shows the faster convergence rate, and we can relax the upper bound for $c$. 
For example, the additional terms $m_{p}/\sqrt{p^{c(5a_{2}-3)}}$ and $m_{p}/\sqrt{p^{5a_{1}-3-5c(1-a_{2})}}$ in $\omega_{T,2}$ are  negligible when $0<c<\frac{5a_{1}-3}{5(1-a_{2})}$ and $c\leq \frac{a_{1}}{a_{2}}$ as  $p \rightarrow \infty$.
In addition, when $m_{p} = O(1)$,  $\frac{1}{2(5a_{2}-3)} < c < \frac{a_{1}-0.7}{1-a_{2}}$ and $c\leq \frac{a_{1}}{a_{2}}$ are required to  make $\|\widehat{\bSigma}^{\mathcal{D}} - \bSigma\|_{\Sigma}$ converge.
Therefore, the Double-POET method does not require the upper bound for $c$ when the global and local factors are strong.
For example, the number of group, $J$, can be fixed (i.e., $c=1$) as long as $a_{1}\geq a_{2}$.
\end{remark}

We now consider the POET estimator that regards both the global and local factors as the global factors under the orthogonality condition.
We call this the POET2 estimator. 
The estimator is then defined as follows:
$$
\widehat{\bSigma}_{2}^{\mathcal{T}} = \widehat{\bV}_{2}\widehat{\bGamma}_{2}\widehat{\bV}_{2}' +  \widehat{\bSigma}_{u,2}^{\mathcal{T}},
$$
where $\widehat{\bGamma}_2 = \diag(\widehat{\delta}_{1}, \dots, \widehat{\delta}_{k+r})$ and $\widehat{\bV}_{2}=(\widehat{v}_{1},\dots, \widehat{v}_{k+r})$.
Then, when $a_{1}= a_2 =1$, the POET2 estimator yields
\begin{align}
	&\|\widehat{\bSigma}_{u,2}^{\mathcal{T}}-\bSigma_{u}\|_{2} = O_{P}(m_{p}\widetilde{\omega}_{T,2}^{1-q}),\label{errorrate_POET2}\\
	&\|\widehat{\bSigma}_{2}^{\mathcal{T}}-\bSigma\|_{\Sigma} = O_{P}\left(m_{p}\widetilde{\omega}_{T,2}^{1-q} + p^{15(1-c)+\frac{1}{2}}\frac{\log p}{T} + \frac{m_{p}^{2}}{p^{15c-\frac{27}{2}}}\right),\label{entropy_POET2}
\end{align}
where $\widetilde{\omega}_{T,2} = p^{7(1-c)}\sqrt{\log p/T} + m_{p}/p^{7c-6}$.	
We note that the additional term $m_{p}/p^{7c-6}$ in $\widetilde{\omega}_{T,2}$ can be negligible only when $c > \frac{6}{7}$ as $p$ increases.
This is because the POET2 estimator includes noises on the off-diagonal blocks of the local factor component.
Therefore, it requires the number of groups to be sufficiently small.
Otherwise, the POET2 estimator does not perform well.
That is,  since the local factors are considered as the weak factors of the global factor, the local factor should have enough signals.

We compare the convergence  rates of the  POET2 and  Double-POET estimators under the orthogonality condition.
	When $q = 0$, $m_{p} = O(1)$, $a_{1} =1$, and $a_{2}=1$, we have 
	\begin{align*}
		&\|\widehat{\bSigma}_{2}^{\mathcal{T}}-\bSigma\|_{\Sigma} = O_{P}\left(p^{7(1-c)}\sqrt{\frac{\log p}{T}} + \frac{1}{p^{7c-6}}+p^{15(1-c)+\frac{1}{2}}\frac{\log p}{T} + \frac{1}{p^{15c-\frac{27}{2}}}\right),\\		
		&\|\widehat{\bSigma}^{\mathcal{D}} - \bSigma\|_{\Sigma}=O_{P}\left(\sqrt{\frac{\log p}{T}} + \frac{1}{p^{c}}+\frac{\sqrt{p}\log p}{T} + \frac{1}{p^{2c-\frac{1}{2}}}\right).
	\end{align*}	
	When $c<1$, the convergece rate of Double-POET is faster than that of POET2. 
	Furthermore,  $\|\widehat{\bSigma}^{\mathcal{D}} - \bSigma\|_{\Sigma}$ can be consistent as long as $p=o(T^{2})$ and $c>\frac{1}{4}$, while $\|\widehat{\bSigma}_{2}^{\mathcal{T}}-\bSigma\|_{\Sigma}$ can be consistent when $p^{15(1-c)+\frac{1}{2}}=o(T)$ and $c > \frac{9}{10}$.
		This implies that POET2 performs well only for the multi-level factor model with a small number of groups (i.e., weak factors with enough signals).
		We also note that when the local factor has the same signal as the global factor ($c=1$), the POET2 and Double-POET procedures have the same convergence rate.


\subsection{Blessing of Dimensionality}  \label{blessinng of Dim}
In this section, we demonstrate the blessing of dimensionality for  country-wise covariance matrix estimation by using the proposed Double-POET procedure. 
For the country $j$, we have
\begin{equation*}
	y_{t}^{j} = B^{j}G_{t} + \Lambda^{j}f_{t}^{j}+ u_{t}^{j}, 
\end{equation*}
where $B^{j} = (b_{1}^{j},\dots, b_{p_j}^{j})'$, $\Lambda^{j} = (\lambda_{1}^{j},\dots, \lambda_{p_j}^{j})'$, and $u_{t}^{j} = (u_{1t}, \dots, u_{p_{j}t})'$. 
Then, the covariance matrix of the $j$th country can be written as
\begin{equation*}
	\bSigma^{j} = B^{j}\cov(G_{t})B^{j\prime} + \Lambda^{j}\cov(f_{t}^{j})\Lambda^{j\prime} + \bSigma_{u}^{j}.
\end{equation*}
Then, the sparsity level of $\bSigma_{u}^{j}$ is 
$$
m_{p_j} := \max_{i\leq p_j}\sum_{j\leq p_j}|\sigma_{u,ij}|^{q}, \text{ for some } q \in [0,1).
$$
To estimate the $j$th country's covariance matrix $\bSigma^{j}$, assuming that the common factors include both global and national factors with the orthogonality condition between their factor loadings,  the regular POET method can be applied  with $k+r_j$ principal components using the  $j$th country's stock market data (i.e., $T \times p_j$ observation matrix), denoted by $\widehat{\bSigma}^{j,\mathcal{T}}$. 
As long as $p_j \rightarrow \infty$, the asymptotic results of \cite{fan2013large} can be directly applied by replacing $p$ by $p_j \asymp p^{c}$ for $c \in (0,1]$.
In contrast, we suggest the Double-POET method by extracting the $j$th diagonal block of $\widehat{\bSigma}^{\mathcal{D}}$, denoted by $(\widehat{\bSigma}^{\mathcal{D}})^{j}$. 
 Then, when estimating  the global factor component,  Double-POET uses more data, which might result in a more accurate global factor estimator. 
Specifically, we have the following convergence rates of  Double-POET estimator for the local covariance matrix.
\begin{thm}\label{blessing thm}
	Under the assumptions of Theorem \ref{DPOET with orthogonality}, the Double-POET estimator satisfies,
	if $m_{p_j}\omega_{T,2}^{1-q} = o(1)$, 
	\begin{align}
		&\|(\widehat{\bSigma}^{\mathcal{D}})^{j} - \bSigma^{j}\|_{\max} = O_{P}(\omega_{T,2}), \label{local max}\\
		&\|((\widehat{\bSigma}^{\mathcal{D}})^{j})^{-1} - (\bSigma^{j})^{-1}\|_{2} = O_{P}(m_{p_j}\omega_{T,2}^{1-q}). \label{local inverse}
	\end{align}
	In addition, if $a_{1}>\frac{3}{4}$ and $a_{2}>\frac{3}{4}$, we have
	\begin{align}
		&\|(\widehat{\bSigma}^{\mathcal{D}})^{j} - \bSigma^{j}\|_{\Sigma^{j}}=O_{P}\left(m_{p_j}\omega_{T,2}^{1-q}+ p^{5(1-a_{1})+c(\frac{11}{2}-5a_{2})}\frac{\log p}{T} + \frac{m_{p}^{2}}{p^{5a_{1}-3+c(5a_{2}-\frac{11}{2})}}+ \frac{m_{p}^2}{p^{c(5a_{2}-\frac{7}{2})}}\right), \label{local entropy}		
	\end{align}
where $\|A\|_{\Sigma^{j}} = p_j^{-1/2}\|{\Sigma^{j}}^{-1/2}A{\Sigma^{j}}^{-1/2}\|_{F}$ is the relative Frobenius norm.
\end{thm}

\begin{remark}\label{blessing remark}
We compare the rates of convergence between the Double-POET and POET estimators as follows.
  Following \cite{fan2018large}, under the assumptions of Theorem \ref{blessing thm}, the regular POET estimator satisfies  the following conditions, with $\tau \asymp \ddot{\omega}_{T}$,  if $m_{p_j}\ddot{\omega}_{T} =o(1)$, $a_{1}=1$, and $a_{2}=1$:
\begin{align*}
	&\|\widehat{\bSigma}^{j,\mathcal{T}} - \bSigma^{j}\|_{\max} = O_{P}(\ddot{\omega}_{T}),\\	
	& \left\|(\widehat{\bSigma}^{j,\mathcal{T}})^{-1} -(\bSigma^{j})^{-1}\right\|_{2} = O_{P}(m_{p_j}\ddot{\omega}_{T}^{1-q}),\\
	&\|\widehat{\bSigma}^{j,\mathcal{T}} - \bSigma^{j}\|_{\Sigma^{j}}=O_{P}\left(m_{p_j}\ddot{\omega}_{T}^{1-q}+\frac{\sqrt{p^c}\log p}{T} \right), 
\end{align*}
where  $\ddot{\omega}_{T}=\sqrt{\log p/T} + m_{p_j}/p^{c}$.
The overall rates of convergence under each norm are the same.
This might be because the convergence rate is dominated by the national factor estimator, which is estimated by the Double-POET and POET procedures, using the same amount of the data.
 However, when we focus on the global factor components, the Double-POET procedure indeed can have a faster rate of convergence than POET. 
  For example, when $a_{1}=1$ and $a_{2}=1$, under the max norm, the convergence rates  of Double-POET and POET for $B^{j}B^{j\prime}$ are $O_{P}(\sqrt{\log p/T} + m_{p}/p)$ and $O_{P}(\sqrt{\log p/ T} + m_{p}/p^{c})$, respectively.
   In addition,  in terms of the relative Frobenius norm,  the convergence rates  of  Double-POET and POET for $B^{j}B^{j\prime}$ are $O_{P}(\sqrt{\log p/T} + \sqrt{p^c}\log p/T + m_{p}/p + m_{p}^{2}/p^{2-c/2})$ and $O_{P}(1/\sqrt{T} + \sqrt{p^c}\log p/T + m_{p}^{2}/p^{3c/2})$, respectively. 
Thus, the global factor estimator of  Double-POET has a faster convergence rates than that of POET when $c<1$ (under the max norm) or $c<\frac{2}{3}$ when $m_{p}=O(1)$ (under the relative Frobenius norm).
Furthermore, the numerical study presented  in Section \ref{simulation} shows that Double-POET outperforms POET, especially for small $p_j$. 
This implies that Double-POET  enjoys the blessing of dimensionality by incorporating other countries' observations.
That is, the proposed Double-POET model not only presents a way to investigate the global stock market, but also shows the benefits of incorporating financial big data.
\end{remark}

\subsection{Determination of the Number of Factors} \label{sec number of global factors}

 To implement  Double-POET, we need to determine the number of factors. 
 In this section, we describe a data-driven method for determining the number of global and national factors.

 We suggest a modified version of the eigenvalue ratio method proposed by \cite{ahn2013eigenvalue} as follows.
We first consider a model selection problem between the multi-level factor model and the single-level factor model, and then propose estimators for the number of factors.
 Let $\widehat{\delta}_{m}$ be the $m$th largest eigenvalue of the sample covariance matrix and $\er(m) = \widehat{\delta}_{m}/\widehat{\delta}_{m+1}$ be the $m$th eigenvalue ratio.
Under the multi-level factor model (i.e., $k>0$ and $r>0$), the first $k$ eigenvalues of the sample covariance matrix are asymptotically determined by the eigenvalues of $\bB\cov(G_{t})\bB'$, the next $r = \sum_{j=1}^{J}r_{j}$ eigenvalues by the eigenvalues of $\bLambda\cov(F_t)\bLambda'$, and the other eigenvalues by those of the idiosyncratic covariance matrix.
 Accordingly, when $a_{1}=1$ and $a_{2}=1$, $\er(m) = O_{P}(1)$ for $m\neq k$ and $m\neq k+r$,  $\er(k) = O_{P}(p^{1-c})$, and $\er(k+r) = O_{P}(p^{c})$.
 This implies that there are two diverging eigenvalue ratios when $0<c<1$, while all other ratios of two adjacent eigenvalues are asymptotically bounded. 
In contrast, under the single-level factor model, there exists only one diverging eigenvalue ratio.
 We define
 $$
 \hat{k}_{1} = \max_{1\leq m \leq k_{\max}}\er(m) \;\;\; \text{ and } \;\;\; \hat{k}_{2} = \max_{1\leq m \neq \hat{k}_{1} \leq k_{\max}}\er(m),
 $$
 for a prespecified $k_{\max} < \min(p,T)$.
Let $\varphi_{p}$ be the tuning parameter, which grows slowly. 
 	In practice, we set $\varphi_{p} = d\log p$ for a positive constant $d$.
 	We then select the single-level factor model when $\er(\hat{k}_{1}) > \varphi_{p}$ and $\er(\hat{k}_{2}) \leq \varphi_{p}$, and one can apply the regular POET with $\hat{k}_{1}$ factors.
 	When $\er(\hat{k}_{2})  > \varphi_{p}$, we select the multi-level factor model and estimate the number of global factors $k$ by $\hat{k} = \min\{\hat{k}_{1}, \hat{k}_{2}\}$.
 Under the following technical conditions, we can show the consistency of  $\hat{k}$.
 \begin{assum} \label{assum for num factors} ~
 	\begin{enumerate}
 		\item[(i)] There exist  $T \times T$ and $p \times p$ positive semidefinite matrices $R_{T}$ and $G_{T}$ such that $\bU = R_{T}^{1/2}\Upsilon G_{p}^{1/2}$, where the idiosyncratic observation matrix $\bU$ is defined in \eqref{modelUU}, and $\Upsilon' = [\upsilon_{it}]_{p \times T}$.

 		\item[(ii)]  $\upsilon_{it}$ indicates i.i.d. random variables with uniformly bounded moments up to the fourth order.
 		
 		\item[(iii)] There are generic positive constants $d_{1}, d_{2} >0$ such that $\lambda_{\max}(R_{T}) < d_{1}$, $\lambda_{\max}(G_{p}) < d_{1}$, and $\lambda_{\min}(R_{T}) > d_{2}$. 
 		
 		\item[(iv)]  Let $y^{\ast} = \lim_{m\rightarrow \infty}m/p$, where $m = \min(p,T)$. 
 		Then, there exists a real number $d^{\ast} \in (1-y^{\ast}, 1]$ such that the $(d^{\ast}p)$th largest eigenvalue of $G_{p}$ is strictly greater than zero for all $p$.
 	\end{enumerate}
 \end{assum}
 Then, we have the following theorem.
 \begin{thm}\label{num of global factors}
 	Suppose Assumptions \ref{assum1}--\ref{assum for num factors} hold, and in addition suppose that $m_{p} = o(p^{ca_{2}})$, $c < \frac{5a_{1}-3}{4-2a_{2}}$, and $\max\{p^{1-ca_{2}}, p^{5(1-a_{1})+2c(1-a_{2})}\log p\} = o(T)$.
 	\begin{enumerate}
 		\item[(i)] Given a tunning parameter $\varphi_{p} < \min\{p^{a_{1}-ca_{2}}, \frac{m}{p^{1-ca_{2}}}\}$, the proposed model selection is consistent.
 		
 		\item[(ii)] When $k>0$ and $r>0$, there exists a constant $d \in (0,1]$ such that $\lim_{m \rightarrow \infty}\Pr(\hat{k} = k) = 1$, for any $k_{\max} \in (k+r, dm-2(k+r)-1]$.
 	\end{enumerate}
 \end{thm}

  Given the consistently estimated number of global factors, we can apply the existing methods to consistently estimate the number of local factors $r_j$ \citep{alessi2010improved, bai2002determining, choi2018multilevel, giglio2021asset, onatski2010determining, trapani2018randomized}. 
 Throughout the paper, we use the eigenvalue ratio method of \cite{ahn2013eigenvalue} as follows: for each group $j$, let $\widehat{\kappa}_{r}^{j}$ be the $r$th largest eigenvalues of the $p_{j} \times p_{j}$ matrix $\widehat{\bSigma}_{E}^{j}$ in the Section \ref{estimation procedure}. 
 Then,  $r_j$ can be estimated by 
 $$
 \hat{r}_{j} = \max_{1 \leq r \leq r_{j,\max}}\frac{\widehat{\kappa}_{r}^{j}}{\widehat{\kappa}_{r+1}^{j}},
 $$
  for a predetermined $r_{j,\max}$.

\subsection{Unknown Local Group Membership}\label{RSC}

In this paper, we assume that local factors are governed by the national regional risk factors, which provides the membership of the local factors.
However, in practice,  the membership of the local factors is unknown.
In this section, we discuss how to use the proposed Double-POET procedure for the unknown local factor membership.

The Double-POET procedure is working as long as the membership for the local factor is known.
Thus, to harness the Double-POET, we need to classify assets and find the latent local structure. 
As discussed in the previous sections, we can estimate the global factor even if the local factor structure is unknown. 
Thus, we can consistently estimate $\bSigma_{E}$ defined in \eqref{decomposition}, which is the remaining covariance matrix after subtracting the global factor part. 
Under some local factor structure, $\bSigma_{E}$ is a form of a block diagonal matrix and its block membership represents the local factor membership. 
Thus, we can detect the group membership, based on  $\bSigma_{E}$.
Specifically, to adjust the scale problem of the convariance matrix, we calculate the correlation matrix of  $\bSigma_{E}$.
Since the sign of the correlations does not contain the membership information, we use the absolute values of the correlations. 
We denote this matrix by $\mathbf{L}$.  
We consider $\mathbf{L}$ as the adjacency matrix of the local factor network. 
Based on the adjacency matrix, many models and methodologies have been developed to detect and identify group memberships. 
Examples include RatioCut \citep{hagen1992new}, Ncut \citep{shi2000normalized}, spectral clustering method \citep{lei2013consistency, mcsherry2001spectral, rohe2011spectral}, regularized spectral clustering \citep{amini2013pseudo},  semi-definite programming \citep{cai2015robust, hajek2015achieving}, Newman-Girvan Modularity \citep{girvan2002community}, and maximum likelihood estimation \citep{amini2013pseudo, bickel2009nonparametric}.
To detect the membership matrix, we employ the regularized spectral clustering (RSC) \citep{amini2013pseudo}.

The regularized spectral clustering (RSC) is based on the regularized row and column normalized adjacency matrix (or the regularized graph Laplacian) \citep{chaudhuri2012spectral, qin2013regularized},
\begin{equation*}
\tilde{\bL}_{deg, a} = \bD _{a} ^{-1/2} \bL  \bD _{a} ^{-1/2},
\end{equation*}
where  the degree matrix is denoted by  $\bD= {\rm diag}( \hat{d}_1,\ldots, \hat{d}_p)$ with   $\hat{d}_i = \sum_{j=1}^p L_{ij}$,  $\bD _a= \bD +a \bI$,   $\bI$ denotes the identity matrix, and  $a\geq 0$ is the regularization parameter.
For the numerical study, we use the average node degree as the regularization parameter $a$. 
Then, calculate the eigenvector matrix corresponding to the  $K$ largest eigenvalue of $\tilde{\bL}_{deg, a}$. 
Using the eigenvector matrix, we apply the k-means clustering procedure and identify the local factor groups. 
Unfortunately, we cannot observe $\bSigma_{E}$, and so we estimate it using the POET procedure.
Then, using the plug-in method, we estimate the regularized row and column normalized adjacency matrix. 
Under some regularity condition, we can show that  the ratio of the misclassification goes to zero \citep{joseph2016impact, qin2013regularized}.
With this estimated membership, we can apply the proposed Double-POET.

\section{Simulation Study} \label{simulation}

In this section, we conducted simulations to examine the finite sample performance of the proposed Double-POET method. 
We considered the following multi-level factor model:
$$
y_{it} = \sum_{l=1}^{k}b_{il}G_{tl} + \sum_{s=1}^{r_{j}}\lambda_{is}^{j}f_{ts}^{j} + u_{it},
$$
where the global factors and national factors, $G_{tl}$ and $f_{ts}^{j}$, respectively,  were all drawn from $\mathcal{N}(0,1)$.
The global factor loadings  $\{b_{i1}, \dots, b_{ik}\}_{i\leq p}$ were drawn from $\mathcal{N}(\mu_{B}, I_{k})$, where each element of $\mu_{B}$ is i.i.d. Uniform$(-0.5,0.5)$;   for each $j$, the local factor loadings $\{\lambda_{i1}^{j}, \dots, \lambda_{ir_{j}}^{j}\}_{i\leq p_j}$ were drawn from $\mathcal{N}(\mu_{\lambda^{j}}, I_{r_{j}})$, where each element of $\mu_{\lambda^j}$ is i.i.d. Uniform$(-0.3,0.3)$. 

We generated the idiosyncratic errors as follows.
Let $D = \diag(d_{1}^{2},\dots,d_{p}^{2})$, where each $\{d_{i}\}$ was generated independently from Gamma $(\alpha, \beta)$ with $\alpha = \beta = 100$. 
We set $s = (s_{1},\dots, s_{p})'$ to be a sparse vector, where each $s_{i}$ was drawn from $\mathcal{N}(0,1)$ with probability $\frac{m}{\sqrt{p}\log{p}}$, and $s_{i} = 0$ otherwise. 
Then, we set a sparse error covariance matrix as $\Sigma_u = D + ss' - \diag\{s_{1}^{2},\dots,s_{p}^{2}\}$. 
In the simulation, we generated $\Sigma_u$ until it is positive definite.
 Note that varying $m>0$ enables us to control  the sparsity level, and we chose $m = 0.3$.  
 Finally, we generated $\{u_{t}\}_{t\leq T}$ from i.i.d. $\mathcal{N}(0, \Sigma_{u})$.

In this simulation study, we chose the number of periods $T = 300$ and the numbers of factors as $k=3$ and $r_{j}=2$ for each $j$.  
Then, we considered two cases: (i) increasing $p$ from 60 to 600 in increments of 30 with a fixed $J=10$ (i.e., each $p_j = p/10$), and (ii) increasing $J$ from 2 to 20 with a fixed $p_j=30$.
Each simulation is replicated 200 times.

\begin{figure}
	\includegraphics[width=\linewidth]{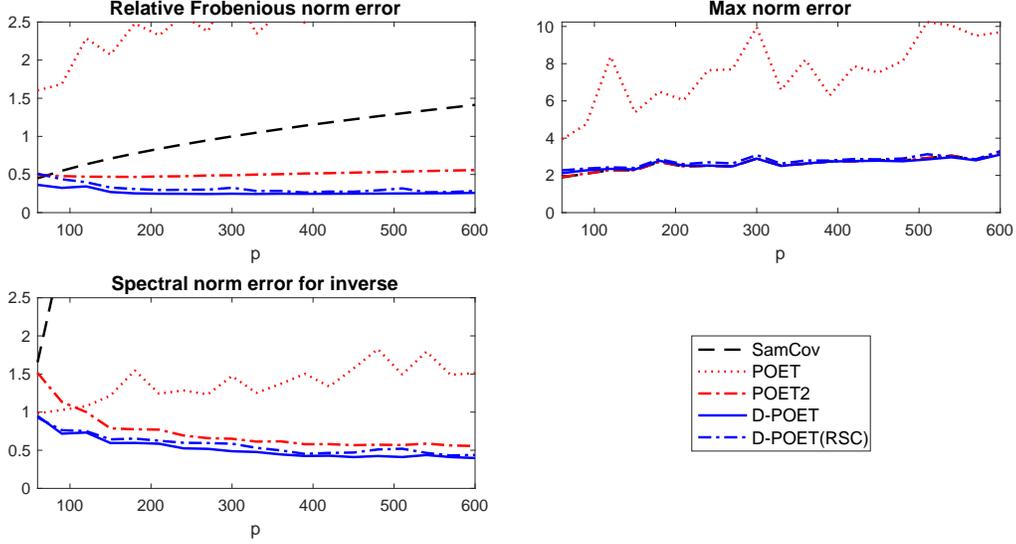}
	\centering	
	\caption{Averages of $\|\widehat{\bSigma} -\bSigma\|_{\Sigma}$, $\|\widehat{\bSigma} -\bSigma\|_{\max}$ and  $\|(\widehat{\bSigma})^{-1} - \bSigma^{-1}\|$ for Double-POET, Double-POET(RSC), POET, POET2 and the sample covariance matrix  against $p$ with a fixed $J=10$.}				\label{global_p}
\end{figure}

For comparison, we employed Double-POET, POET, and the sample covariance matrix (SamCov) to estimate the true covariance matrix of $y$, $\bSigma$. 
The estimation errors are measured in the following norms: $\|\widehat{\bSigma} -\bSigma\|_{\Sigma}$, $\|\widehat{\bSigma} -\bSigma\|_{\max}$, and $\|(\widehat{\bSigma})^{-1} - \bSigma^{-1}\|$, where $\widehat{\bSigma}$ is one of the covariance matrix estimators. 
For the POET estimation, we estimated the covariance matrix with two different numbers of factors: (i) POET uses the $k$ number of factors, and (ii) POET2 uses $k + r$ factors, where $r = J\times r_{j}$.  
For Double-POET, we considered two different cases: (i) D-POET with the known group membership,  and (ii) D-POET(RSC) suggested in Section \ref{RSC} with the unknown group membership.
The proposed numbers of global and local factors estimation method in Section \ref{sec number of global factors} is applied with $k_{\max}=10+r$ and $r_{j,\max}=10$ for each estimation. 
In addition, we employed the soft thresholding scheme for both POET and Double-POET.

\begin{figure}
	\includegraphics[width=\linewidth]{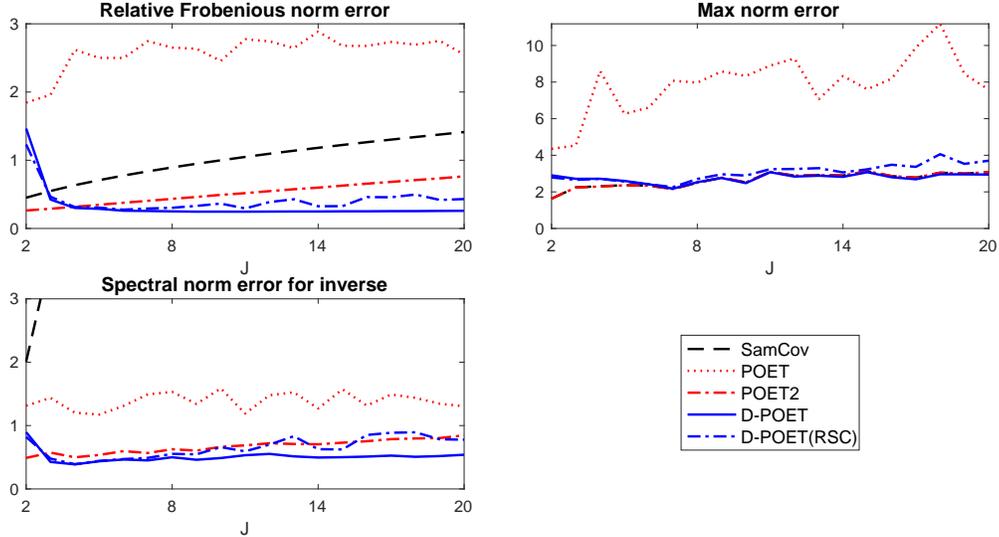}
	\centering	
	\caption{Averages of $\|\widehat{\bSigma} -\bSigma\|_{\Sigma}$, $\|\widehat{\bSigma} -\bSigma\|_{\max}$ and  $\|(\widehat{\bSigma})^{-1} - \bSigma^{-1}\|$ for Double-POET, Double-POET(RSC), POET, POET2 and the sample covariance matrix  against $J$ with a fixed $p_j=30$.}				\label{global_J}
\end{figure}

Figures \ref{global_p} and \ref{global_J} plot the averages of $\|\widehat{\bSigma} -\bSigma\|_{\Sigma}$, $\|\widehat{\bSigma} -\bSigma\|_{\max}$ and  $\|(\widehat{\bSigma})^{-1} - \bSigma^{-1}\|$ from different methods against $p$ and $J$, respectively.
From Figure \ref{global_p}, we find that D-POET performs the best. 
When comparing the POET-type procedures, POET2 performs better than POET.
This may be because POET ignores important local factors.
D-POET(RSC) with the unknown group membership performs better than POET2 under different norms.
This confirms that the RSC method in Section \ref{RSC} can detect the membership well.
Under the max norm, all estimators except POET perform roughly the same.
This is because the thresholding or imposing the local factor structure affects mainly the elements of the covariance matrix that are nearly zero, and the elementwise norm depicts the magnitude of the largest elementwise absolute error.
Figure \ref{global_J} shows the similar results except when $J$ is extremely small.
This may be because for the small $J$ ($J=2$), we have  weak global factors rahter than local factors.  
We note that the estimation errors of POET2 increase as $J$ grows because the estimator includes more noises on the off-diagonal blocks of the local factor component.
The above results support the theoretical findings presented  in Section \ref{asymp}.

We further explored the performance of Double-POET when the group membership is misclassified.  
In Figure \ref{misclassification}, we report the average errors under different norms against the misclassified rate of the group membership for Double-POET(Mix) with $J = 20$ and $p_j= 20$.
Figure \ref{misclassification} implies that Double-POET method performs better than POET2 unless the misclassification rate is greater than 3\%.
 
\begin{figure}
	\includegraphics[width=\linewidth]{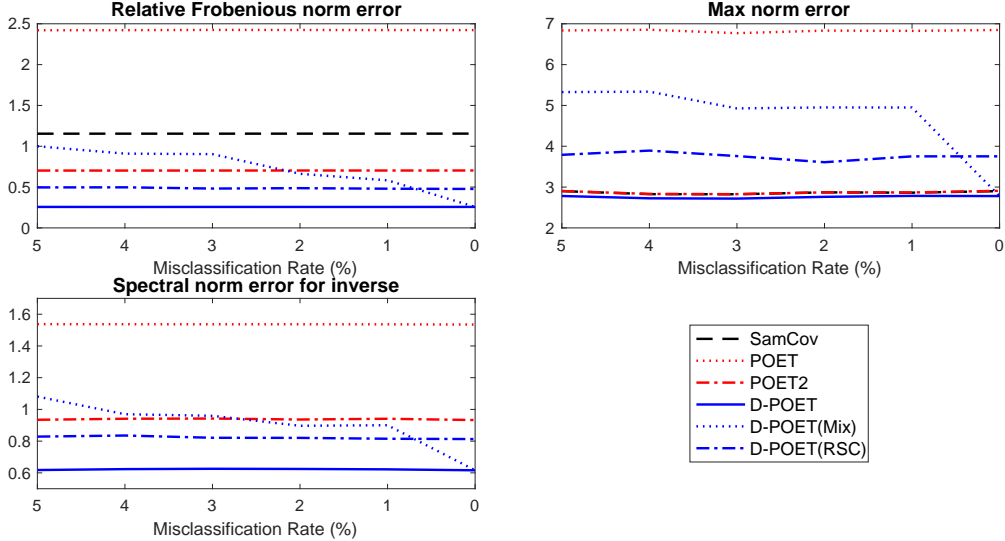}
	\centering	
	\caption{Averages of $\|\widehat{\bSigma} -\bSigma\|_{\Sigma}$, $\|\widehat{\bSigma} -\bSigma\|_{\max}$ and  $\|(\widehat{\bSigma})^{-1} - \bSigma^{-1}\|$ for Double-POET, Double-POET(RSC), Double-POET(Mix), POET, POET2 and the sample covariance matrix  against the misclassification rate with fixed $J=20$ and $p_j=20$.}				\label{misclassification}
\end{figure}

We now demonstrate the blessing of dimensionality using Double-POET when estimating the local covariance matrix $\bSigma^{j}$ and its inverse. 
We generated the data as above. 
The SamCov and POET estimators are obtained using the group sample (i.e., $T \times p_j$ observation matrix). 
For the POET method, we used the number of factors as $k + r_j =  5$ in each estimation.  
The Double-POET and POET2 estimators for $\bSigma^{j}$ are the $j$th diagonal block of $\widehat{\bSigma}^{\mathcal{D}}$ and $\widehat{\bSigma}_{2}^{\mathcal{T}}$, respectively.
We calculated $\|\widehat{\bSigma}^{j} - \bSigma^{j}\|_{ \Sigma^{j}}$ and $\|(\widehat{\bSigma}^{j})^{-1} - (\bSigma^{j})^{-1}\|$ for $j=1$. 
Note that we do not present the results of max norm, $\|\widehat{\bSigma}^{j} - \bSigma^{j}\|_{ \max}$,  as all estimators perform very similar  to that shown in  Figure \ref{global_p}.

\begin{figure}
	\includegraphics[width=\linewidth]{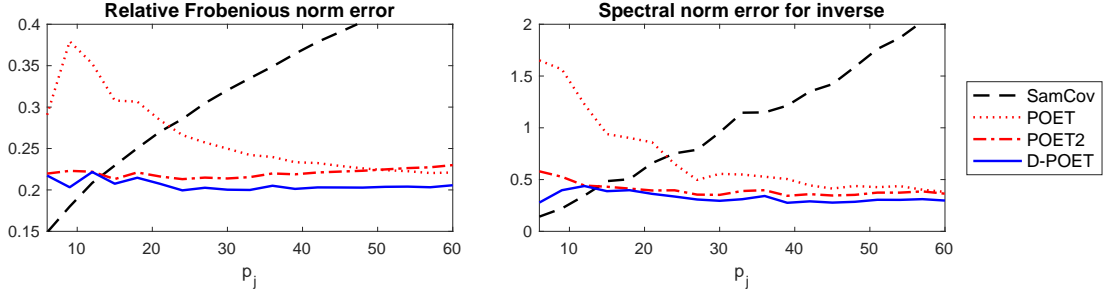}
	\centering
	\caption{Averages  of $\|\widehat{\bSigma}^{j} - \bSigma^{j}\|_{ \Sigma^{j}}$ and $\|(\widehat{\bSigma}^{j})^{-1} - (\bSigma^{j})^{-1}\|$ for  Double-POET, POET, POET2 and the sample covariance matrix  against $p_j$ with a fixed $J=10$.}				\label{blessing_pj}
\end{figure}

 Figure \ref{blessing_pj} depicts the averages  of $\|\widehat{\bSigma}^{j} - \bSigma^{j}\|_{ \Sigma^{j}}$ and $\|(\widehat{\bSigma}^{j})^{-1} - (\bSigma^{j})^{-1}\|$ for Double-POET, POET, and the sample covariance matrix  against $p_j$ with a fixed $J=10$, while Figure \ref{blessing_G} plots their average errors against  $J$ with fixed $p_j=30$.
 Figures \ref{blessing_pj} and \ref{blessing_G} show that Double-POET has smaller estimation errors than other methods under different norms.
 In Figure \ref{blessing_pj}, Double-POET significantly outperforms POET when $p_j$ is small, while the estimation error gap between Double-POET and POET decreases as $p_j$ grows.
 The estimation error of POET2 under the relative Frobenius norm tends to increase as $p_j$ grows.
 This is because, even though the global factor component can be estimated more accurately, the estimator includes redundant information on the local factor part.
 In contrast, Figure \ref{blessing_G} shows that, except when $J$ is very small, the Double-POET method constantly dominates the other methods.   
 This might be because by using other local groups' information, Double-POET can estimate global factors better than POET, especially when $p_j$ is small. 
 That is, the proposed Double-POET enjoys the blessing of dimensionality, which corresponds to the theoretical analysis discussed in Section \ref{blessinng of Dim}.

\begin{figure}
	\includegraphics[width=\linewidth]{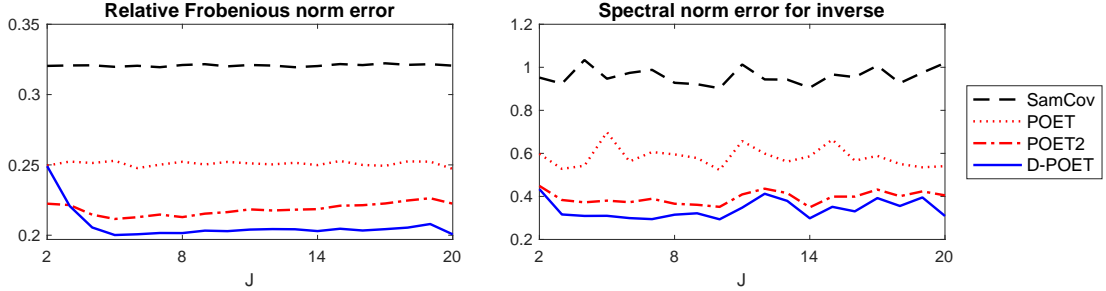}
	\centering
	\caption{Averages of $\|\widehat{\bSigma}^{j} - \bSigma^{j}\|_{ \Sigma^{j}}$ and $\|(\widehat{\bSigma}^{j})^{-1} - (\bSigma^{j})^{-1}\|$ for Double-POET, POET,  POET2 and the sample covariance matrix  against  $J$ with a fixed $p_j=30$.}				\label{blessing_G}
\end{figure}

\begin{table}[t]
	\caption{\label{k_estimation}Average estimates of the number of global factors.}
	\begin{center}
		\begin{tabular}{lllllll}
			\toprule
			$p$   & $T$      & MER            & BIC$_3$         & ED          & ER         & ABC        \\ \midrule
			&   & \multicolumn{5}{c}{$k = 0$}      \\
			\cline{3-7}
			100 & 150 & 0.26 (0.91) & 3.69 (0.00) & 0.31 (0.82) & 2.59 (0.00) & 0.62 (0.64) \\
			200 & 150 & 0.00 (1.00) & 3.22 (0.00) & 0.05 (0.96) & 3.39 (0.00) & 0.94 (0.55) \\
			300 & 150 & 0.00 (1.00) & 4.20 (0.00) & 0.09 (0.93) & 3.78 (0.00) & 0.90 (0.55) \\
			100 & 300 & 0.05 (0.98) & 2.71 (0.00) & 0.13 (0.93) & 3.41 (0.00) & 0.21 (0.83) \\
			200 & 300 & 0.00 (1.00) & 5.20 (0.00) & 0.02 (0.98) & 3.40 (0.00) & 0.55 (0.68) \\
			300 & 300 & 0.00 (1.00) & 7.44 (0.00) & 0.22 (0.79) & 2.35 (0.00) & 0.42 (0.75) \\ \midrule
			&   & \multicolumn{5}{c}{$k = 3$}      \\
			\cline{3-7}
			100 & 150 & 2.09 (0.50) & 3.56 (0.47) & 3.03 (0.98) & 3.00 (1.00) & 3.16 (0.86) \\
			200 & 150 & 2.99 (1.00) & 4.31 (0.02) & 3.11 (0.90) & 3.00 (1.00) & 3.12 (0.90) \\
			300 & 150 & 3.00 (1.00) & 4.62 (0.03) & 3.01 (0.99) & 3.00 (1.00) & 3.20 (0.84) \\
			100 & 300 & 2.92 (0.95) & 3.29 (0.71) & 3.01 (0.99) & 3.00 (1.00) & 3.03 (0.98) \\
			200 & 300 & 3.00 (1.00) & 5.53 (0.00) & 3.02 (0.98) & 3.00 (1.00) & 3.06 (0.95) \\
			300 & 300 & 3.00 (1.00) & 6.83 (0.00) & 3.03 (0.97) & 3.00 (1.00) & 3.11 (0.91) \\ \midrule
			&   & \multicolumn{5}{c}{$k = 6$}      \\
			\cline{3-7}
			100 & 150 & 5.53 (0.77) & 6.03 (0.97) & 4.67 (0.76) & 5.75 (0.85) & 6.14 (0.87) \\
			200 & 150 & 5.96 (0.99) & 6.63 (0.41) & 6.00 (1.00) & 6.00 (1.00) & 6.08 (0.93) \\
			300 & 150 & 6.00 (1.00) & 6.21 (0.79) & 6.02 (0.98) & 6.00 (1.00) & 6.03 (0.97) \\
			100 & 300 & 6.00 (1.00) & 6.00 (1.00) & 6.05 (0.98) & 6.00 (1.00) & 6.01 (0.99) \\
			200 & 300 & 6.00 (1.00) & 6.46 (0.55) & 6.00 (1.00) & 6.00 (1.00) & 6.02 (0.98) \\
			300 & 300 & 6.00 (1.00) & 8.04 (0.00) & 6.00 (1.00) & 6.00 (1.00) & 6.01 (0.99)\\
			\bottomrule
		\end{tabular}
	\end{center}
	\footnotesize
	\renewcommand{\baselineskip}{11pt}
	\textbf{Note:} Each entry depicts the average of $\hat{k}$ over 500 replications with the number in parentheses indicating the percentage correctly estimating $k$, the number of global factors. We set $k_{\max} = 10 + r$ and $\varphi_{p} = 0.3\log p$ for MER estimator, and $k_{\max} = 10$ for the other estimators. The number of group is fixed at $J = 10$.
\end{table}

	We also examined the performance of the modified eigenvalue ratio (MER) estimator introduced in Section \ref{sec number of global factors} for the number of global factors. 
	For comparison, we employed alternative estimators: the BIC$_3$ estimator of \cite{bai2002determining}, the ED estimator of \cite{onatski2010determining}, the ER estimator of \cite{ahn2013eigenvalue}, and the estimator of \cite{alessi2010improved} (ABC).
	We used the same data generating process as before but considered different number of global factors, $k \in \{0, 3, 6\}$, and sample sizes, $p \in \{100, 200, 300\}$ and $T \in \{150, 300\}$.
	We set $k_{\max} = 10 + r$ and $\varphi_{p} = 0.3\log p$ for MER estimator, and $k_{\max} = 10$ for the other estimators.
	We calculated the average of $\hat{k}$ and the percentage correctly determining $k$ over 500 replications.

The results are reported in Table \ref{k_estimation}.
We find that the proposed MER method performs the best across the different number of global factors.
In particular, when $k=0$, MER and ED tend to estimate correctly, while other estimators tend to overestimate $k$.
We note that ER does not include the case of $k=0$.
When $k>0$, all estimators except BIC$_3$ perform well, but MER and ER slightly outperform ED and ABC.
MER often underestimates $k$ if both $p$ and $T$ are small, but it selects $k$ correctly as $p$ grows.
The results demonstrate that the proposed MER method is appropriate for the model selection problem for both the multi-level factor model and the single-level factor model.

\section{Application to Portfolio Allocation} \label{empiric}
In this section, we applied the proposed Double-POET method to a minimum variance portfolio allocation study using global stock data. 
We collected daily transaction prices of international stock markets over 20 countries by the total market capitalization from January 2, 2016, to December 31, 2021. 
We selected top 100 firms at most for each country based on the market cap and used weekly log-returns to mitigate the effect of different trading hours. 
We excluded stocks with missing returns and no variation in this period.
 This operation leads to total $1892$ stocks for this period. 
 The distribution of our sample is presented in Table \ref{firm numbers}.

\begin{table}
	\centering
	\caption{Data Distribution}   \label{firm numbers}
	\begin{tabular}{ll|ll}
		\toprule
		Country (code)& The number of firms & Country (code)& The number of firms \\ \midrule
		Australia (AU) & 100 & Japan (JP) & 100\\
		Brazil (BR) & 95 & 	Netherlands (NL) & 65\\	
		Canada (CA) & 100 & 	Singapore (SG) & 100\\
		China (CN) & 77 &		South Africa (ZA) & 100\\
		France (FR) & 100 &South Korea (KR) & 100\\	
		Germany (DE) & 100  & Spain (ES) & 61\\
		Hong Kong (HK) & 100 & Switzerland (CH) &100\\
		India (IN) & 100 & Thailand (TH) & 100\\
		Indonesia (ID) & 94 & United Kingdom (GB) & 100\\
		Italy (IT) & 100 & United States (US) & 100\\ \midrule
		&& Total & 1892\\
		\bottomrule
	\end{tabular}
\end{table}

We calculated the Double-POET, POET, and SamCov estimators for each month. In both Double-POET and POET procedures, we estimated the idiosyncratic volatility matrix  based on 11 Global Industrial Classification Standard (GICS) sectors  \citep{ait2017using, fan2016incorporating}. 
For example, the idiosyncratic components for the different sectors were set  to zero, and we maintained these for the same sector. 
This location-based thresholding  preserves positive definiteness and corresponds to the hard-thresholding scheme with sector information. 
We varied the number of (global) factors $k$ from 1 to 5 for both Double-POET and POET.  
To determine the number of local factors for Double-POET, we used the eigenvalue ratio method suggested by \cite{ahn2013eigenvalue} with $r_{j, \max} = 10$.
We also considered POET2 using $k + \sum_{j=1}^{20}\hat{r}_j$ number of factors from the best performing Double-POET estimator.

To analyze the out-of-sample portfolio allocation performance, we considered the following constrained minimum variance portfolio allocation problem \citep{fan2012vast, jagannathan2003risk}:
\begin{equation*}\label{min problem}
	\min_{\omega} \omega^{T}\widehat{\bSigma}\omega, \text{ subject to } \omega^{\top}\mathbf{1} = 1, \; \|\omega\|_{1}\leq c,
\end{equation*}
where $\mathbf{1} = (1,\dots,1)^{\top} \in \mathbb{R}^{p}$, the gross exposure constraint $c$ was varied from 1 to 4, and $\widehat{\bSigma}$ is one of the volatility matrix estimators from Double-POET, POET, and SamCov. 
We constructed the portfolio at the beginning of each month, based on the stock weights calculated using the data from the past 24 months ($T=104$). 
We then held the portfolio for one month and calculated the square root of the realized volatility using the weekly portfolio log-returns. Their average was used for the out-of-sample risk. 
We considered five out-of-sample periods: 2018, 2019, 2020, 2021 and the whole period (2018-2021). 

\begin{figure}[t]
	\includegraphics[width=\linewidth]{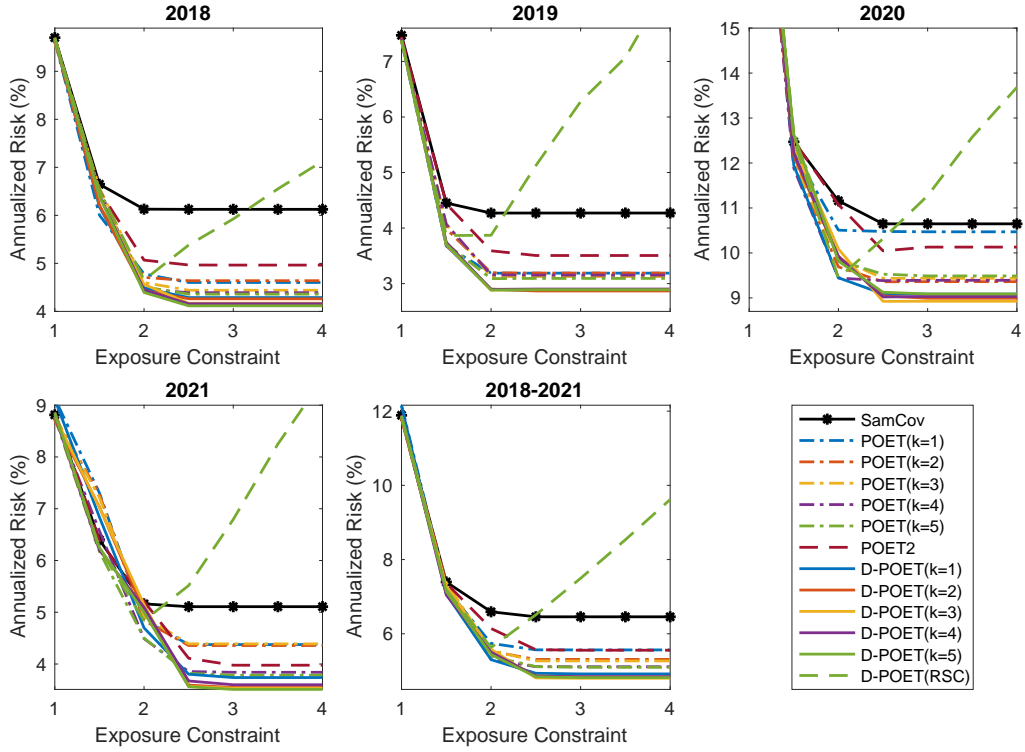}
	\centering
	\caption{Out-of-sample risks of the optimal portfolios constructed by the SamCov, POET, POET2, Double-POET and Double-POET(RSC) estimators for the global stock market.}				\label{global_2017_2020}
\end{figure}

Figure \ref{global_2017_2020} depicts the out-of-sample risks of the portfolios constructed by SamCov, POET, POET2, Double-POET, and Double-POET(RSC) with $k=5$ against the exposure constraint. 
From Figure \ref{global_2017_2020}, we find that Double-POET outperforms POET for the same number of factors $k$, while $k=3$ and $k=5$ yields the best performances for Double-POET and POET, respectively.  
Double-POET($k=3$) reduces the minimum risks by 4.3\%--6.7\% compared to POET($k=5$). 
We confirmed that for the purpose of portfolio allocation based on the global stock market, the Double-POET based on the global and national factor model outperforms POET based on the single-level factor model.
Thus, we can conjecture that incorporating the latent global and national factor models helps account for the global market dynamics.
We note that Double-POET(RSC) does not perform well due to misclassified group membership.
One possible explanation is that  there might be other type of local factors in addition to the nation-specific factors. 
For example, if the industrial risk and the national risk are nested on the idiosyncratic part after removing the global factors, the suggested RSC method cannot properly detect only the national group membership.
This is an interesting research topic to control the unknown nested local groups, so we leave it for a future study.

\begin{figure}[t]
	\includegraphics[width=\linewidth]{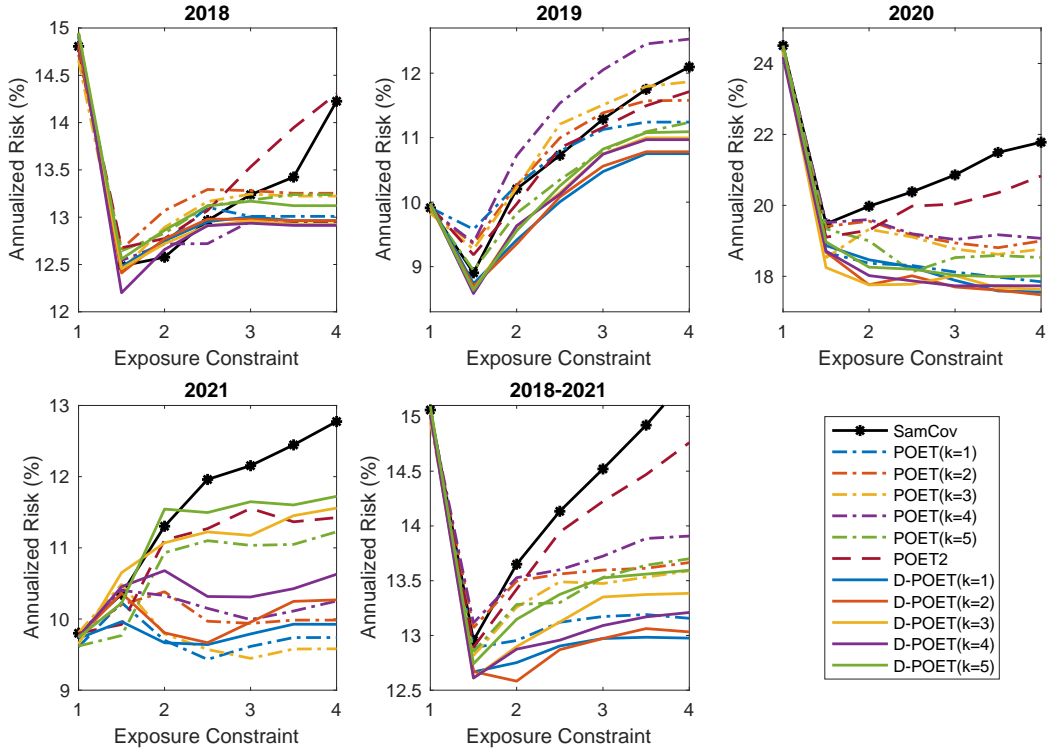}
	\centering
	\caption{Out-of-sample risks of the optimal portfolios constructed by the SamCov, POET, POET2 and Double-POET estimators for the US stock market.}				\label{US_2017_2020}
\end{figure}

We also conducted the country-wise volatility matrix estimation based on SamCov, POET, POET2, and Double-POET and applied them to the same portfolio allocation problem as before. Specifically, we used the sample of each country for POET and SamCov procedures. 
The Double-POET and POET2 estimators for each local covariance matrix are obtained by extracting each diagonal block of the large Double-POET and POET2 estimators, respectively.
Figure \ref{US_2017_2020} presents the portfolio behavior of the top 100 stocks in the US.
Double-POET shows stable results and reduces the minimum risks by 1.8\%--4.0\% compared to POET except the period 2021. 
We note that a powerful bull market lasted in 2021, and the risks could be sufficiently explained by only the first principal component (i.e., the market factor), so that, POET($k = 1$) performs the best in this period.
Nevertheless, the overall results indicate that the proposed Double-POET method enjoys the blessing of dimensionality.

\begin{figure}
	\includegraphics[width=\linewidth]{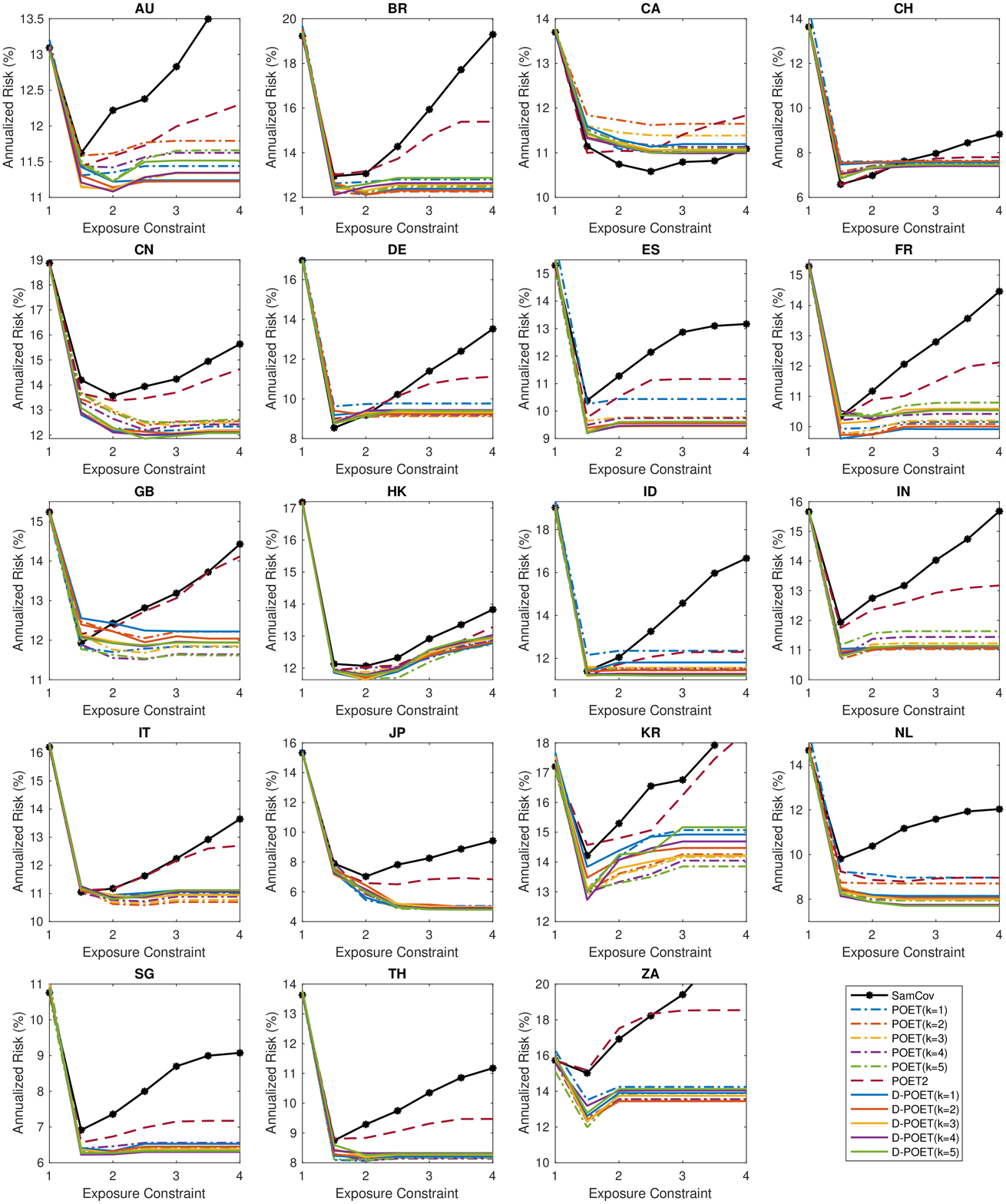}
	\caption{Out-of-sample risks of the optimal portfolios constructed by the SamCov, POET, POET2 and Double-POET estimators for each country's stock market. The out-of-sample period is from 2018 to 2021.}	\label{Countries_2017_2020}
\end{figure}

Figure \ref{Countries_2017_2020} shows the results of other 19 countries for the whole period.
Except for five countries (GB, IN, IT, TH, and ZA), Double-POET outperforms POET. 
We note that the SamCov estimator sometimes outperforms others for a few countries. 
Overall, these results indicate that Double-POET can accurately estimate the global factors by harnessing other countries' observations.

\section{Conclusion} \label{conclusion}
This paper proposes a novel large volatility matrix inference procedure based on the latent global and national factor models. 
We show the asymptotic behaviors of the proposed Double-POET method  and discuss its blessing of dimensionality and efficiency of estimating a large volatility matrix compared to the regular POET procedure. 
To determine the number of global factors, we extend the eigenvalue ratio procedure  \citep{ahn2013eigenvalue}.
In addition, when the membership of the local factors is unknown, we suggest the regularized spectral clustering method to find the latent local structure.

In the empirical study, in terms of portfolio allocation, the proposed estimator shows the best performance. It confirms the presence of the national factor structure in global financial markets, which provides the theoretical basis for employing the Double-POET method.
 In addition, for the country-wise covariance matrix estimation,  the Double-POET procedure yields the blessing of dimensionality by accurately estimating the latent global factors using the information outside the local group.
 
In this paper, we focus on the national risk factor as only the local factor. 
 However, in practice, there could be other types of risk factors that are nested in the local level.
 Thus, it is interesting and important to develop a large volatility matrix estimation procedure based on unknown-membership local factor models with the nested local-level group factors.
 We leave this for a future study.

\section*{Acknowledgments}

The authors thank the Editor Professor Torben Andersen, the Associate Editor, and  two referees for their careful reading of this paper and valuable comments. 
The research of Donggyu Kim was supported in part by the National Research Foundation of Korea (NRF) (2021R1C1C1003216). 
 
\bibliography{Double_POET}

\begin{thebibliography}{63}
\newcommand{\enquote}[1]{``#1''}
\expandafter\ifx\csname natexlab\endcsname\relax\def\natexlab#1{#1}\fi

\bibitem[\protect\citeauthoryear{Ahn and Horenstein}{Ahn and
  Horenstein}{2013}]{ahn2013eigenvalue}
\textsc{Ahn, S.~C. and A.~R. Horenstein} (2013): \enquote{Eigenvalue ratio test
  for the number of factors,} \emph{Econometrica}, 81, 1203--1227.

\bibitem[\protect\citeauthoryear{Ait-Sahalia and Xiu}{Ait-Sahalia and
  Xiu}{2017}]{ait2017using}
\textsc{Ait-Sahalia, Y. and D.~Xiu} (2017): \enquote{Using principal component
  analysis to estimate a high dimensional factor model with high-frequency
  data,} \emph{Journal of Econometrics}, 201, 384--399.

\bibitem[\protect\citeauthoryear{Alessi, Barigozzi, and Capasso}{Alessi
  et~al.}{2010}]{alessi2010improved}
\textsc{Alessi, L., M.~Barigozzi, and M.~Capasso} (2010): \enquote{Improved
  penalization for determining the number of factors in approximate factor
  models,} \emph{Statistics \& Probability Letters}, 80, 1806--1813.

\bibitem[\protect\citeauthoryear{Amini, Chen, Bickel, and Levina}{Amini
  et~al.}{2013}]{amini2013pseudo}
\textsc{Amini, A.~A., A.~Chen, P.~J. Bickel, and E.~Levina} (2013):
  \enquote{Pseudo-likelihood methods for community detection in large sparse
  networks,} \emph{The Annals of Statistics}, 41, 2097--2122.

\bibitem[\protect\citeauthoryear{Ando and Bai}{Ando and
  Bai}{2015}]{ando2015asset}
\textsc{Ando, T. and J.~Bai} (2015): \enquote{Asset pricing with a general
  multifactor structure,} \emph{Journal of Financial Econometrics}, 13,
  556--604.

\bibitem[\protect\citeauthoryear{Ando and Bai}{Ando and
  Bai}{2016}]{ando2016panel}
---\hspace{-.1pt}---\hspace{-.1pt}--- (2016): \enquote{Panel data models with
  grouped factor structure under unknown group membership,} \emph{Journal of
  Applied Econometrics}, 31, 163--191.

\bibitem[\protect\citeauthoryear{Ando and Bai}{Ando and
  Bai}{2017}]{ando2017clustering}
---\hspace{-.1pt}---\hspace{-.1pt}--- (2017): \enquote{Clustering huge number
  of financial time series: A panel data approach with high-dimensional
  predictors and factor structures,} \emph{Journal of the American Statistical
  Association}, 112, 1182--1198.

\bibitem[\protect\citeauthoryear{Aruoba, Diebold, Kose, and Terrones}{Aruoba
  et~al.}{2011}]{aruoba2011globalization}
\textsc{Aruoba, S.~B., F.~X. Diebold, M.~A. Kose, and M.~E. Terrones} (2011):
  \enquote{Globalization, the business cycle, and macroeconomic monitoring,} in
  \emph{NBER international seminar on macroeconomics}, University of Chicago
  Press Chicago, IL, vol.~7, 245--286.

\bibitem[\protect\citeauthoryear{Bai}{Bai}{2003}]{bai2003inferential}
\textsc{Bai, J.} (2003): \enquote{Inferential theory for factor models of large
  dimensions,} \emph{Econometrica}, 71, 135--171.

\bibitem[\protect\citeauthoryear{Bai and Ng}{Bai and
  Ng}{2002}]{bai2002determining}
\textsc{Bai, J. and S.~Ng} (2002): \enquote{Determining the number of factors
  in approximate factor models,} \emph{Econometrica}, 70, 191--221.

\bibitem[\protect\citeauthoryear{Bai and Wang}{Bai and
  Wang}{2015}]{bai2015identification}
\textsc{Bai, J. and P.~Wang} (2015): \enquote{Identification and Bayesian
  estimation of dynamic factor models,} \emph{Journal of Business \& Economic
  Statistics}, 33, 221--240.

\bibitem[\protect\citeauthoryear{Bernanke, Boivin, and Eliasz}{Bernanke
  et~al.}{2005}]{bernanke2005measuring}
\textsc{Bernanke, B.~S., J.~Boivin, and P.~Eliasz} (2005): \enquote{Measuring
  the effects of monetary policy: a factor-augmented vector autoregressive
  (FAVAR) approach,} \emph{The Quarterly journal of economics}, 120, 387--422.

\bibitem[\protect\citeauthoryear{Bickel and Chen}{Bickel and
  Chen}{2009}]{bickel2009nonparametric}
\textsc{Bickel, P.~J. and A.~Chen} (2009): \enquote{A nonparametric view of
  network models and Newman--Girvan and other modularities,} \emph{Proceedings
  of the National Academy of Sciences}, 106, 21068--21073.

\bibitem[\protect\citeauthoryear{Bickel and Levina}{Bickel and
  Levina}{2008}]{bickel2008covariance}
\textsc{Bickel, P.~J. and E.~Levina} (2008): \enquote{Covariance regularization
  by thresholding,} \emph{The Annals of Statistics}, 36, 2577--2604.

\bibitem[\protect\citeauthoryear{Boivin and Ng}{Boivin and
  Ng}{2006}]{boivin2006more}
\textsc{Boivin, J. and S.~Ng} (2006): \enquote{Are more data always better for
  factor analysis?} \emph{Journal of Econometrics}, 132, 169--194.

\bibitem[\protect\citeauthoryear{Cai and Liu}{Cai and
  Liu}{2011}]{cai2011adaptive}
\textsc{Cai, T. and W.~Liu} (2011): \enquote{Adaptive thresholding for sparse
  covariance matrix estimation,} \emph{Journal of the American Statistical
  Association}, 106, 672--684.

\bibitem[\protect\citeauthoryear{Cai and Li}{Cai and Li}{2015}]{cai2015robust}
\textsc{Cai, T.~T. and X.~Li} (2015): \enquote{Robust and computationally
  feasible community detection in the presence of arbitrary outlier nodes,}
  \emph{The Annals of Statistics}, 43, 1027--1059.

\bibitem[\protect\citeauthoryear{Cai, Ma, and Wu}{Cai
  et~al.}{2013}]{cai2013sparse}
\textsc{Cai, T.~T., Z.~Ma, and Y.~Wu} (2013): \enquote{Sparse PCA: Optimal
  rates and adaptive estimation,} \emph{The Annals of Statistics}, 41,
  3074--3110.

\bibitem[\protect\citeauthoryear{Cand{\`e}s and Recht}{Cand{\`e}s and
  Recht}{2009}]{candes2009exact}
\textsc{Cand{\`e}s, E.~J. and B.~Recht} (2009): \enquote{Exact matrix
  completion via convex optimization,} \emph{Foundations of Computational
  mathematics}, 9, 717--772.

\bibitem[\protect\citeauthoryear{Chamberlain and Rothschild}{Chamberlain and
  Rothschild}{1983}]{Chamberlain1983}
\textsc{Chamberlain, G. and M.~Rothschild} (1983): \enquote{Arbitrage, Factor
  Structure, and Mean-Variance Analysis on Large Asset Markets,}
  \emph{Econometrica}, 51, 1281--1304.

\bibitem[\protect\citeauthoryear{Chaudhuri, Graham, and Tsiatas}{Chaudhuri
  et~al.}{2012}]{chaudhuri2012spectral}
\textsc{Chaudhuri, K., F.~C. Graham, and A.~Tsiatas} (2012): \enquote{Spectral
  Clustering of Graphs with General Degrees in the Extended Planted Partition
  Model.} in \emph{COLT}, vol.~23, 35--1.

\bibitem[\protect\citeauthoryear{Choi, Kim, Kim, and Kwark}{Choi
  et~al.}{2018}]{choi2018multilevel}
\textsc{Choi, I., D.~Kim, Y.~J. Kim, and N.-S. Kwark} (2018): \enquote{A
  multilevel factor model: Identification, asymptotic theory and applications,}
  \emph{Journal of Applied Econometrics}, 33, 355--377.

\bibitem[\protect\citeauthoryear{Fan, Fan, and Lv}{Fan
  et~al.}{2008}]{fan2008high}
\textsc{Fan, J., Y.~Fan, and J.~Lv} (2008): \enquote{High dimensional
  covariance matrix estimation using a factor model,} \emph{Journal of
  Econometrics}, 147, 186--197.

\bibitem[\protect\citeauthoryear{Fan, Furger, and Xiu}{Fan
  et~al.}{2016}]{fan2016incorporating}
\textsc{Fan, J., A.~Furger, and D.~Xiu} (2016): \enquote{Incorporating global
  industrial classification standard into portfolio allocation: A simple
  factor-based large covariance matrix estimator with high-frequency data,}
  \emph{Journal of Business \& Economic Statistics}, 34, 489--503.

\bibitem[\protect\citeauthoryear{Fan and Kim}{Fan and
  Kim}{2018}]{fan2018robust}
\textsc{Fan, J. and D.~Kim} (2018): \enquote{Robust high-dimensional volatility
  matrix estimation for high-frequency factor model,} \emph{Journal of the
  American Statistical Association}, 113, 1268--1283.

\bibitem[\protect\citeauthoryear{Fan and Kim}{Fan and
  Kim}{2019}]{fan2019structured}
---\hspace{-.1pt}---\hspace{-.1pt}--- (2019): \enquote{Structured volatility
  matrix estimation for non-synchronized high-frequency financial data,}
  \emph{Journal of Econometrics}, 209, 61--78.

\bibitem[\protect\citeauthoryear{Fan, Li, and Wang}{Fan
  et~al.}{2017}]{fan2017estimation}
\textsc{Fan, J., Q.~Li, and Y.~Wang} (2017): \enquote{Estimation of high
  dimensional mean regression in the absence of symmetry and light tail
  assumptions,} \emph{Journal of the Royal Statistical Society: Series B
  (Statistical Methodology)}, 79, 247--265.

\bibitem[\protect\citeauthoryear{Fan, Liao, and Mincheva}{Fan
  et~al.}{2011}]{fan2011high}
\textsc{Fan, J., Y.~Liao, and M.~Mincheva} (2011): \enquote{High dimensional
  covariance matrix estimation in approximate factor models,} \emph{The Annals
  of Statistics}, 39, 3320.

\bibitem[\protect\citeauthoryear{Fan, Liao, and Mincheva}{Fan
  et~al.}{2013}]{fan2013large}
---\hspace{-.1pt}---\hspace{-.1pt}--- (2013): \enquote{Large covariance
  estimation by thresholding principal orthogonal complements,} \emph{Journal
  of the Royal Statistical Society. Series B, Statistical methodology}, 75.

\bibitem[\protect\citeauthoryear{Fan, Liu, and Wang}{Fan
  et~al.}{2018{\natexlab{a}}}]{fan2018large}
\textsc{Fan, J., H.~Liu, and W.~Wang} (2018{\natexlab{a}}): \enquote{Large
  covariance estimation through elliptical factor models,} \emph{The Annals of
  Statistics}, 46, 1383.

\bibitem[\protect\citeauthoryear{Fan, Wang, and Zhong}{Fan
  et~al.}{2018{\natexlab{b}}}]{fan2018eigenvector}
\textsc{Fan, J., W.~Wang, and Y.~Zhong} (2018{\natexlab{b}}): \enquote{An
  $\infty$ eigenvector perturbation bound and its application to robust
  covariance estimation,} \emph{Journal of Machine Learning Research}, 18,
  1--42.

\bibitem[\protect\citeauthoryear{Fan, Wang, and Zhu}{Fan
  et~al.}{2021}]{fan2021shrinkage}
\textsc{Fan, J., W.~Wang, and Z.~Zhu} (2021): \enquote{A shrinkage principle
  for heavy-tailed data: High-dimensional robust low-rank matrix recovery,}
  \emph{The Annals of Statistics}, 49, 1239.

\bibitem[\protect\citeauthoryear{Fan, Zhang, and Yu}{Fan
  et~al.}{2012}]{fan2012vast}
\textsc{Fan, J., J.~Zhang, and K.~Yu} (2012): \enquote{Vast portfolio selection
  with gross-exposure constraints,} \emph{Journal of the American Statistical
  Association}, 107, 592--606.

\bibitem[\protect\citeauthoryear{Giglio and Xiu}{Giglio and
  Xiu}{2021}]{giglio2021asset}
\textsc{Giglio, S. and D.~Xiu} (2021): \enquote{Asset pricing with omitted
  factors,} \emph{Journal of Political Economy}, 129, 1947--1990.

\bibitem[\protect\citeauthoryear{Girvan and Newman}{Girvan and
  Newman}{2002}]{girvan2002community}
\textsc{Girvan, M. and M.~E. Newman} (2002): \enquote{Community structure in
  social and biological networks,} \emph{Proceedings of the national academy of
  sciences}, 99, 7821--7826.

\bibitem[\protect\citeauthoryear{Gregory and Head}{Gregory and
  Head}{1999}]{gregory1999common}
\textsc{Gregory, A.~W. and A.~C. Head} (1999): \enquote{Common and
  country-specific fluctuations in productivity, investment, and the current
  account,} \emph{Journal of Monetary Economics}, 44, 423--451.

\bibitem[\protect\citeauthoryear{Hagen and Kahng}{Hagen and
  Kahng}{1992}]{hagen1992new}
\textsc{Hagen, L. and A.~B. Kahng} (1992): \enquote{New spectral methods for
  ratio cut partitioning and clustering,} \emph{IEEE transactions on
  computer-aided design of integrated circuits and systems}, 11, 1074--1085.

\bibitem[\protect\citeauthoryear{Hajek, Wu, and Xu}{Hajek
  et~al.}{2015}]{hajek2015achieving}
\textsc{Hajek, B., Y.~Wu, and J.~Xu} (2015): \enquote{Achieving exact cluster
  recovery threshold via semidefinite programming: Extensions,} \emph{arXiv
  preprint arXiv:1502.07738}.

\bibitem[\protect\citeauthoryear{Hallin and Li{\v{s}}ka}{Hallin and
  Li{\v{s}}ka}{2011}]{hallin2011dynamic}
\textsc{Hallin, M. and R.~Li{\v{s}}ka} (2011): \enquote{Dynamic factors in the
  presence of blocks,} \emph{Journal of Econometrics}, 163, 29--41.

\bibitem[\protect\citeauthoryear{Han}{Han}{2021}]{han2021shrinkage}
\textsc{Han, X.} (2021): \enquote{Shrinkage estimation of factor models with
  global and group-specific factors,} \emph{Journal of Business \& Economic
  Statistics}, 39, 1--17.

\bibitem[\protect\citeauthoryear{Jagannathan and Ma}{Jagannathan and
  Ma}{2003}]{jagannathan2003risk}
\textsc{Jagannathan, R. and T.~Ma} (2003): \enquote{Risk reduction in large
  portfolios: Why imposing the wrong constraints helps,} \emph{The Journal of
  Finance}, 58, 1651--1683.

\bibitem[\protect\citeauthoryear{Johnstone and Lu}{Johnstone and
  Lu}{2009}]{johnstone2009consistency}
\textsc{Johnstone, I.~M. and A.~Y. Lu} (2009): \enquote{On consistency and
  sparsity for principal components analysis in high dimensions,} \emph{Journal
  of the American Statistical Association}, 104, 682--693.

\bibitem[\protect\citeauthoryear{Joseph, Yu et~al.}{Joseph
  et~al.}{2016}]{joseph2016impact}
\textsc{Joseph, A., B.~Yu, et~al.} (2016): \enquote{Impact of regularization on
  spectral clustering,} \emph{The Annals of Statistics}, 44, 1765--1791.

\bibitem[\protect\citeauthoryear{Jung, Kim, and Yu}{Jung
  et~al.}{2022}]{jung2022next}
\textsc{Jung, K., D.~Kim, and S.~Yu} (2022): \enquote{Next generation models
  for portfolio risk management: An approach using financial big data,}
  \emph{Journal of Risk and Insurance}, 89, 765--787.

\bibitem[\protect\citeauthoryear{Kim and Fan}{Kim and
  Fan}{2019}]{kim2019factor}
\textsc{Kim, D. and J.~Fan} (2019): \enquote{Factor GARCH-It{\^o} models for
  high-frequency data with application to large volatility matrix prediction,}
  \emph{Journal of Econometrics}, 208, 395--417.

\bibitem[\protect\citeauthoryear{Kose, Otrok, and Whiteman}{Kose
  et~al.}{2003}]{kose2003international}
\textsc{Kose, M.~A., C.~Otrok, and C.~H. Whiteman} (2003):
  \enquote{International business cycles: World, region, and country-specific
  factors,} \emph{American Economic Review}, 93, 1216--1239.

\bibitem[\protect\citeauthoryear{Lam and Yao}{Lam and
  Yao}{2012}]{lam2012factor}
\textsc{Lam, C. and Q.~Yao} (2012): \enquote{Factor modeling for
  high-dimensional time series: inference for the number of factors,} \emph{The
  Annals of Statistics}, 694--726.

\bibitem[\protect\citeauthoryear{Lei and Rinaldo}{Lei and
  Rinaldo}{2013}]{lei2013consistency}
\textsc{Lei, J. and A.~Rinaldo} (2013): \enquote{Consistency of spectral
  clustering in sparse stochastic block models,} \emph{arXiv preprint
  arxiv:1312.2050}.

\bibitem[\protect\citeauthoryear{Ma}{Ma}{2013}]{ma2013sparse}
\textsc{Ma, Z.} (2013): \enquote{Sparse principal component analysis and
  iterative thresholding,} \emph{The Annals of Statistics}, 41, 772--801.

\bibitem[\protect\citeauthoryear{McSherry}{McSherry}{2001}]{mcsherry2001spectral}
\textsc{McSherry, F.} (2001): \enquote{Spectral partitioning of random graphs,}
  in \emph{Foundations of Computer Science, 2001. Proceedings. 42nd IEEE
  Symposium on}, IEEE, 529--537.

\bibitem[\protect\citeauthoryear{Moench, Ng, and Potter}{Moench
  et~al.}{2013}]{moench2013dynamic}
\textsc{Moench, E., S.~Ng, and S.~Potter} (2013): \enquote{Dynamic hierarchical
  factor models,} \emph{Review of Economics and Statistics}, 95, 1811--1817.

\bibitem[\protect\citeauthoryear{Negahban and Wainwright}{Negahban and
  Wainwright}{2011}]{negahban2011estimation}
\textsc{Negahban, S. and M.~J. Wainwright} (2011): \enquote{Estimation of
  (near) low-rank matrices with noise and high-dimensional scaling,} \emph{The
  Annals of Statistics}, 39, 1069--1097.

\bibitem[\protect\citeauthoryear{Onatski}{Onatski}{2010}]{onatski2010determining}
\textsc{Onatski, A.} (2010): \enquote{Determining the number of factors from
  empirical distribution of eigenvalues,} \emph{The Review of Economics and
  Statistics}, 92, 1004--1016.

\bibitem[\protect\citeauthoryear{Qin and Rohe}{Qin and
  Rohe}{2013}]{qin2013regularized}
\textsc{Qin, T. and K.~Rohe} (2013): \enquote{Regularized spectral clustering
  under the degree-corrected stochastic blockmodel,} in \emph{Advances in
  Neural Information Processing Systems}, 3120--3128.

\bibitem[\protect\citeauthoryear{Rohe, Chatterjee, and Yu}{Rohe
  et~al.}{2011}]{rohe2011spectral}
\textsc{Rohe, K., S.~Chatterjee, and B.~Yu} (2011): \enquote{Spectral
  clustering and the high-dimensional stochastic blockmodel,} \emph{The Annals
  of Statistics}, 1878--1915.

\bibitem[\protect\citeauthoryear{Rothman, Levina, and Zhu}{Rothman
  et~al.}{2009}]{rothman2009generalized}
\textsc{Rothman, A.~J., E.~Levina, and J.~Zhu} (2009): \enquote{Generalized
  thresholding of large covariance matrices,} \emph{Journal of the American
  Statistical Association}, 104, 177--186.

\bibitem[\protect\citeauthoryear{Shi and Malik}{Shi and
  Malik}{2000}]{shi2000normalized}
\textsc{Shi, J. and J.~Malik} (2000): \enquote{Normalized cuts and image
  segmentation,} \emph{IEEE Transactions on pattern analysis and machine
  intelligence}, 22, 888--905.

\bibitem[\protect\citeauthoryear{Shin, Kim, and Fan}{Shin
  et~al.}{2021}]{shin2021adaptive}
\textsc{Shin, M., D.~Kim, and J.~Fan} (2021): \enquote{Adaptive robust large
  volatility matrix estimation based on high-frequency financial data,}
  \emph{arXiv preprint arXiv:2102.12752}.

\bibitem[\protect\citeauthoryear{Stein and James}{Stein and
  James}{1961}]{stein1961estimation}
\textsc{Stein, C. and W.~James} (1961): \enquote{Estimation with quadratic
  loss,} in \emph{Proc. 4th Berkeley Symp. Mathematical Statistics
  Probability}, vol.~1, 361--379.

\bibitem[\protect\citeauthoryear{Stock and Watson}{Stock and
  Watson}{2002}]{stock2002forecasting}
\textsc{Stock, J.~H. and M.~W. Watson} (2002): \enquote{Forecasting using
  principal components from a large number of predictors,} \emph{Journal of the
  American statistical association}, 97, 1167--1179.

\bibitem[\protect\citeauthoryear{Trapani}{Trapani}{2018}]{trapani2018randomized}
\textsc{Trapani, L.} (2018): \enquote{A randomized sequential procedure to
  determine the number of factors,} \emph{Journal of the American Statistical
  Association}, 113, 1341--1349.

\bibitem[\protect\citeauthoryear{Vershynin}{Vershynin}{2010}]{vershynin2010introduction}
\textsc{Vershynin, R.} (2010): \enquote{Introduction to the non-asymptotic
  analysis of random matrices,} \emph{arXiv preprint arXiv:1011.3027}.

\bibitem[\protect\citeauthoryear{Wang and Fan}{Wang and
  Fan}{2017}]{wang2017asymptotics}
\textsc{Wang, W. and J.~Fan} (2017): \enquote{Asymptotics of empirical
  eigenstructure for high dimensional spiked covariance,} \emph{The Annals of
  Statistics}, 45, 1342.

\end{thebibliography}

\newpage

\appendix

	\section{Appendix} \label{proofs}

	\subsection{Proof of Theorem \ref{thm1}} \label{proof of Thm1}
	We first provide useful lemmas below. Let $\{\delta_i, v_i\}_{i=1}^{p}$ be the eigenvalues and their corresponding eigenvectors of $\bSigma$ in decreasing order. 
	Let $\{\bar{\delta}_i, \bar{v}_i\}_{i=1}^{k}$ be the leading eigenvalues and eigenvectors of $\bB\bB'$ and the rest zero. 
	Similarly, for each group $j$, let $\{\kappa_i^{j}, \eta_i^{j}\}_{i=1}^{p_j}$ be the eigenvalues and eigenvectors of $\bSigma_{E}^{j}$ in decreasing order, and $\{\bar{\kappa}_i^{j}, \bar{\eta}_i^{j}\}_{i=1}^{r_j}$ for $\Lambda^{j}{\Lambda^{j}}'$. 
	
	By Weyl's theorem, we have the following lemma under the pervasive conditions.
	\begin{lem}\label{weyl's}
		Under Assumption \ref{assum1}(i), we have
		$$
		|\delta_{i} - \bar{\delta}_{i}| \leq \|\bSigma_{E}\| \text{ for } i \leq k, \;\;\;\;\; |\delta_{i}| \leq \|\bSigma_{E}\|  \text{ for } i > k,	
		$$
		and, for $i \leq k$,
		$\bar{\delta}_{i}/p^{a_{1}}$  is strictly bigger than zero for all $p$. In addition, for each group $j$, we have
		$$
		|\kappa_{i}^{j} - \bar{\kappa}_{i}^{j}| \leq \|\bSigma_{u}^{j}\|  \text{ for } i \leq r_j, \;\;\;\;\; |\kappa_{i}^{j}| \leq \|\bSigma_{u}^{j}\|  \text{ for } i > r_j,	
		$$
		and, for $i \leq r_j$, 
		$\bar{\kappa}_i^{j}/p_{j}^{a_{2}}$  is strictly bigger than  zero for all $p_j$.
	\end{lem}

	The following lemma  presents the individual convergence rate of leading eigenvectors using Lemma \ref{weyl's} and the $l_{\infty}$ norm perturbation bound theorem of \cite{fan2018eigenvector}. 
	
	\begin{lem}\label{individual}
		Under Assumption \ref{assum1}(i), we have the following results.
		\begin{itemize}
			\item[(i)] 
			We have		
			$$
			\max_{i\leq k} \|\bar{v}_{i}-v_{i}\|_{\infty} \leq C \frac{\|\bSigma_{E}\|_{\infty}}{p^{3(a_{1} -\frac{1}{2} )}}.
			$$
			\item[(ii)] For each group $j$, we have
			$$
			\max_{i\leq r_j} \|\bar{\eta}_{i}^{j}-\eta_{i}^{j}\|_{\infty} \leq C \frac{\|\bSigma_{u}^{j}\|_{\infty}}{p_j^{3(a_{2} -\frac{1}{2} )}}.
			$$
		\end{itemize}
	\end{lem}
	\begin{proof}
		(i) Let $\bB = (\tilde{b}_{1},\dots, \tilde{b}_{k})$. Then, for $i\leq k$, $\bar{\delta}_{i} = \|\tilde{b}_{i}\|^{2} \asymp p^{a_{1}}$ from Lemma \ref{weyl's} and $\bar{v}_{i} = \tilde{b}_{i}/\|\tilde{b}_{i}\|$.
		Hence, $\|\bar{v}_{i}\|_{\infty} \leq \|\bB\|_{\max}/\|\tilde{b}_{i}\| \leq C/\sqrt{p^{a_{1}}}$. 
		In addition, for $\widetilde{\bV} = (\bar{v}_{1},\dots,\bar{v}_{k})$, the coherence $\mu(\widetilde{\bV}) = p\max_{i}\sum_{j=1}^{k}\widetilde{\bV}_{ij}^{2}/k \leq C p^{1-a_1}$, where $\widetilde{\bV}_{ij}$ is the $(i,j)$ entry of $\widetilde{\bV}$.
		Thus, by Theorem 1 of \cite{fan2018eigenvector}, we have
		$$
		\max_{i\leq k} \|\bar{v}_{i}-v_{i}\|_{\infty} \leq C p^{2(1-a_1)} \frac{\|\bSigma_{E}\|_{\infty}}{\bar{\gamma} \sqrt{p}},
		$$
		where the eigengap $\bar{\gamma} = \min\{\bar{\delta}_{i} - \bar{\delta}_{i+1} : 1 \leq i \leq k\}$ and  $\delta_{k+1} = 0$.				
		By the similar argument, we can show the result (ii).
	\end{proof}

	\begin{lem}\label{global_pro}
		Under Assumption \ref{assum1}, for $i\leq k$, we have
		\begin{align*}
		&|\widehat{\delta}_{i}/\delta_{i} -1 | = O_{P}(p^{1-a_{1}}\sqrt{\log p /T}),\\
		&\|\widehat{v}_{i}-v_{i}\|_{\infty} = O_{P}\left(p^{3(1-a_{1})}\sqrt{\frac{\log p}{Tp}}+\cfrac{\|\bSigma_{E}\|_{\infty}}{p^{3(a_{1}-\frac{1}{2})}}\right).
		\end{align*}
	\end{lem}
	\begin{proof}
		The first statement is followed by Assumption \ref{assum1} and Weyl's theorem.		
		We have
		$$
		\widehat{\bSigma} = \bB\bB' + \bLambda\bLambda' + \bSigma_{u} + (\widehat{\bSigma}-\bSigma) = \bB\bB' + \bSigma_{E} + (\widehat{\bSigma}-\bSigma).
		$$
		We can treat $\bB\bB'$ as a low rank matrix and the remaining terms as a perturbation matrix.
		By Theorem 1 of \cite{fan2018eigenvector}, Lemma \ref{individual} and Assumption \ref{assum1}, we have
		\begin{align*}
		\|\widehat{v}_{i}-v_{i}\|_{\infty}  &\leq C p^{2(1-a_1)} \frac{\| \bSigma_{E} + (\widehat{\bSigma}-\bSigma)\|_{\infty}}{p^{a_{1}}\sqrt{p}}
		\leq Cp^{2(1-a_1)} \frac{\|\bSigma_{E}\|_{\infty}}{p^{a_{1}}\sqrt{p}} + Cp^{2(1-a_1)} \frac{\|\widehat{\bSigma}-\bSigma\|_{\max}}{p^{a_{1}-1}\sqrt{p}}\\
		&= O_{P}\left(\cfrac{\|\bSigma_{E}\|_{\infty}}{p^{3(a_{1}-\frac{1}{2}) } }+p^{3(1-a_{1})}\sqrt{\frac{\log p}{Tp}}\right).
		\end{align*}		
	\end{proof}

	\begin{lem}\label{local_pro}
		Under the assumptions of Theorem \ref{thm1}, for $i\leq r_j$, we have 
		\begin{align*}
		&|\widehat{\kappa}_{i}^{j}/\kappa_{i}^{j} -1 | = O_{P} \left (p^{\frac{5}{2}(1-a_{1})+c(1-a_{2})}\sqrt{\log p/T} + 1/p^{\frac{5a_{1}}{2}-\frac{3}{2}-2c +c a_{2}}\right ),\\
		&\|\widehat{\eta}_{i}^{j}-\eta_{i}^{j}\|_{\infty} = O_{P}\left (p^{\frac{5}{2}(1-a_{1})+ c(\frac{5}{2}-3a_{2})}\sqrt{\frac{\log p}{T}}  +\cfrac{1}{p^{\frac{5 a_{1}}{2} -\frac{3}{2} +c(3a_2-\frac{7}{2})}}  +\cfrac{m_{p}}{p^{3c(a_{2}-\frac{1}{2})}}\right).
		\end{align*}
	\end{lem}
	\begin{proof}
		We have
		\begin{align*}
		\|\bSigma_{E}\|  \leq   \|\bLambda\bLambda'\|  + \|\bSigma_{u}\|  \leq  \|\bLambda\bLambda'\| + O(m_{p}) = O(p^{c a_2 }).
		\end{align*}
		Let $\bB\bB' = \widetilde{\bV}\widetilde{\bGamma}\widetilde{\bV}'$, where $\widetilde{\bGamma} = \diag(\bar{\delta}_{1},\dots, \bar{\delta}_{k})$ and their corresponding leading $k$ eigenvectors $\widetilde{\bV}=(\bar{v}_{1},\dots,\bar{v}_{k})$.
		Also, we let $\bGamma = \diag(\delta_1,\dots, \delta_k)$ and the corresponding eigenvectors $\bV = (v_1,\dots, v_k)$ of covariance matrix $\bSigma$. Note that $\|\bB\|_{\max} = \|\widetilde{\bV}\widetilde{\bGamma}^{1/2}\|_{\max}=O(1)$ and $\|\bSigma_{E}\|_{\infty} = O(p^{c})$. 
		By Lemmas \ref{weyl's}-\ref{individual}, we have
		\begin{align*}
		\|\bV\bGamma^{\frac{1}{2}} - \widetilde{\bV}\widetilde{\bGamma}^{{\frac{1}{2}}}\|_{\max} &\leq \|\bB\widetilde{\bGamma}^{-{\frac{1}{2}}}(\bGamma^{{\frac{1}{2}}}-\widetilde{\bGamma}^{{\frac{1}{2}}})\|_{\max} + \|(\bV-\widetilde{\bV})\bGamma^{{\frac{1}{2}}}\|_{\max} \\
		& \leq C\cfrac{\|\bSigma_{E}\|}{p^{a_{1}}} + C\cfrac{\|\bSigma_{E}\|_{\infty}}{\sqrt{p^{5a_{1}-3}}} = o \left(1\right).	
		\end{align*}
		Hence,  we have $\|\bV\bGamma^{\frac{1}{2}}\|_{\max} = O(1)$ and $\|\bV\|_{\max} = O(1/\sqrt{p^{a_{1}}})$. 
		By this fact and the results from Lemmas \ref{weyl's}-\ref{global_pro}, we have 
		\begin{align*}
		&\|\widetilde{\bV}\widetilde{\bGamma}\widetilde{\bV}' - \bV\bGamma\bV'\|_{\max} \leq\|\widetilde{\bV}(\widetilde{\bGamma}-\bGamma)\widetilde{\bV}'\|_{\max} + \|(\widetilde{\bV}-\bV)\bGamma(\widetilde{\bV}-\bV)'\|_{\max} + 2\|\bV\bGamma(\widetilde{\bV}-\bV)'\|_{\max}\\
		&\;\;\;\;\; = O(p^{-a_{1}}\|\widetilde{\bGamma}-\bGamma\|_{\max} + \sqrt{p^{a_{1}}}\|\widetilde{\bV}-\bV\|_{\max}) = O( 1/p^{\frac{5a_{1}}{2}-\frac{3}{2}-c}),\\
		&\|\widehat{\bV}\widehat{\bGamma}\widehat{\bV}' - \bV\bGamma\bV'\|_{\max}  \leq\|\widehat{\bV}(\widehat{\bGamma}-\bGamma)\widehat{\bV}'\|_{\max} + \|(\widehat{\bV}-\bV)\bGamma(\widehat{\bV}-\bV)'\|_{\max} + 2\|\bV\bGamma(\widehat{\bV}-\bV)'\|_{\max}\\
		& \;\;\;\;\; = O_{P}(p^{-a_{1}}\|\widehat{\bGamma}-\bGamma\|_{\max} + \sqrt{p^{a_{1}}}\|\widehat{\bV}-\bV\|_{\max}) = O_{P}(p^{\frac{5}{2}(1-a_{1})}\sqrt{\log p/T} + 1/p^{\frac{5a_{1}}{2}-\frac{3}{2}-c}).
		\end{align*}
		Thus, we have
		\begin{equation}
		\|\widehat{\bV}\widehat{\bGamma}\widehat{\bV}' - \bB\bB'\|_{\max} = O_{P}(p^{\frac{5}{2}(1-a_{1})}\sqrt{\log p/T} + 1/p^{\frac{5a_{1}}{2}-\frac{3}{2}-c}). \label{global_ind}
		\end{equation}		
		Then, we have 
		\begin{eqnarray}
		\|\widehat{\bSigma}_{E}-\bSigma_{E}\|_{\max} &\leq& \|\widehat{\bSigma}-\bSigma\|_{\max} + \|\widehat{\bV}\widehat{\bGamma}\widehat{\bV}' - \bB\bB'\|_{\max}\cr
		& =& O_{P}(p^{\frac{5}{2}(1-a_{1})}\sqrt{\log p/T} + 1/p^{\frac{5a_{1}}{2}-\frac{3}{2}-c}). \label{sigmaE_max}
		\end{eqnarray}
		Therefore, the first statement is followed by \eqref{sigmaE_max} and the Weyl's theorem.

		We decompose the sample covariance matrix $\widehat{\bSigma}_{E}^{j}$ for each group $j$ as follows:
		$$
		\widehat{\bSigma}_{E}^{j} = \Lambda^{j}\Lambda^{j\prime} + \bSigma_{u}^{j} + (\widehat{\bSigma}_{E}^{j}-\bSigma_{E}^{j}).
		$$ 
		Then, by Theorem 1 of \cite{fan2018eigenvector}, Lemma \ref{individual} and (\ref{sigmaE_max}), we have
		\begin{align*}
		\|\widehat{\eta}_{i}^{j}-\eta_{i}^{j}\|_{\infty} &\leq C p_j ^{2 (1-a_2)}\frac{\|\bSigma_{u}^{j} + (\widehat{\bSigma}_{E}^{j}-\bSigma_{E}^{j})\|_{\infty}}{p_j^{a_{2}}\sqrt{p_j}}\\
		&\leq Cp_j ^{2 (1-a_2)} \frac{\|\bSigma_{u}^{j}\| _{\infty} }{p_j^{a_{2}}\sqrt{p_j}} + Cp_j ^{2 (1-a_2)} \frac{\|\widehat{\bSigma}_{E}^{j}-\bSigma_{E}^{j}\|_{\max}}{p_{j}^{a_{2}-1}\sqrt{p_{j}}}\\
		&=O_{P}\left(p^{\frac{5}{2}(1-a_{1})+ c(\frac{5}{2}-3a_{2})}\sqrt{\frac{\log p}{T}}  +\cfrac{1}{p^{\frac{5 a_{1}}{2} -\frac{3}{2} +c(3a_2-\frac{7}{2})}}  +\cfrac{m_{p}}{p^{3c(a_{2}-\frac{1}{2})}}\right).
		\end{align*}
	\end{proof}

	\textbf{Proof of Theorem \ref{thm1}.}
	We first consider \eqref{u_max}. 
	By definition, $\|\widehat{\bSigma}_{u}^{\mathcal{D}}-\widehat{\bSigma}_{u}\|_{\max} = \max_{ij}|s_{ij}(\widehat{\sigma}_{ij})-\widehat{\sigma}_{ij}| \leq \max_{ij}\tau_{ij} = O_{P}(\tau)$. Hence, it suffices to show $\|\widehat{\bSigma}_{u}-\bSigma_{u}\|_{\max} = O_{P}(\omega_{T})$. 
	We have
	$$
	\|\widehat{\bSigma}_{u}-\bSigma_{u}\|_{\max} \leq \|\widehat{\bSigma}-\bSigma\|_{\max} + \|\widehat{\bV}\widehat{\bGamma}\widehat{\bV}' - \bB\bB'\|_{\max} + \|\widehat{\bPhi}\widehat{\bPsi}\widehat{\bPhi}' - \bLambda\bLambda'\|_{\max}.
	$$
	By Assumption \ref{assum1}(iii) and \eqref{global_ind}, we have
	\begin{equation}\label{local_result1}
	\|\widehat{\bSigma}-\bSigma\|_{\max} + \|\widehat{\bV}\widehat{\bGamma}\widehat{\bV}' - \bB\bB'\|_{\max} = O_{P}(p^{\frac{5}{2}(1-a_{1})}\sqrt{\log p/T} + 1/p^{\frac{5a_{1}}{2}-\frac{3}{2}-c}).
	\end{equation}
	For $ \|\widehat{\bPhi}\widehat{\bPsi}\widehat{\bPhi}' - \bLambda\bLambda'\|_{\max}$, we have
	$$
	\|\widehat{\bPhi}\widehat{\bPsi}\widehat{\bPhi}' - \bLambda\bLambda'\|_{\max} = \max_{j} \|\widehat{\Phi}^{j}\widehat{\Psi}^{j}\widehat{\Phi}^{j \prime} - \Lambda^{j}{\Lambda^{j}}'\|_{\max}.
	$$
	For each group $j$, let $\Lambda^{j}{\Lambda^{j}}' = \widetilde{\Phi}^{j}\widetilde{\Psi}^{j}\widetilde{\Phi}^{j\prime}$,  where $\widetilde{\Psi}^{j} = \diag(\bar{\kappa}_{1}^{j},\dots,\bar{\kappa}_{r_j}^{j})$ and the corresponding eigenvectors $\widetilde{\Phi}^{j} = (\bar{\eta}_{1}, \dots, \bar{\eta}_{r_j})$.
	In addition, let $\Psi^{j} = \diag(\kappa_{1}^{j},\dots,\kappa_{r_j}^{j})$ and $\Phi^{j} = (\eta_{1}, \dots, \eta_{r_j})$ to be the leading eigenvalues and the corresponding eigenvectors of $\bSigma_{E}^{j}$, respectively. 
	Then, we have
	\begin{align}
	\|\Phi^{j}{\Psi^{j}}^{\frac{1}{2}}- \widetilde{\Phi}^{j}{{}\widetilde{\Psi}^{j}}^{\frac{1}{2}}\|_{\max} &\leq \|\Lambda^{j}{{}\widetilde{\Psi}^{j}}^{-\frac{1}{2}}({\Psi^{j}}^{\frac{1}{2}}-{{}\widetilde{\Psi}^{j}}^{\frac{1}{2}})\|_{\max} + \|(\Phi^{j}-\widetilde{\Phi}^{j}){\Psi^{j}}^{\frac{1}{2}}\|_{\max} \nonumber \\ 
	&\leq \frac{\|\bSigma_{u}^{j}\|}{p_{j}^{a_{2}}} + \frac{\|\bSigma_{u}^{j}\|_{\infty}}{\sqrt{p_{j}^{5a_{2}-3}}} = o(1).  \label{local indiv 1}
	\end{align}
	Since $\|\Lambda^{j}\|_{\max} =\|\widetilde{\Phi}^{j}{{}\widetilde{\Psi}^{j}}^{\frac{1}{2}}\|_{\max}= O(1)$, $\|\Phi^{j}{\Psi^{j}}^{\frac{1}{2}}\|_{\max} = O(1)$ and $\|\Phi^{j}\|_{\max} = O(1/\sqrt{p_j^{a_{2}}})$. 
	Using this fact and results from Lemmas \ref{weyl's}, \ref{individual} and \ref{local_pro}, we can show
	\begin{align*}
	&\|\widetilde{\Phi}^{j}\widetilde{\Psi}^{j}\widetilde{\Phi}^{j \prime} - \Phi^{j}\Psi^{j}\Phi^{j \prime}\|_{\max} \leq O(p_{j}^{-a_{2}}\|\widetilde{\Psi}^{j} -\Psi^{j}\|_{\max} + \sqrt{p_{j}^{a_{2}}}\|\widetilde{\Phi}^{j} - \Phi^{j}\|_{\max}) = O(m_{p}/\sqrt{p^{c(5a_{2}-3)}}),	\\
	& \|\widehat{\Phi}^{j}\widehat{\Psi}^{j}\widehat{\Phi}^{j \prime} - \Phi^{j}\Psi^{j}\Phi^{j \prime}\|_{\max} \leq O_{P}(p_{j}^{-a_{2}}\|\widehat{\Psi}^{j} -\Psi^{j}\|_{\max} + \sqrt{p_j^{a_{2}}}\|\widehat{\Phi}^{j} - \Phi^{j}\|_{\max}) \\
	&\qquad\qquad\qquad = O_{P}(p^{\frac{5}{2}(1-a_{1})+\frac{5}{2}c(1-a_{2})}\sqrt{\log p/T}+1/p^{\frac{5}{2}a_{1}-\frac{3}{2}+c(\frac{5}{2}a_{2}-\frac{7}{2})} +m_{p}/\sqrt{p^{c(5a_{2}-3)}}).
	\end{align*}
	By using these rates,  we obtain 
	\begin{equation}
	\|\widehat{\bPhi}\widehat{\bPsi}\widehat{\bPhi}' - \bLambda\bLambda'\|_{\max} =  O_{P}(p^{\frac{5}{2}(1-a_{1})+\frac{5}{2}c(1-a_{2})}\sqrt{\log p/T}+1/p^{\frac{5}{2}a_{1}-\frac{3}{2}+c(\frac{5}{2}a_{2}-\frac{7}{2})} +m_{p}/\sqrt{p^{c(5a_{2}-3)}}).\label{local_ind}
	\end{equation}
	By \eqref{local_result1} and \eqref{local_ind}, we have
	$$
	\|\widehat{\bSigma}_{u}-\bSigma_{u}\|_{\max} = O_{P}\left(p^{\frac{5}{2}(1-a_{1})+\frac{5}{2}c(1-a_{2})}\sqrt{\frac{\log p}{T}}+\frac{1}{p^{\frac{5}{2}a_{1}-\frac{3}{2}+c(\frac{5}{2}a_{2}-\frac{7}{2})}} +\frac{m_{p}}{\sqrt{p^{c(5a_{2}-3)}}}\right).
	$$
	Therefore, $\|\widehat{\bSigma}^{\mathcal{D}}_{u}-\bSigma_{u}\|_{\max} = O_{P}(\tau + \omega_{T}) =O_{P}(\omega_{T})$, when $\tau$ is chosen as the same order of $\omega_{T} = p^{\frac{5}{2}(1-a_{1})+\frac{5}{2}c(1-a_{2})}\sqrt{\log p/T}+1/p^{\frac{5}{2}a_{1}-\frac{3}{2}+c(\frac{5}{2}a_{2}-\frac{7}{2})} +m_{p}/\sqrt{p^{c(5a_{2}-3)}}$.

	Consider (\ref{error rate}).
	Similar to the proofs of Theorem 2.1 in \cite{fan2011high}, we can show $\|\widehat{\bSigma}^{\mathcal{D}}_{u}-\bSigma_{u}\|_{2} = O_{P}(m_{p}\omega_{T}^{1-q})$.
	In addition, since $\lambda_{\min}(\bSigma_{u})>c_1$ and $m_{p}\omega_{T}^{1-q}=o(1)$, the minimum eigenvalue of $\widehat{\bSigma}_{u}^{\mathcal{D}}$ is strictly bigger than 0 with probability approaching 1.
	Then, we have $\|(\widehat{\bSigma}_{u}^{\mathcal{D}})^{-1} - \bSigma_{u}^{-1}\|_{2}\leq \lambda_{\min}(\bSigma_{u})^{-1} \|\widehat{\bSigma}_{u}^{\mathcal{D}}- \bSigma_{u}\|_{2} \lambda_{\min}(\widehat{\bSigma}_{u}^{\mathcal{D}})^{-1} = O_{P}(m_{p}\omega_{T}^{1-q})$.

	Consider (\ref{maxnorm}).
	By the results of (\ref{global_ind}), (\ref{local_ind}) and (\ref{u_max}), we have
	$$
	\|\widehat{\bSigma}^{\mathcal{D}} - \bSigma\|_{\max} \leq  \|\widehat{\bV}\widehat{\bGamma}\widehat{\bV}' - \bB\bB'\|_{\max} + \|\widehat{\bPhi}\widehat{\bPsi}\widehat{\bPhi}' - \bLambda\bLambda'\|_{\max} +\|\widehat{\bSigma}^{\mathcal{D}}_{u}-\bSigma_{u}\|_{\max} = O_{P}(\omega_{T}).
	$$

	Consider (\ref{inverse rate}). 
	Define $\widehat{\bSigma}_{E}^{\mathcal{D}} = \widehat{\bPhi}\widehat{\bPsi}\widehat{\bPhi}' +\widehat{\bSigma}_{u}^{\mathcal{D}}$. 
	We first show that $\|(\widehat{\bSigma}_{E}^{\mathcal{D}})^{-1} - \bSigma_{E}^{-1}\| = O_{P}(m_{p}\omega_{T}^{1-q})$. 
	Let $\widehat{\bJ} = \widehat{\bPsi}^{\frac{1}{2}}\widehat{\bPhi}'(\widehat{\bSigma}_{u}^{\mathcal{D}})^{-1}\widehat{\bPhi}\widehat{\bPsi}^{\frac{1}{2}}$ and  $\widetilde{\bJ} = \widetilde{\bPsi}^{\frac{1}{2}}\widetilde{\bPhi}'\bSigma_{u}^{-1}\widetilde{\bPhi}\widetilde{\bPsi}^{\frac{1}{2}}$. 
	Using the Sherman-Morrison-Woodbury formula, we have
	$$
	\|(\widehat{\bSigma}_{E}^{\mathcal{D}})^{-1} - \bSigma_{E}^{-1}\|  \leq \|(\widehat{\bSigma}_{u}^{\mathcal{D}})^{-1} - \bSigma_{u}^{-1}\| + \Delta_{1},
	$$
	where $\Delta_{1} = \|(\widehat{\bSigma}_{u}^{\mathcal{D}})^{-1}\widehat{\bPhi}\widehat{\bPsi}^{\frac{1}{2}}(\bI_{r} + \widehat{\bJ})^{-1}\widehat{\bPsi}^{\frac{1}{2}}\widehat{\bPhi}'(\widehat{\bSigma}_{u}^{\mathcal{D}})^{-1} - \bSigma_{u}^{-1}\widetilde{\bPhi}\widetilde{\bPsi}^{\frac{1}{2}}(\bI_{r} + \widetilde{\bJ})^{-1}\widetilde{\bPsi}^{\frac{1}{2}}\widetilde{\bPhi}'\bSigma_{u}^{-1}\|$.
	Then, the right hand side can be bounded by following terms:
	\begin{align*}
	& L_{1} = \|((\widehat{\bSigma}_{u}^{\mathcal{D}})^{-1}-\bSigma_{u}^{-1})\widetilde{\bPhi}\widetilde{\bPsi}^{\frac{1}{2}}(\bI_{r} + \widetilde{\bJ})^{-1}\widetilde{\bPsi}^{\frac{1}{2}}\widetilde{\bPhi}'\bSigma_{u}^{-1}\|, \\
	&L_{2} = \| \bSigma_{u}^{-1}(\widehat{\bPhi}\widehat{\bPsi}^{\frac{1}{2}}-\widetilde{\bPhi}\widetilde{\bPsi}^{\frac{1}{2}})(\bI_{r} + \widetilde{\bJ})^{-1}\widetilde{\bPsi}^{\frac{1}{2}}\widetilde{\bPhi}'\bSigma_{u}^{-1}\|, \\
	&L_{3} =\|\bSigma_{u}^{-1}\widetilde{\bPhi}\widetilde{\bPsi}^{\frac{1}{2}}((\bI_{r} + \widehat{\bJ})^{-1}-(\bI_{r} + \widetilde{\bJ})^{-1})\widetilde{\bPsi}^{\frac{1}{2}}\widetilde{\bPhi}'\bSigma_{u}^{-1}\|.
	\end{align*}
	By Lemma \ref{local_pro}, $\|\Phi^{j}{\Psi^{j}}^{\frac{1}{2}}- \widehat{\Phi}^{j}{{}\widehat{\Psi}^{j}}^{\frac{1}{2}}\|_{\max} \leq \|\Lambda^{j}{{}\widehat{\Psi}^{j}}^{-\frac{1}{2}}({\Psi^{j}}^{\frac{1}{2}}-{{}\widehat{\Psi}^{j}}^{\frac{1}{2}})\|_{\max} + \|(\Phi^{j}-\widehat{\Phi}^{j}){\Psi^{j}}^{\frac{1}{2}}\|_{\max} = O_{P}(\omega_{T})$, and by (\ref{local indiv 1}) and (\ref{error rate}), we then have
	\begin{align*}
	&\|\widetilde{\bPhi}\widetilde{\bPsi}^{\frac{1}{2}}\| \leq \max_{j}\|\widetilde{\Phi}^{j}{{}\widetilde{\Psi}^{j}}^{\frac{1}{2}}\|  = O_{P}(\sqrt{p^{c}}), \\
	&\|\widehat{\bPhi}\widehat{\bPsi}^{\frac{1}{2}}-\widetilde{\bPhi}\widetilde{\bPsi}^{\frac{1}{2}}\| \leq \max_{j}\sqrt{p^{c}}\|\widehat{\Phi}^{j}{{}\widehat{\Psi}^{j}}^{\frac{1}{2}} -\widetilde{\Phi}^{j}{{}\widetilde{\Psi}^{j}}^{\frac{1}{2}}\|_{\max} =   O_{P}\left(\sqrt{p^{c}}\omega_{T}\right),
	\end{align*}
	and
	\begin{align*}		
	\|\widehat{\bJ}-\widetilde{\bJ}\| & 
	\leq \|(\widehat{\bPsi}^{\frac{1}{2}}\widehat{\bPhi}'-\widetilde{\bPsi}^{\frac{1}{2}}\widetilde{\bPhi}')(\widehat{\bSigma}_{u}^{\mathcal{D}})^{-1}(\widehat{\bPhi}\widehat{\bPsi}^{\frac{1}{2}}-\widetilde{\bPhi}\widetilde{\bPsi}^{\frac{1}{2}})\| \\ 
	&\;\;\; + \|(\widehat{\bPsi}^{\frac{1}{2}}\widehat{\bPhi}'-\widetilde{\bPsi}^{\frac{1}{2}}\widetilde{\bPhi}')(\widehat{\bSigma}_{u}^{\mathcal{D}})^{-1}\widetilde{\bPhi}\widetilde{\bPsi}^{\frac{1}{2}}\|  
	+ \|\widetilde{\bPsi}^{\frac{1}{2}}\widetilde{\bPhi}'((\widehat{\bSigma}_{u}^{\mathcal{D}})^{-1}-\bSigma_{u}^{-1})\widetilde{\bPhi}\widetilde{\bPsi}^{\frac{1}{2}}\|\\
	& = O_{P}(p^{c}m_{p}\omega_{T}^{1-q}).  				
	\end{align*}
	Since $\lambda_{\min}(\bI_{r} + \widetilde{\bJ}) \geq \lambda_{\min}(\widetilde{\bJ})\geq \lambda_{\min}(\bSigma_{u}^{-1})\lambda_{\min}^{2}(\widetilde{\bPhi}\widetilde{\bPsi}^{\frac{1}{2}}) \geq Cp^{c}$, we have $\|(\bI_{r} + \widetilde{\bJ})^{-1}\|=O_{P}(1/p^{c})$. 
	Then, $L_{1} = O_{P}(m_{p}\omega_{T}^{1-q})$ by (\ref{error rate}).
	In addition, 
	$L_{2} = O_{P}(p^{-c/2}\|\widehat{\bPhi}\widehat{\bPsi}^{\frac{1}{2}}-\widetilde{\bPhi}\widetilde{\bPsi}^{\frac{1}{2}}\|) = O_{P}(\omega_{T})$
	and 
	$L_{3} = O_{P}(p^{c}\|(\bI_{r} + \widehat{\bJ})^{-1}-(\bI_{r} + \widetilde{\bJ})^{-1}\|) = O_{P}(p^{-c}\|\widehat{\bJ}-\widetilde{\bJ}\|) = O_{P}(m_{p}\omega_{T}^{1-q})$. Thus, we have
	\begin{equation}
	\Delta_{1} = O_{P}(m_{p}\omega_{T}^{1-q}), \label{Delta_1}
	\end{equation}								
	which yields $\|(\widehat{\bSigma}_{E}^{\mathcal{D}})^{-1} - \bSigma_{E}^{-1}\| = O_{P}(m_{p}\omega_{T}^{1-q})$.
	
	Let $\widehat{\bH} = \widehat{\bGamma}^{\frac{1}{2}}\widehat{\bV}'(\widehat{\bSigma}_{E}^{\mathcal{D}})^{-1}\widehat{\bV}\widehat{\bGamma}^{\frac{1}{2}}$ and  $\widetilde{\bH} = \widetilde{\bGamma}^{\frac{1}{2}}\widetilde{\bV}'\bSigma_{E}^{-1}\widetilde{\bV}\widetilde{\bGamma}^{\frac{1}{2}}$.
	Using the Sherman-Morrison-Woodbury formula again, we have
	$$
	\|(\widehat{\bSigma}^{\mathcal{D}})^{-1} - \bSigma^{-1}\|  \leq \|(\widehat{\bSigma}_{E}^{\mathcal{D}})^{-1} - \bSigma_{E}^{-1}\| + \Delta_{2},
	$$
	where $\Delta_{2} = \|(\widehat{\bSigma}_{E}^{\mathcal{D}})^{-1}\widehat{\bV}\widehat{\bGamma}^{\frac{1}{2}}(\bI_{k} + \widehat{\bH})^{-1}\widehat{\bGamma}^{\frac{1}{2}}\widehat{\bV}'(\widehat{\bSigma}_{E}^{\mathcal{D}})^{-1} - \bSigma_{E}^{-1}\widetilde{\bV}\widetilde{\bGamma}^{\frac{1}{2}}(\bI_{k} + \widetilde{\bH})^{-1}\widetilde{\bGamma}^{\frac{1}{2}}\widetilde{\bV}'\bSigma_{E}^{-1}\|$.
	By Weyl's inequality, we have $\lambda_{\min}(\bSigma_{E}) > c$ since $\lambda_{\min}(\bSigma_{u}) > c$ and $\lambda_{\min}(\bLambda\bLambda') = 0$. 
	Hence,  $\|\bSigma_{E}^{-1}\| = O_{P}(1)$. 
	By Lemmas \ref{weyl's}-\ref{global_pro}, we have $\|\widehat{\bV}\widehat{\bGamma}^{{\frac{1}{2}}}-\widetilde{\bV}\widetilde{\bGamma}^{{\frac{1}{2}}}\|_{\max}  =  O_{P}(p^{\frac{5}{2}(1-a_1)}\sqrt{\log p/T} + 1/p^{\frac{5}{2}a_1-\frac{3}{2}-c})$.
	Similar to the proof of (\ref{Delta_1}), we can show $\Delta_{2} = O_{P}(m_{p}\omega_{T}^{1-q})$. Therefore, we have $\|(\widehat{\bSigma}^{\mathcal{D}})^{-1} - \bSigma^{-1}\| = O_{P}(m_{p}\omega_{T}^{1-q})$.

	Consider  (\ref{entropynorm}).
	We derive the rate of convergence for  $\|\widehat{\bSigma}^{\mathcal{D}} -\bSigma\|_{\Sigma}$.
	The SVD decomposition of $\bSigma$ is 
	$$
	\bSigma = (\bV_{p \times k} \;\; \bPhi_{p \times r} \;\; \bOmega_{p \times (p-k-r)}) \begin{pmatrix}
	\bGamma_{k\times k} &  & \\
	& \bPsi_{r\times r} & \\
	& & \bTheta_{(p-k-r)\times (p-k-r)}
	\end{pmatrix}
	\begin{pmatrix}
	\bV'\\
	\bPhi' \\
	\bOmega'
	\end{pmatrix}.
	$$
	Note that $\bOmega$ is used to denote the precision matrix in Section \ref{estimation procedure}. 
	Moreover, since all the eigenvalues of $\bSigma$ are strictly bigger than 0, for any maxtrix $\bA$, we have $\|\bA\|_{\Sigma}^2 = O_{P}(p^{-1})\|\bA\|_{F}^{2}$.
	Then, we have
	\begin{align*}
	\|\widehat{\bSigma}^{\mathcal{D}} -\bSigma\|_{\Sigma}& \leq p^{-1/2}\Big(\|\bSigma^{-1/2}(\widehat{\bV}\widehat{\bGamma}\widehat{\bV}' - \bB\bB')\bSigma^{-1/2}\|_{F}\\
	& \;\;\;\;\; + \|\bSigma^{-1/2}(\widehat{\bPhi}\widehat{\bPsi}\widehat{\bPhi}' - \bLambda \bLambda')\bSigma^{-1/2}\|_{F} + \|\bSigma^{-1/2}(\widehat{\bSigma}^{\mathcal{D}}_{u}-\bSigma_{u})\bSigma^{-1/2}\|_{F}\Big)\\
	&=: \Delta_{G} + \Delta_{L} + \Delta_{S}
	\end{align*}
	and
	$$
	\Delta_{S} =O_{P}(p^{-1/2}\|\widehat{\bSigma}^{\mathcal{D}}_{u}-\bSigma_{u}\|_{F}) = O_{P}(\|\widehat{\bSigma}^{\mathcal{D}}_{u}-\bSigma_{u}\|_{2}) = O_{P}(m_{p}\omega_{T}^{1-q}).
	$$
	We have 
	\begin{align*}
	\Delta_{G} &= p^{-1/2}\left\|\begin{pmatrix}
	\bGamma^{-\frac{1}{2}}\bV'\\
	\bPsi^{-\frac{1}{2}}\bPhi' \\
	\bTheta^{-\frac{1}{2}}\bOmega'
	\end{pmatrix}
	(\widehat{\bV}\widehat{\bGamma}\widehat{\bV}' - \bB\bB') \left(\bV\bGamma^{-\frac{1}{2}} \;\;	\bPhi\bPsi^{-\frac{1}{2}} \;\; \bOmega\bTheta^{-\frac{1}{2}}\right)\right\|_{F}\\
	& \leq p^{-1/2}\big(\|\bGamma^{-1/2}\bV'(\widehat{\bV}\widehat{\bGamma}\widehat{\bV}' - \bB\bB')\bV\bGamma^{-1/2}\|_{F} + \|	\bPsi^{-1/2}\bPhi'(\widehat{\bV}\widehat{\bGamma}\widehat{\bV}' - \bB\bB')\bPhi\bPsi^{-1/2}\|_{F}\\
	&\;\;\;\; +\|\bTheta^{-1/2}\bOmega'(\widehat{\bV}\widehat{\bGamma}\widehat{\bV}' - \bB\bB')\bOmega\bTheta^{-1/2}\|_{F} + 2\|	\bGamma^{-1/2}\bV'(\widehat{\bV}\widehat{\bGamma}\widehat{\bV}' - \bB\bB')\bPhi\bPsi^{-1/2}\|_{F}\\
	&\;\;\;\;+2\|\bGamma^{-1/2}\bV'(\widehat{\bV}\widehat{\bGamma}\widehat{\bV}' - \bB\bB')\bOmega\bTheta^{-1/2}\|_{F} + 2\|	\bPsi^{-1/2}\bPhi'(\widehat{\bV}\widehat{\bGamma}\widehat{\bV}' - \bB\bB')\bOmega\bTheta^{-1/2}\|_{F}\big)\\
	& =: \Delta_{G1}+\Delta_{G2}+\Delta_{G3}+2\Delta_{G4}+2\Delta_{G5}+2\Delta_{G6}.
	\end{align*}
	In order to find the convergence rate of relative Frobenius norm, we consider the above terms separately. 
	For $\Delta_{G1}$, we have
	\begin{align*}
	\Delta_{G1} &\leq p^{-1/2}\left(\|\bGamma^{-1/2}\bV'(\widehat{\bV}\widehat{\bGamma}\widehat{\bV}' - \bV\bGamma\bV')\bV\bGamma^{-1/2}\|_{F} + \|\bGamma^{-1/2}\bV'(\bV\bGamma\bV' - \bB\bB')\bV\bGamma^{-1/2}\|_{F}\right)\\
	& =:  \Delta_{G1}^{(a)} + \Delta_{G1}^{(b)}.
	\end{align*}		
	We bound the two terms separately. We have
	\begin{align*}
	\Delta_{G1}^{(a)}  & \leq	p^{-1/2}\big(\|\bGamma^{-1/2}(\bV'\widehat{\bV}-\bI)\widehat{\bGamma}(\widehat{\bV}'\bV-\bI)\bGamma^{-1/2}\|_{F} + 2\|\bGamma^{-1/2}(\bV'\widehat{\bV}-\bI)\widehat{\bGamma}\bGamma^{-1/2}\|_{F}\\
	& \;\;\; + \|(\bGamma^{-1/2}(\widehat{\bGamma}-\bGamma)\bGamma^{-1/2}\|_{F}\big) =: I+II+III.
	\end{align*}
	By Lemma \ref{global_pro}, $\|\bV'\widehat{\bV} - \bI\|_{F} = \|\bV'(\widehat{\bV}-\bV)\|_{F} \leq \|\widehat{\bV}-\bV\|_{F}=O_{P}(p^{3(1-a_{1})}\sqrt{\log p/T}+1/p^{3a_{1}-2-c})$.
	Then, $II$ is of order $O_{P}(p^{3(1-a_{1})-\frac{1}{2}}\sqrt{\log p/T}+1/p^{3(a_{1}-\frac{1}{2})-c})$ and $I$ is of smaller order. 
	In addition, we have $III \leq \|\bGamma^{-1/2}(\widehat{\bGamma}-\bGamma)\bGamma^{-1/2}\| = O_{P}(p^{1-a_{1}}\sqrt{\log p/T})$ by Lemma \ref{global_pro}. 
	Thus, $\Delta_{G1}^{(a)} = O_{P}(p^{1-a_{1}}\sqrt{\log p/T} + 1/p^{3(a_{1}-\frac{1}{2})-c})$.
	Similarly, we have
	\begin{align*}
	\Delta_{G1}^{(b)}  & \leq	p^{-1/2}\big(\|\bGamma^{-1/2}(\bV'\widetilde{\bV}-\bI)\widetilde{\bGamma}(\widetilde{\bV}'\bV-\bI)\bGamma^{-1/2}\|_{F} + 2\|\bGamma^{-1/2}(\bV'\widetilde{\bV}-\bI)\widetilde{\bGamma}\bGamma^{-1/2}\|_{F}\\
	& \;\;\; + \|(\bGamma^{-1/2}(\widetilde{\bGamma}-\bGamma)\bGamma^{-1/2}\|_{F}\big) =: I'+II'+III'.
	\end{align*}
	By $\sin\theta$ theorem, $\|\bV'\widetilde{\bV}-\bI\| = \|\bV'(\widetilde{\bV}-\bV)\| \leq \|\widetilde{\bV}-\bV\| = O(\|\bSigma_{E}\|/p^{a_{1}})$.
	Then, we have $II' = O(1/p^{a_{1}-ca_{2}})$ and $I'$ is of smaller order. 
	By Lemma \ref{weyl's}, we have $III' = O(1/p^{a_{1}-ca_{2}})$.		
	Thus, $\Delta_{G1}^{(b)} = O(1/p^{a_{1}-ca_{2}})$.
	Then, we obtain		
	\begin{equation}
	\Delta_{G1} = O_{P}\left(p^{1-a_{1}}\sqrt{\frac{\log p}{T}} + \frac{1}{p^{3(a_{1}-\frac{1}{2})-c}} + \frac{1}{p^{a_{1}-ca_{2}}}\right). \label{delta_G1}		
	\end{equation}
	For $\Delta_{G3}$, we have
	$$
	\Delta_{G3} \leq p^{-1/2}\|		\bTheta^{-1/2}\bOmega'\widehat{\bV}\widehat{\bGamma}\widehat{\bV}'\bOmega\bTheta^{-1/2}\|_{F} + p^{-1/2}\|\bTheta^{-1/2}\bOmega'\widetilde{\bV}\widetilde{\bGamma}\widetilde{\bV}'\bOmega\bTheta^{-1/2}\|_{F} =: \Delta_{G3}^{(a)} + \Delta_{G3}^{(b)}.
	$$
	By Lemma \ref{global_pro}, we have  
	$$
	\|\bOmega'\widehat{\bV}\|_{F} = 	\|\bOmega'(\widehat{\bV}-\bV)\|_{F}=O(\sqrt{p}\|\widehat{\bV}-\bV\|_{\max}) = O_{P}(p^{3(1-a_{1})}\sqrt{\log p/T} + 1/p^{3a_{1}-2-c}).
	$$
	Since $\|\widehat{\bGamma}\| = O_{P}(p^{a_{1}})$, we have
	$$
	\Delta_{G3}^{(a)}  \leq 	p^{-1/2}\|\bTheta^{-1}\|\|\bOmega'\widehat{\bV}\|_{F}^{2}\|\widehat{\bGamma}\| = O_{P}(p^{11/2-5a_{1}}\log p/T + 1/p^{5a_{1}-7/2-2c}).
	$$
	Similarly, $ \Delta_{G3}^{(b)} = O_{P}(1/p^{5a_{1}-7/2-2c})$ because $\|\bOmega'\widetilde{\bV}\|_{F} = O(\sqrt{p}\|\widetilde{\bV}-\bV\|_{\max}) = O_{P}(1/p^{3a_{1}-2-c})$ by Lemma \ref{individual}. 
	Then, we obtain 
	$$
	\Delta_{G3} = O_{P}\left(p^{\frac{11}{2}-5a_{1}}\frac{\log p}{T} + \frac{1}{p^{5a_{1}-\frac{7}{2}-2c}}\right).
	$$ 
	Similarly, we can show that the terms $\Delta_{G2}$, $\Delta_{G4}$, $\Delta_{G5}$ and $\Delta_{G6}$ are dominated by $\Delta_{G1}$ and $\Delta_{G3}$. 
	Therefore, we have
	\begin{equation}
	\Delta_{G} = O_{P}\left(p^{1-a_{1}}\sqrt{\frac{\log p}{T}} + \frac{1}{p^{a_{1}-ca_{2}}} + p^{\frac{11}{2}-5a_{1}}\frac{\log p}{T} + \frac{1}{p^{5a_{1}-\frac{7}{2}-2c}}\right). \label{delta_G}
	\end{equation}

	Similarly, we consider 
	\begin{align*}
	\Delta_{L} &= p^{-1/2}\left\|\begin{pmatrix}
	\bGamma^{-\frac{1}{2}}\bV'\\
	\bPsi^{-\frac{1}{2}}\bPhi' \\
	\bTheta^{-\frac{1}{2}}\bOmega'
	\end{pmatrix}
	(\widehat{\bPhi}\widehat{\bPsi}\widehat{\bPhi}' - \bLambda \bLambda') \left(\bV\bGamma^{-\frac{1}{2}} \;\;	\bPhi\bPsi^{-\frac{1}{2}} \;\; \bOmega\bTheta^{-\frac{1}{2}}\right)\right\|_{F}\\
	&\leq p^{-1/2}\big(\|\bGamma^{-1/2}\bV'(\widehat{\bPhi}\widehat{\bPsi}\widehat{\bPhi}' - \bLambda \bLambda')\bV\bGamma^{-1/2}\|_{F} + \|	\bPsi^{-1/2}\bPhi'(\widehat{\bPhi}\widehat{\bPsi}\widehat{\bPhi}' - \bLambda \bLambda')\bPhi\bPsi^{-1/2}\|_{F} \\
	& \;\;\;\; + \|	\bTheta^{-1/2}\bOmega'(\widehat{\bPhi}\widehat{\bPsi}\widehat{\bPhi}' - \bLambda \bLambda')\bOmega\bTheta^{-1/2}\|_{F} + 2\|	\bGamma^{-1/2}\bV'(\widehat{\bPhi}\widehat{\bPsi}\widehat{\bPhi}' - \bLambda \bLambda')\bPhi\bPsi^{-1/2}\|_{F} \\
	& \;\;\;\; + 2\|	\bGamma^{-1/2}\bV'(\widehat{\bPhi}\widehat{\bPsi}\widehat{\bPhi}' - \bLambda \bLambda')\bOmega\bTheta^{-1/2}\|_{F} + 2\|	\bPsi^{-1/2}\bPhi'(\widehat{\bPhi}\widehat{\bPsi}\widehat{\bPhi}' - \bLambda \bLambda')\bOmega\bTheta^{-1/2}\|_{F} \big) \\
	& =: \Delta_{L1}+\Delta_{L2}+\Delta_{L3}+2\Delta_{L4}+2\Delta_{L5}+2\Delta_{L6}.
	\end{align*}
	For $\Delta_{L2}$, similar to the proof of (\ref{delta_G}), we have
	\begin{align*}
	\Delta_{L2} &\leq p^{-1/2}\left(\|\bPsi^{-1/2}\bPhi'(\widehat{\bPhi}\widehat{\bPsi}\widehat{\bPhi}' - \bPhi\bPsi\bPhi')\bPhi\bPsi^{-1/2}\|_{F} + \|\bPsi^{-1/2}\bPhi'(\bPhi\bPsi\bPhi' - \bLambda\bLambda')\bPhi\bPsi^{-1/2}\|_{F}\right)\\
	& =:  \Delta_{L2}^{(a)} + \Delta_{L2}^{(b)}.
	\end{align*}		
	We have
	\begin{align*}
	\Delta_{L2}^{(a)}  & \leq	p^{-1/2}\big(\|\bPsi^{-1/2}(\bPhi'\widehat{\bPhi}-\bI)\widehat{\bPsi}(\widehat{\bPhi}'\bPhi-\bI)\bPsi^{-1/2}\|_{F} + 2\|\bPsi^{-1/2}(\bPhi'\widehat{\bPhi}-\bI)\widehat{\bPsi}\bPsi^{-1/2}\|_{F}\\
	& \;\;\; + \|(\bPsi^{-1/2}(\widehat{\bPsi}-\bPsi)\bPsi^{-1/2}\|_{F}\big) =: I+II+III.
	\end{align*}
	By Lemma \ref{local_pro}, we have $\|\widehat{\Phi}^{j}-\Phi^{j}\|_{F} \leq \sqrt{p_{j}r_{j}}\|\widehat{\Phi}^{j}-\Phi^{j}\|_{\max} = O_{P}\big(p^{\frac{5}{2}(1-a_{1}) + 3c(1-a_{2})}\sqrt{\log p/T} +1/p^{\frac{5a_{1}}{2}-\frac{3}{2}+c(3a_{2}-4)} +m_{p}/p^{c(3a_{2}-2)}\big)$.
	Because $\widehat{\bPhi}$ and $\bPhi$ are block diagonal matrices, we have
	$$
	\|\widehat{\bPhi}-\bPhi\|_{F}^{2} = \sum_{j=1}^{G}\|\widehat{\Phi}^{j}-\Phi^{j}\|_{F}^2=O_{P}\big(p^{1-c}(p^{5(1-a_{1})+6c(1-a_{2})}\frac{\log p}{T} +\frac{1}{p^{5a_{1}-3+2c(3a_{2}-4)}} +\frac{m_{p}^2}{p^{2c(3a_{2}-2)}})\big).
	$$
	Then, $II$ is of order $O_{P}(p^{\frac{5}{2}(1-a_{1})+c(\frac{5}{2}-3a_{2})}\sqrt{\log p/T} + 1/p^{\frac{5}{2}a_{1}-\frac{3}{2} +c(3a_{2}-\frac{7}{2})} + m_{p}/p^{3c(a_{2}-\frac{1}{2})})$ and $I$ is of smaller order. 
	By Lemma \ref{local_pro}, we have $III = O_{P}(p^{\frac{5}{2}(1-a_{1})+c(1-a_{2})}\sqrt{\log p/T} + 1/p^{\frac{5}{2}a_{1}-\frac{3}{2}-2c+ca_{2}})$.
	Thus, $\Delta_{L2}^{(a)} = O_{P}(p^{\frac{5}{2}(1-a_{1})+c(1-a_{2})}\sqrt{\log p/T} + 1/p^{\frac{5}{2}a_{1}-\frac{3}{2}-2c+ca_{2}} + m_{p}/p^{3c(a_{2}-\frac{1}{2})})$.
	Similarly, we have
	\begin{align*}
	\Delta_{L2}^{(b)}  & \leq	p^{-1/2}\big(\|\bPsi^{-1/2}(\bPhi'\widetilde{\bPhi}-\bI)\widetilde{\bPsi}(\widetilde{\bPhi}'\bPhi-\bI)\bPsi^{-1/2}\|_{F} + 2\|\bPsi^{-1/2}(\bPhi'\widetilde{\bPhi}-\bI)\widetilde{\bPsi}\bPsi^{-1/2}\|_{F}\\
	& \;\;\; + \|(\bPsi^{-1/2}(\widetilde{\bPsi}-\bPsi)\bPsi^{-1/2}\|_{F}\big) =: I'+II'+III'.
	\end{align*}
	By $\sin\theta$ theorem, $\|\bPhi'\widetilde{\bPhi}-\bI\|  \leq \|\widetilde{\bPhi}-\bPhi\| \leq \max_{j}\|\widetilde{\Phi}^{j}-\Phi^{j}\| \leq O(m_{p}/p^{ca_{2}})$.
	Then, we have $II' = O(m_{p}/p^{ca_{2}})$ and $I'$ is of smaller order. 
	By Lemma \ref{weyl's}, we have $III' = O(m_{p}/p^{ca_{2}})$.		
	Thus, $\Delta_{L2}^{(b)} = O(m_{p}/p^{ca_{2}})$.
	Then, we obtain		
	\begin{equation}
	\Delta_{L2} = O_{P}\left(p^{\frac{5}{2}(1-a_{1})+c(1-a_{2})}\sqrt{\frac{\log p}{T}} + \frac{1}{p^{\frac{5}{2}a_{1}-\frac{3}{2}-2c+ca_{2}}} + \frac{m_{p}}{p^{ca_{2}}}\right). \label{delta_L2}
	\end{equation}
	For $\Delta_{L3}$, we have 
	$$
	\Delta_{L3}\leq p^{-1/2}\|	\bTheta^{-1/2}\bOmega'\widehat{\bPhi}\widehat{\bPsi}\widehat{\bPhi}'\bOmega\bTheta^{-1/2}\|_{F} + p^{-1/2}\|\bTheta^{-1/2}\bOmega'\widetilde{\bPhi}\widetilde{\bPsi}\widetilde{\bPhi}'\bOmega\bTheta^{-1/2}\|_{F} =: \Delta_{L3}^{(a)} + \Delta_{L3}^{(b)}.
	$$
	Since $\|\widehat{\bPsi}\| = O_{P}(p^{ca_{2}})$, we have
	\begin{align*}
	\Delta_{L3}^{(a)} &\leq p^{-1/2}\|\bTheta^{-1}\|\|\bOmega'(\widehat{\bPhi}-\bPhi)\|_{F}^{2}\|\widehat{\bPsi}\|\\
	& = O_{P}\left(p^{\frac{11}{2}-5a_{1}+5c(1-a_{2})}\frac{\log p}{T}+ 
	\frac{1}{p^{5a_{1}-\frac{7}{2}-c(7-5a_{2})}} +\frac{m_{p}^2}{p^{5ca_{2}-3c-\frac{1}{2}}}\right).
	\end{align*}
	Similarly,  by Lemma \ref{individual}, $\Delta_{L3}^{(b)} = O_{P}(m_{p}^2/p^{5ca_{2}-3c-1/2})$ because $\|\widetilde{\Phi}^{j}-\Phi^{j}\|_{F} \leq \sqrt{p_{j}r_{j}}\|\widetilde{\Phi}^{j}-\Phi^{j}\|_{\max} =O(m_{p}/p^{c(3a_{2}-2)})$ and $\|\bOmega'\widetilde{\bPhi}\|_{F}^{2} \leq \|\widetilde{\bPhi}-\bPhi\|_{F}^{2} = \sum_{j=1}^{G}\|\widetilde{\Phi}^{j}-\Phi^{j}\|_{F}^{2} = O(m_{p}^2/p^{3c(2a_{2}-1)-1})$. Similarly, we can show  $\Delta_{L1}$, $\Delta_{L4}$, $\Delta_{L5}$ and $\Delta_{L6}$  are dominated by $\Delta_{L2}$ and $\Delta_{L3}$.
	Therefore, we have
	\begin{align}
	\Delta_{L} &= O_{P}\Big(p^{\frac{5}{2}(1-a_{1})+c(1-a_{2})}\sqrt{\frac{\log p}{T}} + \frac{1}{p^{\frac{5}{2}a_{1}-\frac{3}{2}-2c+ca_{2}}} + \frac{m_{p}}{p^{ca_{2}}}  \nonumber  \\
	&\qquad\qquad + p^{\frac{11}{2}-5a_{1}+5c(1-a_{2})}\frac{\log p}{T}+ 
	\frac{1}{p^{5a_{1}-\frac{7}{2}-c(7-5a_{2})}} +\frac{m_{p}^2}{p^{5ca_{2}-3c-\frac{1}{2}}}\Big). \label{delta_L}
	\end{align}
	Combining the terms $\Delta_{G}$, $\Delta_{L}$ and $\Delta_{S}$ together, we complete the proof of (\ref{entropynorm}). 
	$\square$

	\subsection{Proof for Section \ref{regularPOET}}
	We consider (\ref{errorrate_POET}).  
	Similar to the proof  of Theorem 2.1 in \cite{fan2011high}, we can obtain  $\|\widehat{\bSigma}_{E}^{\mathcal{T}}- \bSigma_{E}\|_{2} = O_{P}(\mu_{p}\widetilde{\omega}_{T}^{1-q})$ from (\ref{sigmaE_max}).

	Consider (\ref{entropy_POET}). 
	We have
	\begin{align*}
	\|\widehat{\bSigma}^{\mathcal{T}}-\bSigma\|_{\Sigma} &\leq p^{-1/2}\Big(\|\bSigma^{-1/2}(\widehat{\bV}\widehat{\bGamma}\widehat{\bV}' - \bB\bB')\bSigma^{-1/2}\|_{F} + \|\bSigma^{-1/2}(\widehat{\bSigma}^{\mathcal{T}}_{E}-\bSigma_{E})\bSigma^{-1/2}\|_{F}\Big)\\
	&=: \Delta_{G} + \Delta_{E}
	\end{align*}
	and 
	$$
	\Delta_{E} = O_{P}(p^{-1/2}\|\widehat{\bSigma}_{E}^{\mathcal{T}}- \bSigma_{E}\|_{F}) = O_{P}(\|\widehat{\bSigma}_{E}^{\mathcal{T}}- \bSigma_{E}\|_{2}) =O_{P}(\mu_{p}\widetilde{\omega}_{T}^{1-q}).  
	$$
	Combining the terms $\Delta_{E}$ and  $\Delta_{G}$ from (\ref{delta_G}), we have
	$$
	\|\widehat{\bSigma}^{\mathcal{T}}-\bSigma\|_{\Sigma} = O_{P}\left(\mu_{p}\widetilde{\omega}_{T}^{1-q} +p^{\frac{11}{2}-5a_{1}}\frac{\log p}{T} + \frac{1}{p^{5a_{1}-\frac{7}{2}-2c}}\right).
	$$

\subsection{Proofs for Section \ref{orthogonality}}
Under the orthogonality condition, we have the following modified rates of convergence.
	Define $\bW = [\bB \quad \bLambda]$, which is a $p \times (k+r)$ loading matrix.
	Let $\{\bar{\delta}_i, \bar{v}_i\}_{i=1}^{k+r}$ be the leading eigenvalues and eigenvectors of $\bW\bW'$ in the decreasing order  and the rest zero.

\begin{lem}
	Under the assumptions of Theorem \ref{DPOET with orthogonality}, we have the following results.
	\begin{itemize}\label{global_lemma2}
		\item[(i)] By Weyl's theorem, we have
		$$
		|\delta_{i} - \bar{\delta}_{i}| \leq \|\bSigma_{u}\| \text{ for } i \leq k+r, \;\;\;\;\; |\delta_{i}| \leq \|\bSigma_{u}\|  \text{ for } i > k+r,	
		$$
		$|\delta_{i} - \bar{\delta}_{i}| \leq \|\bSigma_{u}\|$ for $i \leq k+r$ and $|\delta_{i}| \leq \|\bSigma_{u}\|$  for $i > k+r$.
		In addition, for all $p$,
			$\bar{\delta}_{i}/p^{a_{1}}$ and $\bar{\delta}_{i}/p^{ca_{2}}$ are strictly bigger than zero for $i \leq k$ and $ k < i \leq k+r$, respectively.
		\item[(ii)] We have		
		\begin{align*}
		&\max_{i\leq k} \|\bar{v}_{i}-v_{i}\|_{\infty} = O\left(\frac{m_{p}}{p^{3(a_{1} -\frac{1}{2})}}\right),\\
		&\max_{i\leq k} \|\widehat{v}_{i}-v_{i}\|_{\infty} = O_{P}\left(p^{3(1-a_{1})}\sqrt{\frac{\log p}{Tp}}+\cfrac{m_{p}}{p^{3(a_{1}-\frac{1}{2})}}\right).
		\end{align*}
	\end{itemize}	
\end{lem}
\begin{proof}
	(ii) Under the orthogonality condition, we have
	$$
	\widehat{\bSigma} = \bB\bB' + \bLambda\bLambda' + \bSigma_{u} + (\widehat{\bSigma}-\bSigma) = \bW\bW' + \bSigma_{u} + (\widehat{\bSigma}-\bSigma).
	$$
	We can treat $\bW\bW'$ as a low rank matrix and the remaining terms as a perturbation matrix.
	Then, similar to Lemmas \ref{individual} and \ref{global_pro}, we can show 
	\begin{align*}
		\max_{i\leq k} \|\bar{v}_{i}-v_{i}\|_{\infty} &\leq C p^{2(1-a_1)} \frac{\|\bSigma_{u}\|_{\infty}}{\bar{\gamma} \sqrt{p}},\\		
		\max_{i\leq k} \|\widehat{v}_{i}-v_{i}\|_{\infty}  &\leq C p^{2(1-a_1)} \frac{\| \bSigma_{u} + (\widehat{\bSigma}-\bSigma)\|_{\infty}}{\bar{\gamma} \sqrt{p}}
		= O_{P}\left(\cfrac{\|\bSigma_{u}\|_{\infty}}{p^{3(a_{1}-\frac{1}{2}) } }+p^{3(1-a_{1})}\sqrt{\frac{\log p}{Tp}}\right),
	\end{align*}		
	where the eigengap $\bar{\gamma} = \min\{\bar{\delta}_{i} - \bar{\delta}_{i+1} : 1 \leq i \leq k\}$.		
\end{proof}

	\begin{lem}\label{local_pro2}
	Under the assumptions of Theorem \ref{DPOET with orthogonality}, for $i\leq r_j$, we have 
	\begin{align*}
		&|\widehat{\kappa}_{i}^{j}/\kappa_{i}^{j} -1 | = O_{P} \left (p^{\frac{5}{2}(1-a_{1})+c(1-a_{2})}\sqrt{\frac{\log p}{T}} + \frac{m_{p}}{p^{\frac{5a_{1}}{2}-\frac{3}{2}-c(1-a_{2})}}\right),\\
		&\|\widehat{\eta}_{i}^{j}-\eta_{i}^{j}\|_{\infty} = O_{P}\left (p^{\frac{5}{2}(1-a_{1})+ c(\frac{5}{2}-3a_{2})}\sqrt{\frac{\log p}{T}}  +\cfrac{m_{p}}{p^{\frac{5 a_{1}}{2} -\frac{3}{2} +c(3a_2-\frac{5}{2})}}  +\cfrac{m_{p}}{p^{3c(a_{2}-\frac{1}{2})}}\right).
	\end{align*}
\end{lem}
\begin{proof}
	Similar to the proof of Lemma \ref{local_pro} using results from Lemma \ref{global_lemma2}, we can obtain
	\begin{equation}
		\|\widehat{\bV}\widehat{\bGamma}\widehat{\bV}' - \bB\bB'\|_{\max} = O_{P}(p^{\frac{5}{2}(1-a_{1})}\sqrt{\log p/T} + m_{p}/p^{\frac{5a_{1}}{2}-\frac{3}{2}}). \label{global_ind2}
	\end{equation}		
	Then, we have 
	\begin{eqnarray}
		\|\widehat{\bSigma}_{E}-\bSigma_{E}\|_{\max} &\leq& \|\widehat{\bSigma}-\bSigma\|_{\max} + \|\widehat{\bV}\widehat{\bGamma}\widehat{\bV}' - \bB\bB'\|_{\max}\cr
		& =& O_{P}(p^{\frac{5}{2}(1-a_{1})}\sqrt{\log p/T} + m_{p}/p^{\frac{5a_{1}}{2}-\frac{3}{2}}). \label{sigmaE_max2}
	\end{eqnarray}
	Therefore, the first statement is followed by \eqref{sigmaE_max2} and the Weyl's theorem.
	In addition, by Theorem 1 of \cite{fan2018eigenvector} and (\ref{sigmaE_max2}), we have
	\begin{align*}
		\|\widehat{\eta}_{i}^{j}-\eta_{i}^{j}\|_{\infty} &\leq C p_j ^{2 (1-a_2)}\frac{\|\bSigma_{u}^{j} + (\widehat{\bSigma}_{E}^{j}-\bSigma_{E}^{j})\|_{\infty}}{p_j^{a_{2}}\sqrt{p_j}}\\
		&=O_{P}\left(p^{\frac{5}{2}(1-a_{1})+ c(\frac{5}{2}-3a_{2})}\sqrt{\frac{\log p}{T}}  +\cfrac{m_{p}}{p^{\frac{5 a_{1}}{2} -\frac{3}{2} +c(3a_2-\frac{5}{2})}}  +\cfrac{m_{p}}{p^{3c(a_{2}-\frac{1}{2})}}\right).
	\end{align*}
\end{proof}

	\textbf{Proof of Theorem \ref{DPOET with orthogonality}.}
Using the similar proof of (\ref{local_ind}) and results from Lemmas \ref{global_lemma2} and \ref{local_pro2}, we can obtain
\begin{equation}
	\|\widehat{\bPhi}\widehat{\bPsi}\widehat{\bPhi}' - \bLambda\bLambda'\|_{\max} =  O_{P}(p^{\frac{5}{2}(1-a_{1})+\frac{5}{2}c(1-a_{2})}\sqrt{\log p/T}+m_{p}/p^{\frac{5}{2}a_{1}-\frac{3}{2}-\frac{5}{2}c(1-a_{2})} +m_{p}/\sqrt{p^{c(5a_{2}-3)}}).\label{local_ind2}
\end{equation}
By Assumption \ref{assum1}(iii), (\ref{global_ind2}) and (\ref{local_ind2}), we have 
$$
\|\widehat{\bSigma}_{u}-\bSigma_{u}\|_{\max} = O_{P}\left(p^{\frac{5}{2}(1-a_{1})+\frac{5}{2}c(1-a_{2})}\sqrt{\frac{\log p}{T}}+\frac{m_{p}}{\sqrt{p^{5a_{1}-3-5c(1-a_{2})}}} +\frac{m_{p}}{\sqrt{p^{c(5a_{2}-3)}}}\right).
$$
Therefore, $\|\widehat{\bSigma}^{\mathcal{D}}_{u}-\bSigma_{u}\|_{\max} = O_{P}(\tau + \omega_{T,2}) =O_{P}(\omega_{T,2})$, when $\tau$ is chosen as the same order of $\omega_{T,2} = p^{\frac{5}{2}(1-a_{1})+\frac{5}{2}c(1-a_{2})}\sqrt{\log p/T}+m_{p}/\sqrt{p^{5a_{1}-3-5c(1-a_{2})}} +m_{p}/\sqrt{p^{c(5a_{2}-3)}}$.
Similar to the proofs of Theorem \ref{thm1}, we can obtain (\ref{u_max2}) -- (\ref{inverse rate2}).

Consider (\ref{entropynorm2}).
First, we have
\begin{align*}
	\|\widehat{\bSigma}^{\mathcal{D}} -\bSigma\|_{\Sigma}& \leq p^{-1/2}\Big(\|\bSigma^{-1/2}(\widehat{\bV}\widehat{\bGamma}\widehat{\bV}' - \bB\bB')\bSigma^{-1/2}\|_{F}\\
	& \;\;\;\;\; + \|\bSigma^{-1/2}(\widehat{\bPhi}\widehat{\bPsi}\widehat{\bPhi}' - \bLambda \bLambda')\bSigma^{-1/2}\|_{F} + \|\bSigma^{-1/2}(\widehat{\bSigma}^{\mathcal{D}}_{u}-\bSigma_{u})\bSigma^{-1/2}\|_{F}\Big)\\
	&=: \Delta_{G^{\prime}} + \Delta_{L^{\prime}} + \Delta_{S^{\prime}}
\end{align*}
and
$$
\Delta_{S^{\prime}} =O_{P}(p^{-1/2}\|\widehat{\bSigma}^{\mathcal{D}}_{u}-\bSigma_{u}\|_{F}) = O_{P}(\|\widehat{\bSigma}^{\mathcal{D}}_{u}-\bSigma_{u}\|_{2}) = O_{P}(m_{p}\omega_{T,2}^{1-q}).
$$
By Lemma \ref{global_lemma2}, we have 
\begin{align*}
	\|\bV'\widehat{\bV} - \bI\|_{F} &= \|\bV'(\widehat{\bV}-\bV)\|_{F} \leq \|\widehat{\bV}-\bV\|_{F}=O_{P}(p^{3(1-a_{1})}\sqrt{\log p/T}+m_{p}/p^{3a_{1}-2}),\\
	\|\bV'\widetilde{\bV}-\bI\|_{F} &= \|\bV'(\widetilde{\bV}-\bV)\|_{F} \leq \|\widetilde{\bV}-\bV\|_{F} = O(m_{p}/p^{3a_{1}-2}).
\end{align*}
Then, similar to the proof of (\ref{delta_G}), we can obtain
	\begin{equation}
	\Delta_{G^{\prime}} = O_{P}\left(p^{1-a_{1}}\sqrt{\frac{\log p}{T}} + \frac{m_{p}}{p^{a_{1}}} + p^{\frac{11}{2}-5a_{1}}\frac{\log p}{T} + \frac{m_{p}^2}{p^{5a_{1}-\frac{7}{2}}}\right). \label{delta_G2}
\end{equation}
By Lemma \ref{local_pro2}, we have 
\begin{align*}
	\|\widehat{\bPhi}-\bPhi\|_{F}^{2} &= \sum_{j=1}^{J}\|\widehat{\Phi}^{j}-\Phi^{j}\|_{F}^2\\
	&=O_{P}\left(p^{1-c}\Big(p^{5(1-a_{1})+6c(1-a_{2})}\frac{\log p}{T} +\frac{m_{p}^{2}}{p^{5a_{1}-3-6c(1-a_{2})}} +\frac{m_{p}^2}{p^{2c(3a_{2}-2)}}\Big)\right),\\
	\|\widetilde{\bPhi}-\bPhi\|_{F}^{2} &= \sum_{j=1}^{J}\|\widetilde{\Phi}^{j}-\Phi^{j}\|_{F}^2 = O(m_{p}^{2}/p^{2c(3a_{2}-2)-(1-c)}).
\end{align*}
Then, similar to the proof of (\ref{delta_L}), we can obtain
\begin{align*}
	\Delta_{L^{\prime}} &= O_{P}\Big(p^{\frac{5}{2}(1-a_{1})+c(1-a_{2})}\sqrt{\frac{\log p}{T}} + \frac{m_{p}}{p^{\frac{5}{2}a_{1}-\frac{3}{2}-c(1-a_{2})}} + \frac{m_{p}}{p^{ca_{2}}} \\
	&\qquad\qquad + p^{\frac{11}{2}-5a_{1}+5c(1-a_{2})}\frac{\log p}{T}+ 
	\frac{m_{p}^{2}}{p^{5a_{1}-\frac{7}{2}-5c(1-a_{2})}} +\frac{m_{p}^2}{p^{5ca_{2}-3c-\frac{1}{2}}}\Big).
\end{align*}
Combining the terms $\Delta_{G^{\prime}}$, $\Delta_{L^{\prime}}$ and $\Delta_{S^{\prime}}$ together, we complete the proof of (\ref{entropynorm2}).

\textbf{Proofs for (\ref{errorrate_POET2}) and (\ref{entropy_POET2}).}
We  consider  (\ref{errorrate_POET2}).  
Let $\bW =  (\tilde{w}_{1},\dots, \tilde{w}_{k+r})$. 
Then, $\bar{\delta}_{i} = \|\tilde{w}_{i}\|^{2} \asymp p$ for $i\leq k$, $\bar{\delta}_{i} = \|\tilde{w}_{i}\|^{2} \asymp p^c$ for $k < i\leq k+r$ by Lemma \ref{global_lemma2}(i), and $\bar{v}_{i} = \tilde{w}_{i}/\|\tilde{w}_{i}\|$.
Hence, $\|\bar{v}_{i}\|_{\infty} \leq \|\bW\|_{\max}/\|\tilde{w}_{i}\| \leq C/\sqrt{p^{c}}$ for $i\leq k+r$. 
In addition, for $\widetilde{\bV}_{2} = (\bar{v}_{1},\dots,\bar{v}_{k+r})$, the coherence $\mu(\widetilde{\bV}_{2}) = \frac{p}{k+r}\max_{i}\sum_{j=1}^{k+r}\widetilde{\bV}_{2,ij}^{2} \leq C p^{1-c}$, where $\widetilde{\bV}_{2,ij}$ is the $(i,j)$ entry of $\widetilde{\bV}_{2}$.
Thus, by Theorem 1 of \cite{fan2018eigenvector} and $r \asymp p^{1-c}$, we have
			\begin{align*}
			\max_{i\leq k+r} \|\bar{v}_{i}-v_{i}\|_{\infty} &\leq C (k+r)^4 p^{2(1-c)} \frac{\|\bSigma_{u}\|_{\infty}}{p^{c} \sqrt{p}} = 	O\left(\frac{m_{p}}{p^{7c-\frac{11}{2}}}\right),\\
			\max_{i\leq k+r} \|\widehat{v}_{i}-v_{i}\|_{\infty} & \leq C (k+r)^{4}p^{2(1-c)}\frac{\| \bSigma_{u} + (\widehat{\bSigma}-\bSigma)\|_{\infty}}{p^c\sqrt{p}}
			= O_{P}\left(\frac{m_{p}}{p^{7c-\frac{11}{2}}}+p^{\frac{13}{2}-7c}\sqrt{\frac{\log p}{T}}\right).			
			\end{align*}
By the similar argument of (\ref{global_ind}), we can obtain
$\|\widehat{\bV}_{2}\widehat{\bGamma}_{2}\widehat{\bV}_{2}' - \bW\bW'\|_{\max} = O_{P}(p^{7(1-c)}\sqrt{\log p/T} + m_{p}/p^{7c-6}).$ 
Hence, we have 
$$
\|\widehat{\bSigma}_{u,2}-\bSigma_{u}\|_{\max} \leq \|\widehat{\bSigma}-\bSigma\|_{\max} + \|\widehat{\bV}_{2}\widehat{\bGamma}_{2}\widehat{\bV}_{2}' - \bW\bW'\|_{\max} =  O_{P}(p^{7(1-c)}\sqrt{\log p/T} + m_{p}/p^{7c-6}).
$$
Similar to the proof  of (\ref{error rate}), we can obtain  $\|\widehat{\bSigma}_{u,2}^{\mathcal{T}}- \bSigma_{u}\|_{2} = O_{P}(m_{p}\widetilde{\omega}_{T,2}^{1-q})$, where $\widetilde{\omega}_{T,2} =p^{7(1-c)}\sqrt{\log p/T} + m_{p}/p^{7c-6}$.

Consider (\ref{entropy_POET2}). 
We have
\begin{align*}
	\|\widehat{\bSigma}_{2}^{\mathcal{T}}-\bSigma\|_{\Sigma} &\leq p^{-1/2}\Big(\|\bSigma^{-1/2}(\widehat{\bV}_{2}\widehat{\bGamma}_{2}\widehat{\bV}_{2}' - \bW\bW')\bSigma^{-1/2}\|_{F} + \|\bSigma^{-1/2}(\widehat{\bSigma}^{\mathcal{T}}_{u,2}-\bSigma_{u})\bSigma^{-1/2}\|_{F}\Big)\\
	&=: \Delta_{G_2} + \Delta_{S_2}
\end{align*}
and 
$$
\Delta_{S_2} = O_{P}(p^{-1/2}\|\widehat{\bSigma}_{u,2}^{\mathcal{T}}- \bSigma_{u}\|_{F}) = O_{P}(\|\widehat{\bSigma}_{u,2}^{\mathcal{T}}- \bSigma_{u}\|_{2}) =O_{P}(m_{p}\widetilde{\omega}_{T,2}^{1-q}).  
$$
Note that $\|\widehat{\bV}_{2}-\bV_{2}\|_{F} \leq\sqrt{p^{2-c}} \|\widehat{v}_{i}-v_{i}\|_{\infty} = O_{P}(m_{p}/\sqrt{p^{15c-13}} + p^{\frac{15}{2}(1-c)}\sqrt{\log p /T})$, where $\bV_{2} = (v_1,\dots, v_{k+r})$ is the leading eigenvectors of covariance matrix $\bSigma$.
Similar to the proof of (\ref{delta_G}), we can obtain	
$$
\Delta_{G_2} = O_{P}\left((p^{7-\frac{15}{2}c} + p^{1-c})\sqrt{\frac{\log p}{T}} + \frac{m_{p}}{p^{c}} + \frac{m_{p}}{p^{\frac{15}{2}c-6}}+ p^{15(1-c)+\frac{1}{2}}\frac{\log 	p}{T} + \frac{m_{p}^2}{p^{15c-\frac{27}{2}}}\right).
$$
Combining the terms $\Delta_{G_2}$ and  $\Delta_{S_2}$, we have
$$
\|\widehat{\bSigma}_{2}^{\mathcal{T}}-\bSigma\|_{\Sigma} = O_{P}\left(m_{p}\widetilde{\omega}_{T,2}^{1-q} +p^{15(1-c)+\frac{1}{2}}\frac{\log p}{T} + \frac{m_{p}^{2}}{p^{15c-\frac{27}{2}}}\right).
$$

\subsection{Proof of Theorem \ref{blessing thm}}
\textbf{Proof of Theorem \ref{blessing thm}.}
We  consider (\ref{local max}). 
By (\ref{maxnorm2}), for each group $j$, we have
			$$
			\|(\widehat{\bSigma}^{\mathcal{D}})^{j} - \bSigma^{j}\|_{\max} \leq \|\widehat{\bSigma}^{\mathcal{D}} - \bSigma\|_{\max} = O_{P}(\omega_{T,2}).
			$$

Consider (\ref{local inverse}). 
Since $\lambda_{\min}(\bSigma_{u}^{j})>c_1$ and $m_{p_j}\omega_{T,2}^{1-q}=o(1)$, the minimum eigenvalue of $(\widehat{\bSigma}_{u}^{\mathcal{D}})^{j}$ is strictly bigger than 0 with probability approaching 1.
Then,  we have $\|((\widehat{\bSigma}_{u}^{\mathcal{D}})^{j})^{-1} - (\bSigma_{u}^{j})^{-1}\|_{2} \leq \lambda_{\min}(\bSigma_{u}^{j})^{-1} \|(\widehat{\bSigma}_{u}^{\mathcal{D}})^{j}- \bSigma_{u}^{j}\|_{2} \lambda_{\min}((\widehat{\bSigma}_{u}^{\mathcal{D}})^{j})^{-1} = O_{P}(m_{p_{j}}\omega_{T,2}^{1-q})$.
Similar to the proofs of (\ref{inverse rate}), we can show the statement.

Consider (\ref{local entropy}).
We denote $(\bA)^{j}$ the $j$th diagonal block for a matrix $\bA$. 
We have $B^{j}B^{j\prime} =(\bB\bB')^{j} = (\widetilde{\bV}\widetilde{\bGamma}\widetilde{\bV}')^{j} = \widetilde{\bV}^{j}\widetilde{\bGamma}\widetilde{\bV}^{j\prime}$, where $\widetilde{\bV}^{j}$ is the $j$th $p_j \times k$ submatrix of $\widetilde{\bV}$. 
Similarlly, $(\widehat{\bV}\widehat{\bGamma}\widehat{\bV}')^{j}= \widehat{\bV}^{j}\widehat{\bGamma}\widehat{\bV}^{j\prime}$, where $\widehat{\bV}^{j}$ is the $j$th $p_j \times k$ submatrix of $\widehat{\bV}$. 
For each group $j$, we have
$$
(\widehat{\bSigma}^{\mathcal{D}})^{j} = \widehat{\bV}^{j}\widehat{\bGamma}\widehat{\bV}^{j\prime} + \widehat{\Phi}^{j}\widehat{\Psi}^{j}\widehat{\Phi}^{j \prime}  + (\widehat{\bSigma}_{u}^{\mathcal{D}})^{j}.
$$
By SVD, we have $\bSigma^{j} = \bV_{p_j}\bGamma_{p_j}\bV_{p_j}'$, where $\bV_{p_j} =(\bV^{j}, \Phi^{j}, \Omega^{j})$ and $\bGamma_{p_j} = \diag(\bGamma, \Psi^{j}, \Theta^{j})$. 
Then, we have
			\begin{align*}
			\|(\widehat{\bSigma}^{\mathcal{D}})^{j} - \bSigma^{j}\|_{\Sigma^{j}}  &\leq (p^{j})^{-1/2}\Big(\left\|(\bSigma^{j})^{-\frac{1}{2}}\left(\widehat{\bV}^{j}\widehat{\bGamma}\widehat{\bV}^{j\prime}-B^{j}B^{j\prime}\right)(\bSigma^{j})^{-\frac{1}{2}}\right\|_{F} \\
			&+\left\|(\bSigma^{j})^{-\frac{1}{2}}(\widehat{\Phi}^{j}\widehat{\Psi}^{j}\widehat{\Phi}^{j \prime}-\Lambda^{j}\Lambda^{j\prime})(\bSigma^{j})^{-\frac{1}{2}}\right\|_{F} +\left\|(\bSigma^{j})^{-\frac{1}{2}}\left((\widehat{\bSigma}_{u}^{\mathcal{D}})^{j} -\bSigma_u^{j}\right)(\bSigma^{j})^{-\frac{1}{2}}\right\|_{F}\Big)\\
			&=: \Delta_{g} + \Delta_{l}+\Delta_{s}.
			\end{align*}
Similar to the proofs of Theorem 2.1 in \cite{fan2011high}, we can show $\|(\widehat{\bSigma}_{u}^{\mathcal{D}})^{j} -\bSigma_u^{j}\|_{2} = O_{P}(m_{p_j}\omega_{T,2}^{1-q}).$
Then, we have
			$$
			\Delta_{s} = O_{P}(p_{j}^{-1/2}\|(\widehat{\bSigma}_{u}^{\mathcal{D}})^{j} -\bSigma_u^{j}\|_{F}) = O_{P}(\|(\widehat{\bSigma}_{u}^{\mathcal{D}})^{j} -\bSigma_u^{j}\|_{2}) = O_{P}(m_{p_j}\omega_{T,2}^{1-q}).
			$$
We have
			\begin{align*}
			&\Delta_{g} 
			 \leq p_{j}^{-1/2}\Big(\left\|\bGamma^{-1/2}\bV^{j\prime}\left(\widehat{\bV}^{j}\widehat{\bGamma}\widehat{\bV}^{j\prime}-B^{j}B^{j\prime}\right)\bV^{j}\bGamma^{-1/2}\right\|_{F}\\
			&+\left\|{\Psi^{j}}^{-1/2}\Phi^{j\prime}\left(\widehat{\bV}^{j}\widehat{\bGamma}\widehat{\bV}^{j\prime}-B^{j}B^{j\prime}\right)\Phi^{j}{\Psi^j}^{-1/2}\right\|_{F} + \left\|{\Theta^{j}}^{-1/2}\Omega^{j\prime}\left(\widehat{\bV}^{j}\widehat{\bGamma}\widehat{\bV}^{j\prime}-B^{j}B^{j\prime}\right)\Omega^{j}{\Theta^{j}}^{-1/2}\right\|_{F} \\ &+2\left\|\bGamma^{-1/2}\bV^{j\prime}\left(\widehat{\bV}^{j}\widehat{\bGamma}\widehat{\bV}^{j\prime}-B^{j}B^{j\prime}\right)\Phi^{j}{\Psi^j}^{-1/2}\right\|_{F} + 2\left\|\bGamma^{-1/2}\bV^{j\prime}\left(\widehat{\bV}^{j}\widehat{\bGamma}\widehat{\bV}^{j\prime}-B^{j}B^{j\prime}\right)\Omega^{j}{\Theta^{j}}^{-1/2}\right\|_{F}\\ &+2\left\|{\Psi^{j}}^{-1/2}\Phi^{j\prime}\left(\widehat{\bV}^{j}\widehat{\bGamma}\widehat{\bV}^{j\prime}-B^{j}B^{j\prime}\right)\Omega^{j}{\Theta^{j}}^{-1/2}\right\|_{F}\Big)\\
			& =: \Delta_{g1}+\Delta_{g2}+\Delta_{g3}+2\Delta_{g4}+2\Delta_{g5}+2\Delta_{g6}.
			\end{align*}
Using the similar proof of (\ref{delta_G1}) and results from Lemma \ref{global_lemma2}, we can obtain 
			$$
			\Delta_{g1} = O_{P}\left(p^{1-a_{1}}\sqrt{\frac{\log p}{T}} + \frac{m_{p}}{p^{a_{1}}}\right).
			$$ 
For $\Delta_{g3}$, we have
			$$
			p_{j}^{-1/2}\Big(\left\|{\Theta^{j}}^{-1/2}\Omega^{j\prime}\widehat{\bV}^{j}\widehat{\bGamma}\widehat{\bV}^{j\prime}\Omega^{j}{\Theta^{j}}^{-1/2}\right\|_{F}+\left\|{\Theta^{j}}^{-1/2}\Omega^{j\prime}\widetilde{\bV}^{j}\widetilde{\bGamma}\widetilde{\bV}^{j\prime}\Omega^{j}{\Theta^{j}}^{-1/2}\right\|_{F}\Big) = : \Delta_{g3}^{(a)} + \Delta_{g3}^{(b)}.
			$$
By Lemma \ref{global_lemma2}, we have 		
			$$
			\|\widehat{\bV}^{j}-\bV^{j}\|_{F} = O_{P}(\sqrt{p^{c}}\|\widehat{\bV}-\bV\|_{\max}) = 	O_{P}\left(\sqrt{p^{c}}(p^{3(1-a_{1})}\sqrt{\log p/Tp}+m_{p}/p^{3(a_{1}-\frac{1}{2})})\right).
			$$ 
Since $\|\widehat{\bGamma}\| = O_{P}(p^{a_{1}})$, we then have
			$$
			\Delta_{g3}^{(a)} \leq 
			p_j^{-1/2}\|{\Theta^{j}}^{-1}\|\|\Omega^{j\prime}(\widehat{\bV}^{j}-\bV^{j})\|_{F}^{2}\|\widehat{\bGamma}\| = O_{P}(p^{5(1-a_{1})+c/2}\log p/T  + m_{p}^{2}/p^{5a_{1}-3-c/2}).
			$$
Similarly, since $\|\widetilde{\bV}^{j}-\bV^{j}\|_{F} = O_{P}(\sqrt{p^{c}}\|\widetilde{\bV}-\bV\|_{\max}) = O_{P}(m_{p}/p^{3(a_{1}-1/2)-c/2})$ by Lemma \ref{global_lemma2}, $\Delta_{g3}^{(b)} =  O_{P}(m_{p}^{2}/p^{5a_{1}-3-c/2})$. 
Then, we obtain 
			$$
			\Delta_{g3} = O_{P}\left(p^{5(1-a_{1})+\frac{c}{2}}\frac{\log p}{T}  + \frac{m_{p}^{2}}{p^{5a_{1}-3-\frac{c}{2}}}\right).
			$$ 
Similarly, we can show that the terms $\Delta_{g2}$, $\Delta_{g4}$, $\Delta_{g5}$ and $\Delta_{g6}$ are dominated by $\Delta_{g1}$ and $\Delta_{g3}$. 
Thus, we have
			$$
			\Delta_{g} = O_{P}\left(p^{1-a_{1}}\sqrt{\frac{\log p}{T}} + \frac{m_{p}}{p^{a_{1}}} +p^{5(1-a_{1})+\frac{c}{2}}\frac{\log p}{T}  + \frac{m_{p}^{2}}{p^{5a_{1}-3-\frac{c}{2}}}\right).
			$$

We have
		\begin{align*}
			\Delta_{l} 
			 & \leq  p_{j}^{-1/2}\Big(\left\|\bGamma^{-1/2}\bV^{j\prime}(\widehat{\Phi}^{j}\widehat{\Psi}^{j}\widehat{\Phi}^{j \prime}-\Lambda^{j}\Lambda^{j\prime})\bV^{j}\bGamma^{-1/2}\right\|_{F}\\ 
			 &\;\;\;+\left\|{\Psi^{j}}^{-1/2}\Phi^{j\prime}(\widehat{\Phi}^{j}\widehat{\Psi}^{j}\widehat{\Phi}^{j \prime}-\Lambda^{j}\Lambda^{j\prime})\Phi^{j}{\Psi^j}^{-1/2}\right\|_{F}+ \left\|{\Theta^{j}}^{-1/2}\Omega^{j\prime}(\widehat{\Phi}^{j}\widehat{\Psi}^{j}\widehat{\Phi}^{j \prime}-\Lambda^{j}\Lambda^{j\prime})\Omega^{j}{\Theta^{j}}^{-1/2}\right\|_{F}\\ 
			 &\;\;\;+2\left\|\bGamma^{-1/2}\bV^{j\prime}(\widehat{\Phi}^{j}\widehat{\Psi}^{j}\widehat{\Phi}^{j \prime}-\Lambda^{j}\Lambda^{j\prime})\Phi^{j}{\Psi^j}^{-1/2}\right\|_{F}+ 2\left\|\bGamma^{-1/2}\bV^{j\prime}(\widehat{\Phi}^{j}\widehat{\Psi}^{j}\widehat{\Phi}^{j \prime}-\Lambda^{j}\Lambda^{j\prime})\Omega^{j}{\Theta^{j}}^{-1/2}\right\|_{F}\\ &\;\;\;+2\left\|{\Psi^{j}}^{-1/2}\Phi^{j\prime}(\widehat{\Phi}^{j}\widehat{\Psi}^{j}\widehat{\Phi}^{j \prime}-\Lambda^{j}\Lambda^{j\prime})\Omega^{j}{\Theta^{j}}^{-1/2}\right\|_{F}\Big)\\
			& =: \Delta_{l1}+\Delta_{l2}+\Delta_{l3}+2\Delta_{l4}+2\Delta_{l5}+2\Delta_{l6}.
		\end{align*}
For $\Delta_{l2}$, similar to the proof of (\ref{delta_L2}) using results from Lemma \ref{local_pro2}, we can obtain 
		$$
		\Delta_{l2} = O_{P}\left(p^{\frac{5}{2}(1-a_{1})+c(1-a_{2})}\sqrt{\frac{\log p}{T}} + \frac{m_{p}}{p^{\frac{5}{2}a_{1}-\frac{3}{2}-c(1-a_{2})}} + \frac{m_{p}}{p^{ca_{2}}}\right).
		$$ 
For $\Delta_{l3}$, we have
		$$
		\Delta_{l3}\leq p_{j}^{-1/2}\Big(\left\|{\Theta^{j}}^{-1/2}\Omega^{j\prime}\widehat{\Phi}^{j}\widehat{\Psi}^{j}\widehat{\Phi}^{j \prime}\Omega^{j}{\Theta^{j}}^{-1/2}\right\|_{F}+\left\|{\Theta^{j}}^{-1/2}\Omega^{j\prime} \widetilde{\Phi}^{j}\widetilde{\Psi}^{j}\widetilde{\Phi}^{j\prime}\Omega^{j}{\Theta^{j}}^{-1/2}\right\|_{F}\Big) = : \Delta_{l3}^{(a)} + \Delta_{l3}^{(b)}.
		$$
By Lemma \ref{local_pro2}, we have $\|\widehat{\Phi}^{j}-\Phi^{j}\|_{F} \leq \sqrt{p_{j}r_{j}}\|\widehat{\Phi}^{j}-\Phi^{j}\|_{\max} = O_{P}\big(p^{\frac{5}{2}(1-a_{1}) + 3c(1-a_{2})}\sqrt{\log p/T} +m_{p}/p^{\frac{5a_{1}}{2}-\frac{3}{2}-3c(1-a_{2})} +m_{p}/p^{c(3a_{2}-2)}\big)$ and $\|\widehat{\Psi}^{j}\| = O_{P}(p^{ca_{2}})$,	we have
		\begin{align*}
		\Delta_{l3}^{(a)} &\leq p_j^{-1/2}\|{\Theta^{j}}^{-1}\|\|\Omega^{j\prime}(\widehat{\Phi}^{j}-\Phi^{j})\|_{F}^{2}\|\widehat{\Psi}^{j}\|\\
		&= O_{P}\left(p^{5(1-a_{1})+c(\frac{11}{2}-5a_{2})}\frac{\log p}{T} + \frac{m_{p}^{2}}{p^{5a_{1}-3+c(5a_{2}-\frac{11}{2})}}  + \frac{m_{p}^2}{p^{c(5a_{2}-\frac{7}{2})}} \right).
		\end{align*}
Similarly, by Lemma \ref{individual}, $\|\widetilde{\Phi}^{j}-\Phi^{j}\|_{F} = O_{P}(\sqrt{p^{c}}\|\widetilde{\Phi}^{j}-\Phi^{j}\|_{\max}) = O_{P}(m_{p}/p^{c(3a_{2}-2)})$, then  $\Delta_{l3}^{(b)} =  O_{P}(m_{p}^{2}/p^{c(5a_{2}-\frac{7}{2})})$. 
Similarly, we can show $\Delta_{l1}$, $\Delta_{l4}$, $\Delta_{l5}$ and $\Delta_{l6}$ are dominated by $\Delta_{l2}$ and $\Delta_{l3}$. 
Thus, we have
		\begin{align*}
		\Delta_{l} &= O_{P}\Big(p^{\frac{5}{2}(1-a_{1})+c(1-a_{2})}\sqrt{\frac{\log p}{T}} + \frac{m_{p}}{p^{\frac{5}{2}a_{1}-\frac{3}{2}-c(1-a_{2})}} + \frac{m_{p}}{p^{ca_{2}}}\\
		& \qquad\qquad + p^{5(1-a_{1})+c(\frac{11}{2}-5a_{2})}\frac{\log p}{T} + \frac{m_{p}^{2}}{p^{5a_{1}-3+c(5a_{2}-\frac{11}{2})}}  + \frac{m_{p}^2}{p^{c(5a_{2}-\frac{7}{2})}}\Big).
		\end{align*}
Combining the terms $\Delta_{g}$, $\Delta_{l}$ and $\Delta_{s}$ together, we complete the proof of (\ref{local entropy}).
$\square$

\subsection{Proof of Theorem \ref{num of global factors}}

\begin{lem} \label{Ahn's lemma}
	Under the assumptions of Theorem \ref{num of global factors}, for $i = 1,\dots, [dm]-2(k+r)$, there exist constants $\underline{c}, \overline{c} > 0$ such that
	$$
	\underline{c}+o_{P}(1) \leq m\widehat{\delta}_{k+r+i}/p \leq \overline{c}+o_{P}(1).
	$$
\end{lem}
\begin{proof}
	Similar to the proofs of Lemma A.9 in \cite{ahn2013eigenvalue}, we can show the statement. 
\end{proof}

\begin{lem} \label{lemma for num}
	Under the multi-level factor model and assumptions of Theorem \ref{num of global factors}, for $i = 1, \dots, k$, we have
	$$
	\widehat{\delta}_{i} = \bar{\delta}_{i} + O_{P}\left(p\sqrt{\frac{\log p}{T}}\right) + O_{P}\left(p^{ca_{2}}\right). 
	$$
	In addition, for $i = k+1, \dots, k+r$, we have
	$$
	\widehat{\delta}_{i} = \bar{\delta}_{i} + O_{P}\left(p^{\frac{5}{2}(1-a_{1})+c}\sqrt{\frac{\log p}{T}} + \frac{1}{p^{\frac{5}{2}a_{1}-\frac{3}{2}-2c}}\right) + O_{P}\left(m_{p}\right). 
	$$
\end{lem}
\begin{proof}
	By Lemmas \ref{weyl's} and \ref{global_pro}, the first statement is showed. The order of $|\widehat{\delta}_{i} - \bar{\delta}_{i}|$ for $i = k+1, \dots, k+r$ is the same as that of $|\widehat{\kappa}_{i}^{j} - \bar{\kappa}_{i}^{j}|$ for $i =1,\dots,r_j$ and each group $j$. Then the second statement is showed from Lemmas \ref{weyl's} and \ref{local_pro}.
\end{proof}

\textbf{Proof of Theorem \ref{num of global factors}.}
(i) When $k>0$ and $r>0$, by Lemmas  \ref{weyl's} and \ref{lemma for num},  $\widehat{\delta}_{i}/\widehat{\delta}_{i+1} = \bar{\delta}_{i}/\bar{\delta}_{i+1} + o_{P}(1) = O_{P}(1)$ for $i = 1,2, \dots, k-1, k+1 \dots, k+r-1$. 
In addition, we have
$$
\frac{\widehat{\delta}_{k}}{\widehat{\delta}_{k+1}} = \frac{\bar{\delta}_{k} + O_{P}\left(p\sqrt{\log 	p/T}\right) + O_{P}\left(p^{ca_{2}}\right)}{\bar{\delta}_{k+1} + O_{P}\left(p^{\frac{5}{2}(1-a_{1})+c}\sqrt{\log p/T} + 1/p^{\frac{5}{2}a_{1}-\frac{3}{2}-2c}\right) + O_{P}\left(m_{p}\right)} = O_{P}\left(p^{a_{1}-ca_{2}}\right).
$$		
By Lemmas \ref{Ahn's lemma} and \ref{lemma for num},  we have
$$
\frac{\widehat{\delta}_{k+r}}{\widehat{\delta}_{k+r+1}} \geq \frac{\bar{\delta}_{k+r} + 	O_{P}\left(p^{\frac{5}{2}(1-a_{1})+c}\sqrt{\log p/T} + 1/p^{\frac{5}{2}a_{1}-\frac{3}{2}-2c}\right) + O_{P}\left(m_{p}\right)}{p[\overline{c}+o_{P}(1)]/m} = O_{P}\left(\frac{m}{p^{1-ca_{2}}}\right),
$$
which diverges to $\infty$ if $p^{1-ca_{2}} \ll T$. 
By Lemma \ref{Ahn's lemma}, for $i = 1, \dots, [dm] - 2(k+r) -1$, $\widehat{\delta}_{k+r+i}/\widehat{\delta}_{k+r+i+1} \leq (\overline{c}+o_{P}(1))/(\underline{c}+o_{P}(1))$. 

When $k=0$ and $r>0$, for $i\leq r$, we have $|\bar{\delta}_{i} - \delta_{i}| \leq \|\bSigma_{u}^{j}\| =O(m_{p})$ and $|\widehat{\delta}_{i}-\delta_{i}| \leq p^{c}\|\widehat{\bSigma}^{j}-\bSigma^{j}\|_{\max}= O_{P}(p^c\sqrt{\log p/T})$.
Then, $\widehat{\delta}_{i}/\widehat{\delta}_{i+1} = \bar{\delta}_{i}/\bar{\delta}_{i+1} + o_{P}(1) = O_{P}(1)$ for $i = 1,2, \dots, r-1$. 
In addition, by Lemmas \ref{Ahn's lemma},  we have
$$
\frac{\widehat{\delta}_{r}}{\widehat{\delta}_{r+1}} \geq \frac{\bar{\delta}_{r} + 	O_{P}\left(p^c\sqrt{\log p/T}\right) + O_{P}\left(m_{p}\right)}{p[\overline{c}+o_{P}(1)]/m} = O_{P}\left(\frac{m}{p^{1-ca_{2}}}\right),
$$
which diverges to $\infty$ if $p^{1-ca_{2}} \ll T$. 
By Lemma \ref{Ahn's lemma}, for $i = 1, \dots, [dm] - 2r -1$, $\widehat{\delta}_{r+i}/\widehat{\delta}_{r+i+1} \leq (\overline{c}+o_{P}(1))/(\underline{c}+o_{P}(1))$. 

When $k>0$ and $r=0$, for $i\leq k$, we have $|\bar{\delta}_{i} - \delta_{i}| \leq \|\bSigma_{u}\| =O(m_{p})$ and $|\widehat{\delta}_{i}-\delta_{i}| = O_{P}(p\sqrt{\log p/T})$.
Then, $\widehat{\delta}_{i}/\widehat{\delta}_{i+1} = O_{P}(1)$ for $i = 1,2, \dots, k-1$. 
In addition, we have
$$
\frac{\widehat{\delta}_{k}}{\widehat{\delta}_{k+1}} \geq \frac{\bar{\delta}_{k} + O_{P}\left(p\sqrt{\log 	p/T}\right) + O_{P}\left(m_{p}\right)}{p[\overline{c}+o_{P}(1)]/m} = O_{P}\left(\frac{m}{p^{1-a_{1}}}\right),
$$
which diverges to $\infty$ if $p^{1-a_{1}} \ll T$. 
By Lemma \ref{Ahn's lemma}, for $i = 1, \dots, [dm] - 2k -1$, $\widehat{\delta}_{k+i}/\widehat{\delta}_{k+i+1} \leq (\overline{c}+o_{P}(1))/(\underline{c}+o_{P}(1))$. 
These results show the consistency of model selection.

(ii) 
Under the multi-level factor model (i.e., $k>0$ and $r>0$), the above results imply that $\hat{k}$ is consistent by using the fact that $k < k+r$.
$\square$

\end{document}